\documentclass[english, 12pt]{article}
\usepackage{enumerate}
\usepackage{amsopn,amsfonts,amsmath,amssymb,mathrsfs}
\usepackage{amsthm} 
\usepackage{graphicx}
\usepackage{verbatim}
\usepackage{booktabs}
\usepackage[english]{babel}
\usepackage[]{float}
\usepackage[]{placeins}
\usepackage{flafter}
\usepackage[]{longtable}
\usepackage{multirow}
\usepackage{verbatim}
\usepackage{dsfont}
\usepackage{array}	
\usepackage{natbib}
\usepackage{tabularx}
\usepackage{xr}
\usepackage{color}
\usepackage[usenames,dvipsnames,svgnames,table]{xcolor}
\usepackage[colorlinks=true,	
            linkcolor=red,
            urlcolor=blue,
            citecolor=gray]{hyperref}
\usepackage{layout}
\usepackage{mathtools} 
\usepackage{bm} 
\usepackage{tocbibind} 
\usepackage[flushleft, para]{threeparttable}
\usepackage{tikz}
\usepackage{pgf}
\usepackage{subfig}
\usepackage{ifthen}

\usepackage{setspace} 

\usepackage{lmodern} 
\DeclareTextCommandDefault{\mathdefault}[1]{#1}

\def\indic{{\rm 1\hspace{-2.3pt}}}

\newcommand\indep{\protect\mathpalette{\protect\independenT}{\perp}}
\def\independenT#1#2{\mathrel{\rlap{$#1#2$}\mkern2mu{#1#2}}}

\def\dim{{\rm dim}}
\def\argmin{{\rm argmin}}

\DeclareMathOperator*{\argmax}{arg\,max}

\newcommand{\V}{\text{Var}}
\newcommand{\Supp}{\mathcal{S}}

\theoremstyle{remark}

\theoremstyle{plain}
\newtheorem{remark}{Remark}
\newtheorem{theorem}{Theorem}

\newtheorem{lemma}{Lemma}
\newtheorem{corollary}{Corollary}

\newtheorem{manualtheoreminner}{Assumption}
\newenvironment{manualassumption}[1]{%
  \renewcommand\themanualtheoreminner{#1}%
  \manualtheoreminner
}{\endmanualtheoreminner}

\usepackage[pagewise]{lineno}
\usepackage{geometry}
\newcolumntype{C}[1]{>{\centering\let\newline\\\arraybackslash\hspace{0pt}}m{#1}}

\begin{document}
\title{\vspace{-3cm} 
Heterogeneity, Uncertainty and Learning: Semiparametric Identification and Estimation\thanks{\scriptsize{First version: November 2022. This paper has benefited from detailed comments and suggestions by Jim Heckman, Elena Pastorino, Bernard Salani\'e and Yuya Sasaki. We also thank Karun Adusumilli, Victor Aguirregabiria, Peter Arcidiacono, St\'{e}phane Bonhomme, Xavier D'Haultfoeuille, Yifan Gong, Yuichi Kitamura, Mauricio Olivares, Chris Taber and Daniel Wilhelm for useful comments and discussions. We thank seminar participants at Aarhus, CREST-PSE, LMU, TSE, UC Davis, UT Austin, participants at the $34^{\textrm{th}}$ $\textrm{EC}^2$ conference, the Workshop on Human Capital and Imperfect Information (TSE, May 2025), the 2024 NAWMES and Barcelona GSE Structural Micro conference, 2023 IAAE, SETA and SOLE meetings. We thank Zhangchi Ma and Chinncy Qin for capable research assistance.}} 
}

\author{Jackson Bunting \thanks{\scriptsize{University of Washington, \href{mailto:buntingj@uw.edu}{buntingj@uw.edu}.}} \and Paul Diegert\thanks{\scriptsize{Toulouse School of Economics, \href{mailto:paul.diegert@tse-fr.eu}{paul.diegert@tse-fr.eu}.}} \and Arnaud Maurel\thanks{\scriptsize{Duke University, NBER and IZA, \href{mailto:arnaud.maurel@duke.edu}{arnaud.maurel@duke.edu}.}}}
\maketitle      
~\vspace{-1.8cm}
\begin{abstract}
We provide identification results for a broad class of learning models in which continuous outcomes depend on three types of unobservables: known heterogeneity, initially unknown heterogeneity that may be revealed over time, and transitory uncertainty. We consider a common environment where the researcher only has access to a short panel on choices and realized outcomes. We establish identification of the outcome equation parameters and the distribution of the unobservables, under the standard assumption that unknown heterogeneity and uncertainty are normally distributed. We also show that, absent known heterogeneity, the model is identified without making any distributional assumption. We then derive the asymptotic properties of a sieve MLE estimator for the model parameters, and devise a tractable profile likelihood-based estimation procedure. Our estimator exhibits good finite-sample properties. Finally, we illustrate our approach with an application to ability learning in the context of occupational choice. Our results point to substantial ability learning based on realized wages.

\end{abstract}
\newpage
\section{Introduction}\label{sec1}

Learning models, in which agents have imperfect information about their environment and update their beliefs over time, are frequently used in economics. These models have received particular interest in various subfields in empirical microeconomics, including labor economics \citep[see, e.g.,][]{Miller84,AG12,Pastorino15, Hincapie20,Pastorino22}, economics of education \citep[see, e.g.,][]{Arcidiacono04,Zafar11,SS12,stange12,Thomas19,KP21,Proctor22,AAMR25}, industrial organization and health \citep[see, e.g.,][for a survey in the context of oligopoly competition]{ackerberg2003advertising,CS04,CS05,AC05,ChanHamilton06,AJ20}. Since the seminal work of \cite{EK96}, learning models have also been popular in the marketing literature \citep[see][for a survey]{CEK13}. However, while learning models are often estimated, much remains to be known about the identification of this important class of models.

In this paper, we provide new semiparametric identification results for a general class of learning models. We consider an environment in which the researcher has access to a short panel of choices and realized outcomes only. As such, our results are widely applicable, including in frequent situations where one does not have access to elicited beliefs data, or to a set of selection-free measurements of unobserved individual heterogeneity. Specifically, throughout our analysis we consider a potential outcome model in which individual $i$'s potential outcome in period $t$ from assignment $d$ is given by
\begin{equation}\label{eq:outcome_model_intro}
Y_{i,t}(d)=X_{i,t}^{\intercal}\beta_{t,d}+(X^{*}_{i})^{\intercal}\lambda_{t,d}+\epsilon_{i,t}(d),
\end{equation}
where $X_{i,t}$ is a vector of explanatory variables associated with individual $i$ in period $t$ (including an intercept), $X^{*}_{i}$ denotes a vector of latent individual effects (or factors), $\epsilon_{i,t}(d)$ is a transitory shock, and $(\beta_{t,d}^{\intercal},\lambda_{t,d}^{\intercal})^{\intercal}$ is an unknown parameter vector. While interactive fixed effects models of this kind have been the object of much interest in econometrics, a key distinctive feature of the set-up considered in this paper is the existence of two different types of individual effects. Namely, we assume that the individual effect $X^{*}_i$ consists of two components: $X^{*}_{k,i}$, which are supposed to be known by the agent, and $X^{*}_{u,i}$ which are initially unknown but may be learned over time. We complement this potential outcome model with a flexible choice model, in which agent $i$'s assignment in period $t$ is allowed to depend arbitrarily on contemporaneous and lagged explanatory variables, assignments and realized outcomes. This framework encompasses most of the decision models that have been considered in the learning literature.

We first establish that the model is identified under two alternative sets of conditions. Our first identification result applies to a set-up where, consistent with most of the Bayesian learning models that have been considered and estimated in the literature, we assume that the transitory shocks from the outcome equations ($\epsilon_{i,t}(d)$), as well as the unknown heterogeneity component ($X^{*}_{u,i}$), are normally distributed. In contrast, the distribution of the known heterogeneity component ($X^{*}_{k,i}$) is left unspecified. From the observation that the distribution of realized outcomes conditional on past choices and outcomes is a mixture of normal distributions, we leverage results from \citet{bruni1985identifiability} to establish identification of the joint distribution of realized outcomes, choices and known heterogeneity component.

We then also show that a pure learning model, with $X^{*}_{u,i}$ as the only source of permanent unobserved heterogeneity, remains identified without making any distributional assumption. A crucial distinction from the previous case is that, from the econometrician's perspective, this model is one of selection on observables, as individual choices depend on beliefs about $X^{*}_{u,i}$ only through prior realized outcomes, choices and covariates. This simple but powerful insight allows us to build on results from the interactive fixed-effects literature to establish identification.

We propose to estimate the model parameters using a sieve maximum likelihood estimator which we show to be consistent. We then focus on a class of functionals of the model parameters, which includes as special cases economically relevant quantities, such as the predictable and unpredictable outcome variances. These variances can in turn be used to evaluate the relative importance of, e.g., uncertainty vs. heterogeneity in the overall lifecycle earnings variability - a question that has been the object of much interest in labor economics \citep[see, among others,][]{CHN05,HVY11,CH16,GSS19}. We show that, under mild regularity conditions, the resulting estimators are consistent and asymptotically normal. We implement our sieve maximum likelihood estimator using a profile likelihood-based procedure. Importantly for practical purposes, the resulting procedure only involves a modest computational cost. Monte Carlo simulation results further indicate that our estimator exhibits good finite-sample properties. 

Finally, we illustrate our approach with an application to ability learning in the context of occupational choice, using data from the National Longitudinal Survey of Youth 1997 (NLSY97). Our method allows us to investigate this question without relying on a measurement system for latent ability, while remaining very flexible regarding how workers choose their occupations. Estimation results indicate that the share of the variance of discounted future earnings that is forecastable by the individuals increases rapidly with accumulated work experience, consistent with workers learning about their productivity through their wages. Accounting for initially known latent productivity is also important in order to understand the dispersion of wages.

\subsubsection*{Related literatures} 

Our paper contributes to several strands of the literature. First and foremost, we contribute to the literature that studies the identification of learning models, generally in the context of specific applications \citep[see, e.g.,][]{AC05,Gong_JMP,Pastorino22,AAMR25}. A central distinction from most of the papers in this literature is that we impose only mild restrictions on the choice process. Importantly, we remain agnostic about how choices depend on individual beliefs about $X^{*}_{u,i}$, while allowing these beliefs to depend arbitrarily on past choices and realized outcomes. Particularly relevant for us is the complementary work of \citet{Pastorino22}, which establishes identification results in a different and non-nested framework of a two-sided learning model in which workers and firms have imperfect information. Key to the identification strategy proposed in that paper is to leverage particular mixture representations of selected one-dimensional outcomes.\footnote{See also recent related work by \cite{dPGPS25} which investigates the identification of a two-sided matching model with learning and human capital accumulation. As in our paper, identification of the outcome equations and the distribution of unobserved heterogeneity relies on \cite{bruni1985identifiability}.} Related mixture representations also play an important role in our analysis. 

Our paper also fits into a literature that focuses on the identification of Markovian dynamic discrete choice models in the presence of persistent unobserved heterogeneity \citep[see][for a review in connection to nonlinear panel data models]{HN07,hu2008instrumental,kasahara2009nonparametric,hu2012nonparametric,sasaki2015heterogeneity,HS18,AGL21,bunting22,arellano2017nonlinear}. Unlike these papers, we do not impose a Markov structure, since current beliefs and decisions are allowed to depend on the entire history of past outcomes and decisions.\footnote{Although our framework is more general, Bayesian learning models often naturally possess a first order Markov structure. There are several additional significant differences between our paper and the listed literature. Notably, \citet{hu2012nonparametric} focus on scalar unobserved heterogeneity, whereas the existence of multivariate unobserved heterogeneity is fundamental to our main setting. Beyond this, several of their assumptions may fail to hold in our set-up. For instance, since the support of the latent beliefs is larger than the support of the choices, the requirement that the observed variables be invertible measurements of the latent variables \citep[][Assumption 2]{hu2012nonparametric} will generally fail to hold.} More broadly, our analysis is related to the literature that deals with the identification of mixture models \citep[see, for example,][and references therein]{HKS14, CK16, KL18}. In particular, central to our main identification result is the observation that the distribution of current outcomes conditional on the sequence of past choices and outcomes is a mixture of normal distributions.

Since the outcome equation in our model involves interactions between unobserved individual- and time-specific effects, our paper fits into the literature that examines the identification and estimation of panel data models with interactive fixed effects \citep[see, e.g.,][]{Madansky64,HS87,bai09,freyberger2017non}. An important distinction comes from the fact that these papers consider a selection-free environment. In contrast, individual choices, along with associated selection issues that affect the potential outcomes, play a central role in our analysis.

Finally, by applying our framework to examine how imperfect information and learning shape occupational choices and wages, our paper also fits into the literature that highlights the important role of imperfect information in labor market trajectories and outcomes \citep[see, e.g.,][]{Miller84,AG12,Papageorgiou14,Pastorino15,CPWZ18,GS19,GSS22,AAMR25}. A distinctive feature of our approach is that it allows us to remain flexible on how agents sort across occupations and form their beliefs about future earnings. Our identification results allow, in particular, for potential deviations from rational expectations on future outcomes, which recent evidence based on subjective beliefs has shown to be important \citep[see, e.g.,][]{DGM21,CGSS24}.

\subsubsection*{Organization of the paper} 

The remainder of the paper is organized as follows. Section \ref{sec2} introduces and discusses the set-up of the model. Section \ref{sec:id} contains our main identification results, both for the general case and for the case of a pure learning model. We discuss in Section \ref{sec:estimation} the estimation and inference on the parameters of interest, before turning in Section \ref{sec:mc} to the implementation of our estimator and its finite-sample performances. We illustrate in Section~\ref{sec6} our approach with an application to ability learning in the context of occupational choice. Section \ref{sec:conclude} concludes. The appendix gathers all the proofs, additional material on the variance decompositions, the implementation of our estimator, and further Monte Carlo simulation results. Finally, our estimation method can be implemented using our companion Python package, \texttt{spmlex}, which is available at \hyperlink{https://github.com/pdiegert/spmlex}{https://github.com/pdiegert/spmlex}.

\medskip 

\textit{Notation}: for a given random variable $A$, we denote by $a$ its realization, $\Supp(A)$ indicates its support, $F_{A}$ denotes its cumulative distribution function, $q_{\alpha}[A]$ its $\alpha \in [0,1]$ quantile, whereas $f_{A}$ indicates its probability mass or density function. For any sequence $(a_{1},a_{2},\dots,a_{S})$ and $s\le{S}$, we let $a^{s}=(a_{1},a_{2},\dots,a_{s})$. $A\indep{B}\mid{C}$ indicates that $A$ and $B$ are statistically independent conditional on $C$. Finally, unless stated otherwise, we suppress the individual subscript $i$ from all random variables in the remainder of the paper.

\section{Set-up}\label{sec2}

Throughout the paper, we consider a set-up where potential outcomes have an interactive fixed-effect structure of the following form:
\begin{equation}\label{eq:outcome_model}
Y_{t}(d)=X_{t}^{\intercal}\beta_{t,d}+X^{*}_{k} \lambda^{k}_{t,d}+(X_{u}^{*})^{\intercal}\lambda^{u}_{t,d}+\epsilon_{t}(d),
\end{equation}
where $d$ represents a possible value of individual $i$'s assignment in period $t$, $Y_{t}(d)$ is a scalar potential outcome variable associated with assignment $d$, $X_{t}$ is a vector of observed explanatory variables, $X^{*}\coloneqq(X_{k}^{*},(X_{u}^{*})^\intercal)^\intercal$ are unobserved (to the econometrician) factors, $(\beta_{t,d}^\intercal,\lambda_{t,d}^\intercal)^\intercal$ with $\lambda_{t,d}\coloneqq(\lambda^{k}_{t,d},(\lambda_{t,d}^{u})^\intercal)^\intercal$ is an unknown parameter vector, and $\epsilon_{t}(d)$ is an idiosyncratic random shock. For example, $Y_{t}(d)$ may represent potential log wages in occupation $d$. $Y_{t}(d)$ may depend on some observed individual and possibly time-varying characteristics ($X_{t}$) as well as on multiple dimensions of unobserved abilities ($X^{*}$), which may play different roles in different occupations \citep[see, e.g., ][]{Hincapie20,AAMR25}. This set-up is fairly general and can be applied in a wide range of contexts. For instance, $Y_t(d)$ may alternatively represent the potential log-quantity of a particular product sold by a firm in a given market $d$ \citep[see, e.g., ][]{Berman_al_19}. This framework can also be used in the health context, where $Y_t(d)$ may correspond to a measure of health outcome associated with a certain drug (e.g., CD4 cell counts associated with a particular HIV drug treatment, as in \citeauthor{ChanHamilton06}, \citeyear{ChanHamilton06}), or to the body mass index associated with a certain type of diet.

Importantly, we allow for two distinct types of latent individual effects. Namely, $X^{*}_{k}$ is assumed to be known by the agent, while $X^{*}_{u}$ is initially unknown but may be gradually revealed over time. For example, worker $i$'s log wage in occupation $d$ at time $t$, $Y_{t}(d)$, may depend on her unobserved (to the econometrician) occupation specific productivity, $X^{*}_{k}\lambda_{t,d}^{k}+(X^{*}_{u})^{\intercal}\lambda_{t,d}^{u}$. As the worker accumulates more experience, she may update her belief about $X^{*}_{u}$, and thus about the initially unknown portion of productivity in each of the possible occupations. 

Turning to the choice and learning process, the key restriction that we place on an individual's assignment in period $t$ (denoted as $D_{t}$) is that it does not directly depend on the unknown component of heterogeneity. Specifically, we assume that:
\begin{equation}\label{eq:choice-ind}
    D_{t} \indep X^{*}_{u} \mid X^t,Y^{t-1},D^{t-1},X^{*}_{k}.
\end{equation}
The above conditional independence assumption highlights the asymmetry between the two types of latent effects: assignments may arbitrarily depend on the known component of the latent effect $X^{*}_{k}$, but not on the unknown component of the latent effect $X^{*}_{u}$. However, we do allow the assignment rule to depend arbitrarily on current and lagged covariates, as well as lagged outcomes and choices. As a result, we do not restrict how agents form their beliefs about $X^{*}_{u}$, provided that such beliefs are a measurable function of $X^t,Y^{t-1},D^{t-1}$ and $X^{*}_{k}$. We also remain agnostic about how assignments depend on agents' beliefs over $X^{*}_{u}$.

This choice process accommodates a wide range of models that have been considered in the learning literature. In particular, this framework is consistent with a set-up in which agents are rational and Bayesian updaters, so that beliefs coincide with the true distribution of $X^{*}_{u}$ conditional on their information set at a given point in time, which may include all realized variables and model parameters. Alternatively, this accommodates situations where individual decisions may not involve beliefs over the distribution of $X^{*}_{u}$, or depend instead on myopic beliefs that are formed based on the prior-period choice and outcome. This set-up also allows for heterogeneous beliefs formation, where, for instance, some agents may have rational expectations about their unobserved characteristic $X^{*}_{u}$, while others may have biased (e.g. over-optimistic) beliefs.\footnote{Our set-up accommodates situations where heterogeneity in beliefs formation depends on $X^{*}_{k}$ and is therefore unobserved to the econometrician.}

Finally, we denote the conditional choice probability (CCP) function as
\begin{align*}
{h}_{t}(d^{t},x^{t},y^{t-1},x_{k}^{*})\coloneqq & \Pr(D_{t}=d\mid{}X^{t}=x^{t},Y^{t-1}=y^{t-1},D^{t-1}=d^{t-1},X^{*}_{k}=x_{k}^{*}). 
\end{align*}
These CCPs play a central role in our identification analysis. In the following section, we provide sufficient conditions under which the CCPs - which are latent objects because of the conditioning on $X^{*}_{k}$ -  are identified. In empirical applications it is very common to impose some structure on the choice process. For example, in a dynamic discrete choice framework, it is standard to assume that 
\begin{equation*}
D_{t}=\argmax_{{d}\in\Supp(D_{t})}\left\{\overline{v}({d},X_{t},X^{*}_{k},S_{t})+\eta_{t}({d})\right\},
\end{equation*}
where the conditional value function $\overline{v}$ is known up to a finite-dimensional vector of parameters, $S_{t}$ are sufficient statistics for the conditional distribution of $X^{*}_{u}$ at time $t$, and $\eta_{t}$ follows a known distribution. Having identified the CCPs, one can then apply standard identification arguments from the dynamic discrete choice literature to identify $\overline{v}$ \citep[see, e.g.,][]{HotzMiller93, AM10, CGS16}, and then recover the primitives of the choice model \citep[see, e.g.,][]{AAMR25}.

\paragraph{Uncertainty and learning.}\label{sec:uncertainty}

A central feature of the model is the distinction between three forms of unobserved heterogeneity: (1) permanent heterogeneity that is known to the agent, $X^{*}_k$, (2) permanent heterogeneity that is initially unknown to the agent, $X^{*}_u$, and (3) transitory time-varying shocks, $\epsilon=\{\epsilon_{t}(d):d\in\Supp(D_t),~t=1,2,\ldots\}$. This provides a framework for quantifying the importance of uncertainty in outcomes. At $t = 1$, the variance in future outcomes can be decomposed into a component that depends on $(X^{*}_u, \epsilon)$ and a component that depends on $X^{*}_k$. \citet{CHN05} and \citet{CH16} consider this decomposition in the context of educational choice, decomposing the variance in lifetime earnings into a component that is predictable when deciding to go to college and a component that is not.

In our framework, the importance of uncertainty can change over time as agents learn about $X^{*}_u$ by observing realized outcomes and covariates, and use this information to self-select into different alternatives. We provide in Appendix \ref{App:var_dec} a class of variance decomposition parameters that includes both the $t = 1$ decomposition as well as $t > 1$ decompositions that incorporate these learning and selection effects. These decompositions, which are identified from the model parameters, each provide different ways to quantify the importance of uncertainty to future outcomes. After establishing identification of the model, we will pay special attention to estimation and inference of a broad class of functionals that encompasses these kinds of variance decompositions.

\section{Identification}\label{sec:id}

We first provide in Subsection~\ref{subsec31} a high-level overview of the underlying reweighting scheme that plays an important role in both of the proposed identification strategies. We then discuss identification in the case with both known and unknown unobserved heterogeneity (Subsection~\ref{sec:kl}), before turning to the pure learning case where the only source of permanent unobserved heterogeneity is initially unknown to the agent (Subsection~\ref{sec:l}).

\subsection{Reweighting strategy}\label{subsec31}
Key to the identification problem analyzed in this paper is how to recover the conditional distributions of potential outcomes (i.e., $f_{Y_t(d_t)|X_{t},X^{*}}$ for each $t$ and $d_t$) and selection probabilities (i.e., $f_{D_t|X_{t},X^{*}_k}$ for each $t$), from the selected population distribution (i.e., $f_{Y^T,D^T, X^T}$) which is directly identified from the data. 

We now provide intuition as to how one can leverage the structure imposed on the choice process to address the censored data problem.
To illustrate, consider a simplified version of our model with a binary choice in each period (i.e., $\Supp(D_{t}) = \{0, 1\}$) and without covariates. Let $D := \prod_{t = 1}^T D_{t}$, $Y := (Y_1, \ldots, Y_T)$ and $Y(1) := (Y_{1}(1), \ldots, Y_{T}(1))$, and focus on identification of the distribution of the potential outcome $Y(1)$. By Bayes' rule, the relationship between the target and censored distributions can be characterized as follows:
\begin{align*}
f_{Y  |  D }(y | 1)\frac{f_{D  }(1 )}{f_{D  | Y (1) }(1 | y)} &= f_{Y (1)  }(y)
\end{align*}
where the conditional density $f_{Y  | D  }(y|1)$, which is directly identified from the data, is weighted by a selection adjustment term, $\frac{f_{D   }(1)}{f_{D  | Y(1)  }(1 | y)}$.

Our framework provides a strategy for identifying these selection weights. Let us first assume that all components of the latent effect are initially unknown. In a learning context where decision makers' actions depend on beliefs over $X^{*}$, it is often natural to assume that beliefs depend only on past realized outcomes and choices, and that:
\begin{align}
f_{D_{t} | Y(1), D^{t-1}}(1 | y, 1) = f_{D_{t} | Y^{t-1}(1), D^{t-1}}(1 | y^{t-1}, 1).\label{eq:selection-ccps}
\end{align}
where the right-hand side of Equation \eqref{eq:selection-ccps} is identified from the joint distribution of $(D^{t}, Y^{t-1})$ conditional on $D^{t-1} = 1$. Applying this reasoning recursively, it follows that $f_{D | Y(1)}(1 | y)$ (and thus the selection weight) is identified as follows:
\begin{align*}
  f_{D_{} | Y (1) }(1 \mid y) = f_{D_{T} | Y ^{T-1}(1), D ^{T-1} }(1 \mid y^{T-1}, 1) f_{D_{T-1} | Y ^{T-2}(1), D ^{T-2} }(1 \mid y^{T-2}, 1) \cdots f_{D_{1}  }(1).
\end{align*}

We build on this idea when establishing in Section \ref{sec:l} identification of a version of the model we call {\it pure learning} (where $X^{*}=X^{*}_u$).\footnote{In this section, we assume there are no covariates for clarity of exposition. With covariates, the selection weights are based on the joint transition probabilities for $(X_t, D_t)$. The exact form of the selection weight is derived in \ref{sec:proofs_l}} The conditional independence restriction in Equation \eqref{eq:selection-ccps} will generally break down, however, when agents also possess persistent private information that affects their decision (i.e., $X^{*}_k$). We propose in Section \ref{sec:kl} an identification strategy that can be used in such situations. A key additional step in this context is to show, relying on results from \cite{bruni1985identifiability}, that maintaining a normality assumption that is very commonly made in the learning literature is sufficient to identify the joint distribution of $(Y ^T, D ^T, X^{*}_k)$ in a first step. The model parameters can then be identified in a second step, along the lines of the reweighting strategy discussed above.

\subsection{Known and unknown heterogeneity}\label{sec:kl}

This section provides sufficient conditions for identification of the baseline model discussed in Section \ref{sec2}. We first impose a form of conditional independence on $(\epsilon_{t}(d), D_{t}, X_{t})$.

\begin{manualassumption}{KL1}\label{assn:kl_independence}
Equation \eqref{eq:outcome_model} holds, and for any $t \geq 2$ and $d\in\Supp(D_{t})$,
\begin{equation*}
F_{\epsilon_{t}(d),D_{t},X_{t}|Y^{t-1},D^{t-1},X^{t-1},X^{*}} = F_{\epsilon_{t}(d)}F_{D_{t}|X^{t},Y^{t-1},D^{t-1},X^{*}_k}F_{X_{t}|Y^{t-1},D^{t-1},X^{t-1}}.
\end{equation*}
Furthermore, for any $d\in\Supp(D_1)$, $F_{\epsilon_{1}(d),D_{1},X_{1}|X^{*}} = F_{\epsilon_{1}(d)}F_{D_{1}|X_{1},X^{*}_k}F_{X_{1}|X^*}$.
\end{manualassumption}

Assumption \ref{assn:kl_independence} imposes the potential outcome model in Equation \eqref{eq:outcome_model} and contains three independence conditions. First, it implies that the additive transitory shock in the outcome equation ($\epsilon_{t}(d)$) is independent of all contemporaneous and lagged variables. This is closely related to the standard fixed effect assumption that dependence in outcomes across periods is due to the latent fixed effect (e.g., \citet[Assumption N5]{freyberger2017non} and \citet[Restriction 2]{sasaki2015heterogeneity}). However, note that we allow for arbitrary within-period dependence between the additive shocks ($\epsilon_{t}(d)$ and $\epsilon_{t}(\tilde{d})$, for $d\neq \tilde{d}$). Second, the unknown factor ($X^{*}_u$) does not directly affect treatment assignments ($D_{t}$), a natural restriction discussed in Section \ref{sec2}. Third, we also impose that the transition of the control variables ($X_{t}$) does not directly depend on the time-invariant unobservables ($X^*$). Importantly, this does allow $X_{t}$ to depend on $X^*$ through past choices and outcomes. For instance, in the context of occupational choices, this restriction is implied by the standard assumption that occupation-specific work experiences depend on $X^*$ through past occupational choices (see, e.g., \citeauthor{keane1997career}, \citeyear{keane1997career}).\medskip

Our second assumption \ref{assn:kl_normality} imposes that the unknown component of the individual effect is drawn from a multivariate normal distribution, and that the random shock in the outcome equation is normally distributed too. This is a common assumption in the Bayesian learning literature, to which we return in Remark \ref{rem:normality}. 

\begin{manualassumption}{KL2}\label{assn:kl_normality} For all $(x_1,x_k^*)\in\Supp(X_1)\times\Supp(X_k^*)$, 
$X^{*}_{u}\mid(X_{1},X^{*}_{k})=(x_1,x_{k}^{*})\sim{N}\left(0,\Sigma_{u}(x_{1})\right)$ and $\forall~d\in\Supp(D_{t}),~\epsilon_{t}(d)\sim{N}(0,\sigma_{t,d}^{2})$.
\end{manualassumption} 

Assumption \ref{assn:kl_normality} implies a Gaussian conjugate posterior distribution for $X_u^*$, which we summarize in Lemma \ref{lemma:conjugate}. Importantly, neither this assumption nor Assumption \ref{assn:kl_independence} place any restriction on the dependence between $X_k^*$ and $X_1$.\footnote{Lemma \ref{lemma:conjugate} and our main identification result would go through if one replaces the first part of Assumption \ref{assn:kl_normality} with $X^{*}_{u}\mid(X_{1}=x_{1},X^{*}_{k}=x_{k}^{*})\sim{N}\left(0,\Sigma_{u}(x_1,x_k)\right)$ under appropriate regularity conditions on $x_k\mapsto\Sigma_{u}(x_1,x_k)$, including for each $x_{k}^{*}-\tilde{x}_{k}^{*}>0$, $\Sigma_{u}(x_{1},x_{k}^{*})-\Sigma_{u}(x_{1},\tilde{x}_{k}^{*})$ is positive (or negative) semi-definite. For simplicity, we maintain the stronger Assumption \ref{assn:kl_normality} when establishing identification in Theorem 1 below and in the rest of the paper.} To do so, define $(\mu_{t},\Sigma_{t})$ recursively as follows. First, $(\mu_{1},\Sigma_{1})=(0,\Sigma_{u}(x_{1}))$. Second,
\begin{align*}
\Sigma_{t+1} & =\left(\Sigma_{t}^{-1}+\lambda_{t,d_{t}}^{u}(\lambda_{t,d_{t}}^{u})^{\intercal}\sigma_{t,d_{t}}^{-2}\right)^{-1}, \\
\mu_{t+1}      & =\Sigma_{t+1}\left(\Sigma_{t}^{-1}\mu_{t}+\lambda_{t,d_{t}}^{u}\frac{y_{t}-x_{t}^{\intercal}\beta_{t,d_{t}}-x^{*}_{k}\lambda_{t,d_{t}}^{k}}{\sigma_{t,d_{t}}^{2}}\right).
\end{align*}

\begin{lemma}\label{lemma:conjugate}
Let Assumptions \ref{assn:kl_independence} and \ref{assn:kl_normality} hold. Then, for all $t \geq 2$, $X^{*}_{u}$ conditional on  $(Y^{t-1},D^{t-1},X^{t},X^{*}_{k})=(y^{t-1},d^{t-1},x^{t},x^{*}_{k})$ is distributed ${N}(\mu_{t},\Sigma_{t})$.
\end{lemma}

Suppose $X^{*}_{u}\in\mathbb{R}^{p}$. Our three remaining assumptions are as follows.

\begin{manualassumption}{KL3}\label{assn:kl_norm}
 (A) For some $d\in\Supp(D_{1})$, the element of $\beta_{1,d}$ associated with the constant term is zero, and $\lambda_{1,d}^{k}=1$. (B) For some $d^{p}\in\Supp(D^{p})$, $\left(\lambda_{1,d_{1}}^{u}\cdots \lambda_{p,d_p}^{u}\right)=I_{p\times{p}}$.
\end{manualassumption}

Assumption \ref{assn:kl_norm} is a location-scale normalization on finite-dimensional parameters, which reflects the fact that the latent factors are only identified up to location and scale. This type of assumption is standard in interactive fixed effect models.\medskip

Finally, we impose in Assumptions~\ref{assn:kl_support} and \ref{assn:kl_regularity} below several regularity conditions. We start with Assumption \ref{assn:kl_support}, which places support restrictions on various objects of the model. In what follows, we let $\theta_{1}\coloneqq\left\{\{\beta_{t},\lambda_{t},\sigma_{t}^{2}\}_{t=1}^{T},\Sigma_{u}(x_{1})\right\}\in\Theta_{1}\subset\mathbb{R}^{\dim\Theta_{1}}$, where $\{\beta_{t},\lambda_{t},\sigma_{t}^{2}\}\coloneqq\{\beta_{t,d},\lambda_{t,d},\sigma_{t,d}^{2}\colon d\in\Supp(D_t)\}$.

\begin{manualassumption}{KL4}\label{assn:kl_support}
(A) For each $x_{1}\in\Supp(X_{1})$, $\Theta_{1}$ is a compact set. (B) $\Supp(X^{*}_{k})$ is compact. (C) For each $t$ and $d\in\Supp(D_t)$, $(\lambda_{t,d}^{u})^{\intercal}\Sigma_{t}\lambda_{t,d}^{u}+\sigma_{t,d}^{2}\neq0$, $\sigma_{t,d}^{2}\neq0$ and $\forall~x_1\in\Supp(X_1)$, $\Sigma_{u}(x_{1})$ is non-singular. (D) For each $y^{t-1},d^{t},x^{t}$ in their support, $\Supp(X^{*}_{k}\mid({Y}^{t-1},D^{t},X^{t})=(y^{t-1},d^{t},x^{t}))=\Supp(X_k^*)$ and $Var(X_k^*)\neq0$.  (E) For each $t$ and $d\in\Supp(D_t)$, $E[X_t X_t^\intercal\mid D_t=d]$ is non-singular. (F) For all $t$, $Var(D_{t})\neq0$.
\end{manualassumption}

Part (A) states that the finite-dimensional parameters $\theta_1$ belong to a compact set. Part (B) requires that the known latent factor $X^{*}_{k}$ has compact support. This holds if the distribution of $X^{*}_{k}$ has discrete support, although this clearly applies to a broader set of distributions. We return to this compactness condition in Remark \ref{rem:compact} below. Part (C) requires certain normally distributed random variables to have non-singleton support. Part (D) imposes a rectangular support condition and a nondegeneracy assumption on the distribution of $X^{*}_{k}$. These conditions are typically satisfied in dynamic discrete choice models with unobserved heterogeneity, which generally impose a large support assumption on the random utility shocks. Part (E) imposes that the support of $X_{t}$ conditional on $D_t$ is sufficiently rich. Finally, Part (F) imposes the requirement that the support of the choice variables contain at least two elements.\medskip

Next, Assumption \ref{assn:kl_regularity} below contains a set of regularity conditions that ensure that the latent individual effect $X^{*}$ alters outcomes sufficiently differently across time and assignments. 

\begin{manualassumption}{KL5}\label{assn:kl_regularity}
(A)   For each $t$ and $d_{t}\in\Supp(D_t)$ there exist two sequences $(d^{t-1}, \tilde{d}^{t-1})\in\Supp(D^{t-1})^2$ such that $(\lambda_{t,d_{t}}^{u})^{\intercal}\Sigma_{t}\sum_{s=1}^{t-1}\left(\lambda_{s,d_{s}}^{u}\frac{\lambda_{s,d_{s}}^{k}}{\sigma_{s,d_{s}}^{2}}-\lambda_{s,\tilde{d}_{s}}^{u}\frac{\lambda_{s,\tilde{d}_s}^{k}}{\sigma_{s,\tilde{d}_s}^{2}}\right)\neq0$. (B) For all $t$ and $d_{t}\in\Supp(D_t)$, $\lambda_{t,d_{t}}^{k}\neq0$.  (C) For all $t$ and $d^{t}\in\Supp(D^t)$, $  \lambda_{t,d_{t}}^{k} - (\lambda_{t,d_{t}}^{u})^{\intercal}\Sigma_{t}  \sum_{s=1}^{t-1}\lambda_{s,d_{s}}^{u}\frac{\lambda_{s,d_{s}}^{k}}{\sigma_{s,d_{s}}^{2}}\neq0.$ (D)  For all $d^2\in\Supp(D^2)$, $(\lambda_{2,d_{2}}^{u})^{\intercal}\Sigma_{2}\lambda_{1,d_{1}}^{u}\frac{\lambda_{1,d_{1}}^{k}}{\sigma_{1,d_{1}}^{2}}\neq0$. (E) There exists $\{(d_{2,i},\tilde{d}_{2,i})\in\Supp{(D_{2})}^2:i=1,2,\dots,p\}$ which satisfy
\begin{align*}
 \left(\lambda_{2,d_{2,1}}^{u}\cdots{\lambda_{2,d_{2,p}}^{u}}\right)^{-\intercal}\mathrm{vec}(\lambda_{2,d_{2,1}}^{k},\dots,\lambda_{2,d_{2,p}}^{k}) 
 \neq\left(\lambda_{2,\tilde{d}_{2,1}}^{u}\cdots{\lambda_{2,\tilde{d}_{2,p}}^{u}}\right)^{-\intercal}\mathrm{vec}(\lambda_{2,\tilde{d}_{2,1}}^{k},\dots,\lambda_{2,\tilde{d}_{2,p}}^{k}).
\end{align*}
(F) For all $d^T\in\Supp(D^T)$, $\{\lambda_{t,d_{t}}^{u}:t=1,\ldots,T\}$ is linearly independent.
\end{manualassumption}

This assumption is fairly mild as it primarily rules out knife-edge cases where the effect of different elements of permanent unobserved heterogeneity is exactly zero.\footnote{This type of assumption is similarly required in latent factor models without selection or learning in order to rule out degeneracies \citep[see, e.g.,][Assumption L4]{freyberger2017non}.} Part (A) requires that the aggregate effect of $X^{*}_{k}$ on outcomes associated with choice $d_{t}$ is different for at least two histories $(d^{t-1},\tilde{d}^{t-1})$. Part (B) assumes that the direct effect of $X^{*}_{k}$ is non-zero in each period and each assignment. Part (C) states that the aggregate effect of $X^{*}_{k}$ on outcomes must be non-zero---that is, that the direct effect $\lambda_{t,d_{t}}^{k}$ is not perfectly offset by the effect mediated through previous choices. Part (D) ensures that there is a non-zero effect of previous choices in $t=2$. Part (E) requires that for $t=2$ the relative effect of known and unknown $X^{*}$ changes across choices. In the special case where $X^{*}_{u}\in\mathbb{R}$ (i.e., $p=1$), the condition reduces to $\frac{\lambda_{2,d_{2}}^{k}}{\lambda_{2,d_{2}}^{u}}\neq\frac{\lambda_{2,\tilde{d}_{2}}^{k}}{\lambda_{2,\tilde{d}_{2}}^{u}}$, i.e., that the ratio of factor loadings varies across some assignments. More generally, for $X^{*}_{u}\in\mathbb{R}^{p}$, this condition implies that, for $t=2$, the set of assignments must contain at least $p+1$ elements. Finally, Part (F) requires that the initially unknown factor affects each outcome via a different linear combination.

\medskip
We are now in a position to state our main identification result. We denote by $\theta=\left\{\{\beta_{t},\lambda_{t},\sigma_{t},g_t,h_{t}\}_{t=1}^{T},\Sigma_{u},F_{X^{*}_{k},X_1}\right\}\in\Theta$ the model parameters, where $g_t: =dF_{X_{t}|Y^{t-1},D^{t-1},X^{t-1}}$. 

\begin{theorem}\label{thm:kl}
Suppose the distribution of $(Y_{t},D_{t},X_{t})_{t=1}^{T}$ is observed for $T={2p}+1$ periods, and that Assumptions \ref{assn:kl_independence}-\ref{assn:kl_regularity} hold. Then $\theta$ is point identified.
\end{theorem}

The first step is to show, from Assumptions~\ref{assn:kl_independence} and~\ref{assn:kl_normality} and Lemma~\ref{lemma:conjugate} that $Y_{t}$ is normally distributed conditional on lagged outcomes $Y^{t-1}$, assignments $D^{t}$, covariates $X^{t}$ and the known component of the latent individual effect, $X^{*}_{k}$. This implies that $Y_{t}$ conditional on $(Y^{t-1}, D^{t},X^{t})$ is a Gaussian mixture distribution parameterized by $X^{*}_{k}$. Then under the compact support and non-degeneracy assumptions (Assumptions \ref{assn:kl_support} (A)-(C)), one can apply a result from \citet{bruni1985identifiability} to identify the aforementioned mixture distribution up to an affine transformation of $X^{*}_k$. Next, the normalization and regularity assumptions (Assumptions \ref{assn:kl_norm}-\ref{assn:kl_regularity}) are used to pin down the affine transformation, leading to identification of the distribution of $(Y^{T},D^{T},X^{T},X^{*}_{k})$. Knowledge of this distribution identifies the components of the model related to the known component of the individual latent effect, namely $\left\{\{\beta_{t},\lambda_{t}^{k},h_{t}\}_{t=1}^{T},F_{X^{*}_{k},X_{1}}\right\}$. The final step is to disentangle the effect of the learned component ($X^{*}_u$) and the idiosyncratic uncertainty ($\epsilon_{t}(d))$ in order to identify $\left\{\{\lambda_{t}^{u},\sigma_{t}^{2}\}_{t=1}^{T},\Sigma_{u}\right\}$. This is done by showing that the joint distribution of $(Y^{T},D^{T},X^{T})$ conditional on $X^{*}_{k}$, suitability weighted by the assignment probabilities, is a normal-weighted mixture of normal distributions. This allows us to identify $\left\{\{\lambda_{t}^{u},\sigma_{t}^{2}\}_{t=1}^{T},\Sigma_{u}\right\}$ from the second moments of the reweighted distribution. We refer the interested reader to Section \ref{sec:proofs_kl} for a formal derivation.\footnote{Note that, while we assume for simplicity that $T=2p+1$, extension to a larger horizon $T$ is straightforward. The same applies for the pure learning model considered in Section \ref{sec:l}.}

\begin{remark}[\textbf{Compact support assumption}]\label{rem:compact}
Assumption \ref{assn:kl_support} (B) imposes that the known component of the latent individual effect has bounded support. In applications, it is common to assume $X^{*}_{k}$ has finite support with known cardinality. Assumption \ref{assn:kl_support} (B) relaxes this restriction in the sense that the number of support points of $X^{*}_{k}$ need not be known a priori, and indeed may be infinite.\footnote{Compactness is used in particular to apply the Stone-Weierstrass approximation theorem, which plays an important role in the identification proof of \citet[Theorem 1]{bruni1985identifiability}.}
\end{remark}

\begin{remark}[\textbf{Normality of unknown factor}]\label{rem:normality}
As summarized in Lemma \ref{lemma:conjugate}, an important implication of the normality assumptions (Assumption \ref{assn:kl_normality}) is the resulting normal conjugate prior with a tractable closed form. For this reason, these assumptions are very common in the learning literature. In the context of our analysis though, the key implication of normality is rather to enable identification of the distribution of $Y_{t}\mid\left(Y^{t-1},D^{t},X^{t},X^{*}_{k},\right)$ from variation in the realized outcome $Y_{t}$ only. Namely, under Assumption~\ref{assn:kl_normality}, the distribution of $Y_{t}\mid\left(Y^{t-1},D^{t},X^{t}\right) $ is a mixture of normal distributions with mixture weights given by the distribution of $X^{*}_{k}\mid\left(Y^{t-1},D^{t},X^{t}\right)$. This allows us to establish identification by leveraging results for mixtures of normal distributions \citep{bruni1985identifiability}.\footnote{
That identification of the distribution of $X^{*}_{k}$ arises from variation in the scalar outcome variable $Y_{t}$ highlights why we restrict $X^{*}_{k}$ to be a scalar random variable. If $Y_{t}$ was vector-valued instead, then we expect that our arguments would easily extend to allow for a multivariate $X^{*}_{k}$.}
\end{remark}

\begin{remark}[\textbf{Role of covariates}] Inspection of the proof shows that the covariates $X_t$ are not needed to identify the parameters $\theta$, beyond $\{\beta_t:t=1,\ldots,T\}$. In particular, one can easily adapt the proof to establish identification for a more flexible specification where $X_t$ enters the outcome equation through an additive nonparametric shifter. We maintain linearity throughout for estimation precision and to preserve tractability. 
\end{remark}

\begin{remark}[\textbf{Invariance to normalization}]
The normalization assumption (Assumption \ref{assn:kl_norm}) is a true normalization in the sense that particular meaningful economic parameters are invariant to the assumption. Specifically, we can show that this is the case of the average and quantile structural functions. To formalize this notion, define $C_{t,d}^{k} \coloneqq X^{*}_{k} \lambda_{t,d}^{k}$, $C_{t,d}^{u} \coloneqq (X^{*}_{u})^{\intercal} \lambda_{t,d}^{u}$ and let $Q_{\alpha}\left[X\right]$  be the $\alpha$-quantile of a random variable $X$. Let $x \in \Supp(X_{t})$ and define the quantile structural functions associated with the potential outcomes $Y_t(d_t)$ as follows:
\begin{align*}
s_{1,t}(x,\alpha)=&x^{\intercal}\beta_{t,d_t}+Q_{\alpha}[C_{t,d_t}^{k}+C_{t,d_t}^{u}+\epsilon_{t}(d_t)],\\
s_{2,t}(x,\alpha_{1},\alpha_{2},\alpha_{3})=&x^{\intercal}\beta_{t,d_t}+Q_{\alpha_{1}}[C_{t,d_t}^{k}]+Q_{\alpha_{2}}[C_{t,d_t}^{u}]+Q_{\alpha_{3}}[\epsilon_{t}(d_t)],
\end{align*}
and the average structural function as $s_{3,t}(x)=x^{\intercal}\beta_{t,d_t}+\int{u}dF_{C^k_{t,d_t}+C^u_{t,d_t}+\epsilon_{t}(d_t)}(u)$. In Appendix \ref{prf:normalized} we prove the following corollary:
\begin{corollary}\label{thm:normalized}
Suppose the Assumptions \ref{assn:kl_independence}, \ref{assn:kl_support} and \ref{assn:kl_regularity} hold and that for each $(x_1,x_{k}^{*})\in\Supp(X_1)\times\Supp(X_k^*)$, $X^{*}_{u}\mid(X_{1},X^{*}_{k})=(x_1,x_{k}^{*})\sim{N}\left(\mu_{u},\Sigma_{u}(x_{1})\right)$ and for all $t$ and $d\in\Supp(D_t)$, $\epsilon_{t}(d)\sim{N}(c_{t,d},\sigma_{t,d}^{2})$. Furthermore, suppose that for some $d^p\in\Supp(D^p)$, $(\lambda_{1,d_{1}}^{u}\cdots\lambda_{p,d_{p}}^{u})$ is full rank. Then $s_{1,t}(x,\cdot)$, $s_{2,t}(x,\cdot,\cdot,\cdot)$ and $s_{3,t}(x)$ are identified for all $x$ on the support of $X_{t}$.
\end{corollary}
\end{remark}

\subsection{Pure learning model}\label{sec:l} 

This section considers a special case of the model of Section \ref{sec2}, in which all components of the latent individual effect are initially unknown to the decision maker ($X^{*}=X^{*}_{u}$). Without needing to distinguish initially known and unknown heterogeneity, a stronger identification result is achieved. In particular, no parametric restrictions on the distribution of the unobservables are required. We establish identification in this model under Assumptions \ref{assn:l_independence}-\ref{assn:l_regularity} stated below.

\begin{manualassumption}{L1}\label{assn:l_independence} For all $t$ and $d\in\Supp(D_t)$,
$Y_{t}(d)=X_{t}^{\intercal}\beta_{t,d}+(X^{*})^{\intercal}\lambda_{t,d}+\epsilon_{t}(d)$.
For any $t \geq 2$ and $d\in\Supp(D_t)$,
\begin{equation*}
  F_{\epsilon_{t}(d),D_{t},X_{t}|Y^{t-1},D^{t-1},X^{t-1},X^{*}}  =
  F_{\epsilon_{t}(d)}F_{D_{t}|Y^{t-1},D^{t-1},X^{t}}F_{X_{t}|Y^{t-1},D^{t-1},X^{t-1}}.
\end{equation*}
Furthermore, for any $d\in\Supp(D_1)$, $F_{\epsilon_{1}(d),D_{1},X_{1}|X^{*}}  =
  F_{\epsilon_{1}(d)}F_{D_{1}|X_{1}}F_{X_{1}|X^*}.$
\end{manualassumption}

Assumption \ref{assn:l_independence} adapts Assumption \ref{assn:kl_independence} to reflect that there is no initially known component of unobserved heterogeneity.

\begin{manualassumption}{L2}\label{assn:l_normality}
(A) The joint density of $(Y,X^{*})$ and $(D,X)$ admits a bounded density with respect to the product measure of the Lebesgue measure on $\Supp(Y)\times\Supp(X^{*})$ and some dominating measure on $\Supp(D)\times\Supp(X)$. All marginal and conditional densities are bounded. (B) For each $x_1\in\Supp(X_1)$, $X^{*}\mid{X_1}=x_1$ has full support. (C) For each $t$ and $d\in\Supp(D_t)$, the characteristic function of $\epsilon_{t}(d)$ is non-vanishing, and $E[\epsilon_{t}]=0$.
\end{manualassumption}

Assumption \ref{assn:l_normality} substantially weakens Assumption \ref{assn:kl_normality} by replacing the normality assumption with a full support assumption. Let $X^{*} \in \mathbb{R}^p$.

\begin{manualassumption}{L3}\label{assn:l_norm}
For some $d^p\in\Supp(D^p)$, (A) $\left(\lambda_{1,d_{1}}\cdots\lambda_{p,d_{p}}\right)=I_{p\times{p}}$ and (B) the element of $\beta_{t,d_t}$ associated with the constant component of $X_t$ is zero.
\end{manualassumption}

\begin{manualassumption}{L4}\label{assn:l_support}
(A) For each $(y^{t-1},x^t)\in\Supp(Y^{t-1},X^t)$, $\Pr(D_t=d\mid Y^{t-1}=y^{t-1},X^{t}=x^t)>0$ for all $d\in\Supp(D_t)$. (B) For each $x_1\in\Supp(X_1)$, the variance-covariance matrix of $X^{*}\mid{X_1=x_1}$ is full rank. (C) For each $t$ and $d\in\Supp(D_t)$, the variance-covariance matrix of $X_t$ conditional on $D_t=d$ is non-singular.
\end{manualassumption}

Assumption \ref{assn:l_norm} are normalization assumptions, which are standard in interactive fixed effect models. Assumption \ref{assn:l_support} (A) is similar to Assumption \ref{assn:kl_support} (D). It requires that for each history ($y^{t-1},d^{t-1},x^{t}$), some units are assigned to $D_{t}=d_{t}$ for each $d_{t}\in\Supp(D_{t})$. This assumption is typically satisfied in parametric dynamic discrete choice models (see, e.g., \citeauthor{keane1997career}, \citeyear{keane1997career} and \citeauthor{Blundell17}, \citeyear{Blundell17} for a survey). At the cost of increased notational burden, this assumption could be weakened to hold for certain sequences of choices only.

\begin{manualassumption}{L5}\label{assn:l_regularity}
For any $d^T\in\Supp(D^T)$, all $p\times p$ submatrices of $(\lambda_{1,d_1}^{u}\cdots\lambda_{t,d_{t}}^{u})$ are full rank.
\end{manualassumption}

Assumption \ref{assn:l_regularity} is a standard assumption in the interactive fixed-effects literature \citep[see, e.g., Assumption N6,][]{freyberger2017non}. Similarly to Assumption \ref{assn:kl_regularity}, it rules out degeneracies by ensuring that the outcome in each period $Y_{t}(d_{t})$ depends on a distinct linear combination of $X^{*}_{u}$.\medskip

We now define the period $t$ conditional choice probability function as $
{h}_{t}(y^{t-1},d^{t},x^{t})\coloneqq  \Pr(D_{t}=d_{t}\mid Y^{t-1}=y^{t-1},D^{t-1}=d^{t-1},X^{t}=x^{t})$. In this pure learning environment, the CCP function does not depend on any latent variable and is thus identified directly from the data. As in Section \ref{sec:kl}, our identification result (Theorem \ref{Thm:scalar} below) does not rely on a particular structure imposed on the belief formation process. However, should there be such structure, our identification result would enable identification of the belief formation process. To illustrate this, consider a situation where agents are rational and Bayesian updaters, and where beliefs about $X^{*}_{u}$ at time $t$ are a known function of the information set and the model parameters. That is, there is a known function $s$ such that beliefs are given by $s(Y^{t-1},D^{t-1},X^{t-1},\theta)$, where $\theta$ are the model parameters. In this case, identification of $\theta$ is sufficient for identification of the beliefs.

\medskip 
We now turn to our identification result. Define $f_{\epsilon_t}=\left\{f_{\epsilon_{t}(d)}\colon{d}\in\Supp(D_{t})\right\}$. Let the model parameter vector be $\theta=\left\{\{\beta_{t},\lambda_{t},f_{\epsilon_t},g_{t},h_{t}\}_{t=1}^{T},\Sigma_{u},F_{X^{*}_{k},X_1}\right\}\in\Theta$. The following theorem states that the previous conditions are sufficient for point identification of $\theta$.

\begin{theorem}\label{Thm:scalar}
Suppose the distribution of $(Y_{t},D_{t},X_{t})_{t=1}^{T}$ is observed for $T={2p}+1$ and that Assumptions \ref{assn:l_independence}-\ref{assn:l_regularity} hold. Then $\theta$ is point identified.
\end{theorem}

Key to this result is a simple but powerful insight, namely that, under Assumption \ref{assn:l_independence}, this pure learning model is a model of selection on observables. That is, although assignment probabilities depend on unobserved beliefs over $X^{*}_{}$, they do not depend on the unobserved factor $X^{*}$ itself. It follows that one can control for beliefs at time $t$ by conditioning on prior outcomes, choices and covariates. This, in turn, allows us to express the joint distribution of $(Y^{t},D^{t},X^{t})$, suitably weighted by the assignment probabilities, as a mixture over the potential outcomes $Y^{t}(d_{t})$, conditional on the latent factor $X^{*}$ and exogenous covariates $X$. From here, the arguments of \citet{freyberger2017non} yield identification of the mixture and component distributions. See Section \ref{sec:proofs_l} for a formal proof.

\begin{remark}[\textbf{Auxiliary measurements}]
In some cases, additional unselected noisy measurements of known heterogeneity factors are available. This includes, in particular, the Armed Services Vocational Aptitude Battery (ASVAB) ability measures that are available in the National Longitudinal Survey of Youth panels. See, among many others, \cite{CHN05}, \cite{CHS10} and \cite{AHMR21}. With such auxiliary data, sufficient conditions for identification of the distribution of the latent effect are well known in the literature \citep{hu2008instrumental,CHS10}. If these conditions are satisfied conditional on each $(Y_{t},D_{t},X_{t})_{t=1}^{T}$, then the joint distribution of $\left((Y_{t},D_{t},X_{t})_{t=1}^{T},X^{*}_{k}\right)$ is identified from the auxiliary measurements. From here, one can redefine $X_{t}$ as $(X_{t},X^{*}_{k})$, and Theorem \ref{Thm:scalar} then yields distribution-free identification of the model with both known and unknown heterogeneity.
\end{remark}

\section{Estimation}\label{sec:estimation}

We propose to estimate the model parameters via sieve maximum likelihood. We let $W_{i}=(Y_{i,t},D_{i,t},X_{i,t}\colon t=1,\ldots,{T})$ and $\theta^{*}\in\Theta$ be the true value of the parameters. In the following, we focus on the model of Section \ref{sec:kl} with both known and unknown heterogeneity.\footnote{While we focus on this specification, analogous conditions can be derived for the pure learning model considered in Section \ref{sec:l}.} Under the conditions of Theorem \ref{thm:kl}, the log-likelihood contribution of $W_{i}=w$ is given by:
\begin{align}
\ell(w;\theta)&=\log\int \int \prod_{t=1}^{T}\frac{1}{\sigma_{t}\left(d_{t}\right)} \phi_{1}\left(\frac{y_{t}-x_{t}^{\intercal}\beta_{t}\left(d_{t}\right)-x_{k}^{*}\lambda_{ t,d_t}^k-(x_{u}^{*})^{\intercal}\lambda_{t,d_t}^u}{\sigma_{t}\left(d_{t}\right)}\right) \notag \\
&\quad \times \prod_{t=1}^{T} {h}_{t}(d^{t},x^{t},y^{t-1},x_{k}^{*})\times\prod_{t=1}^{T-1}g_{t}(x_{t+1};{y^{t},d^{t},x^{t}})dF_{X_1}(x_1) \notag \\
&\quad \times \frac{1}{\sqrt{\left|\Sigma_{u}\left(x_{1}\right)\right|}} \phi_{p}\left(\Sigma_{u}^{-\frac{1}{2}}\left(x_{1}\right) x_{u}^{*}\right)\times  dx_u^* dF_{X_{k}^{*}| X_1}\left(x_{k}^{*}, x_{1}\right) \label{eq:loglik}
\end{align}
where $\phi_{s}$ is the probability distribution function of the standard multivariate normal distribution with $s$ components, $g_{t}$ is the distribution of $X_{t+1}$ conditional on $(Y^{t},D^{t},X^{t})=(y^{t},d^{t},x^{t})$. There are four components of the likelihood function, which are associated with the outcomes, the assignment probabilities, the distribution of the covariates, and the joint distribution of $(X_1, X^{*})$, respectively.

To estimate $\theta$, let $\Theta_{n}$ be a finite-dimensional sieve space that serves as an approximation to $\Theta$. The sieve maximum likelihood estimator for $\theta^{*}$, $\hat\theta$, is defined as
\begin{equation}\label{eq:sieve-mle}
\frac{1}{n}\sum_{i=1}^{n}\ell(w_{i};\hat\theta)\ge\sup_{\theta\in\Theta_{n}}\frac{1}{n}\sum_{i=1}^{n}\ell(w_{i};\theta)-o_{p}(1/n)
\end{equation}
The following result states that, under Assumptions \ref{assn:kl_independence}-\ref{assn:kl_regularity} under which $\theta^{*}$ is identified, and additional standard conditions (stated in Appendix \ref{sec:est-consistency}),  $\hat{\theta}$ is a consistent estimator for $\theta^{*}$.

\begin{theorem}\label{thm:consistency}
Let ${(W_i)}_{i=1}^n$ be i.i.d. data where $T\ge{2p}+1$ and Assumptions \ref{assn:kl_independence}-\ref{assn:kl_regularity} and Assumptions
\ref{assn:markov}-\ref{assn:e5} hold. Then $\hat\theta$ as defined in Equation \eqref{eq:sieve-mle} is consistent for $\theta^{*}$.
\end{theorem}

In practice, researchers are often interested in functionals of the model parameters, such as the variance decompositions discussed in Section \ref{sec2} and Appendix \ref{App:var_dec}. These decompositions involve both the finite dimensional parameters of the model, as well as the distribution of $X^{*}_k$ and the CCPs. We provide in Theorem~\ref{thm:normality} below an inference result for a plug-in estimator of a general class of functionals of the model parameters, which include those defined in Appendix \ref{App:var_dec}. For a functional $f$, under a set of smoothness and regularity conditions similar to those given in \cite{chen2014sieve}, we show that the plug-in estimator $f(\hat{\theta})$ has an asymptotically normal distribution and characterize its asymptotic variance.  

\begin{theorem}\label{thm:normality}
Let $(W_i)_{i=1}^n$ be i.i.d. data where $T\ge{2p}+1$ and Assumptions \ref{assn:kl_independence}-\ref{assn:kl_regularity} and \ref{assn:markov}-\ref{assn:CLT} hold. Then $\sqrt{n}\frac{f(\hat{\theta})-f(\theta^{*})}{\|v_{n}^{*}\|}\underset{d}{\rightarrow} N\left(0,1\right)$ where $v_n^*$ is the sieve Riesz representer of $f(\theta)$ and $\|\cdot\|$ is defined in Equation \eqref{eq:norm} in Appendix~\ref{App:sieve_plugin}.

\end{theorem}

The convergence rate of the plug-in sieve estimator depends on the behavior of the sieve variance $\|v_n^*\|$ as $n$ diverges. Note that Theorem \ref{thm:normality} does not require that $\|v_n^*\|$ is convergent. That is, Theorem \ref{thm:normality} still applies in cases where the parameter of interest is an irregular (i.e., not $\sqrt{n}$ estimable) functional. In either case, consistent estimators for the sieve variance of certain functionals are available \citep[Section 3]{chen2014sieve}.\footnote{We leave it to future work to derive primitive conditions under which functionals such as the variances decompositions discussed in Section \ref{sec2} satisfy the high level conditions of Theorem \ref{thm:normality}.}

\section{Implementation and Monte Carlo simulations}\label{sec:mc}
In this section we show how the sieve MLE estimator introduced in Section \ref{sec:estimation} can be tractably implemented, and then perform a Monte Carlo experiment illustrating the good finite sample performance of the estimator. 

\subsection{Implementation}\label{sec:estim-impl}

We propose an implementation method combining a profiling approach that exploits the parametric components of our model, with a convenient choice of sieve space. Notice first that by integrating out $X^{*}_u$ in Equation \eqref{eq:loglik}, we obtain $\ell(w; \theta) = \log \int \ell^c(w, x_{k}^{*}; \theta^c)dF_{X^{*}_k \mid X_1}(x_{k}^{*}; x_1)$ with
\begin{align*}
  \ell^c(w, x_{k}^{*}; \theta^c) &:= \frac{1}{\sqrt{|V(w, x_{k}^{*};\theta^c)|}} \phi_T \left(V(w, x_{k}^{*};\theta^c_1)^{-\frac{1}{2}}(y^T - m(w, x_{k}^{*};\theta^c)) \right) \\
     &\qquad \times\prod_{t=1}^{T}{h}_{t}(d^{t},x^{t},y^{t-1},x_{k}^{*}) \times\prod_{t=1}^{T-1}g_{t}(x_{t+1}; y^t,d^t,x^t)dF_{X_1}(x_1),
\end{align*}
where $m(w, x_{k}^{*};\theta^c) = \left(\beta_{1,d_{1}}\cdots\beta_{T,d_{T}} \right)^{\intercal}x + \left(\lambda_{1,d_{1}}^{k}\cdots\lambda_{T,d_T}^{k}\right)^{\intercal}x_{k}^{*}$, $V(w, x_{k}^{*};\theta^c) = \left(\lambda_{1,d_{1}}^{u}\cdots\lambda_{T,d_T}^{u}\right)^{\intercal}\Sigma_u(x_1)\left(\lambda_{1,d_{1}}^{u}\cdots\lambda_{T,d_T}^{u}\right) + \text{diag}(\sigma^2_{1,d_1}, \ldots, \sigma^2_{T,d_T})$, and $\theta^c$ denotes the parameter vector excluding $F_{X_k^*|X_1}$. The above re-expression of the likelihood function embodies two insights. First, although the `complete' likelihood function $\ell^c$ is itself an integral over the missing data $X_u^*$, within our model this integral has the convenient analytical expression described above. Second, the $\ell^c$ function does not depend on the distribution of the missing data $X_k^*$, which enables a profiling approach to forming the maximum likelihood estimator.

To explain our profiling approach, suppose for simplicity that $X_k^*\indep X_1$.\footnote{We assume this simply for clarity of exposition. In the general case, one may consider a sieve space for $(X_k^*|X_1)$ as the cross product of unit simplexes over a grid of $\Supp(X_1)$.} The profile likelihood approach boils down to solving Equation \eqref{eq:sieve-mle} as $$\max_{\theta\in\Theta_n}\sum_{i=1}^n\ell(w_i,\theta)=\max_{\theta^c\in\Theta^c_n}\sum_{i=1}^n\log\int\ell^c(w_i,x_k^*;\theta^c)d[F(\theta^c)](x_k^*),$$ where $F(\theta^c)=\argmax_{F\in\mathcal{F}_n}\sum_{i=1}^n\log\int\ell^c(w_i,x_k^*;\theta^c)dF(x_k^*)$, and $\mathcal{F}_n$ and $\Theta^c_n$ are a sieve spaces for $F_{X_k^*}$ and $\theta^c$, respectively. As the non-parametric objects in $\theta^c$ are often context specific (for example, $g_t$ may be estimated in  a first step, or $h_t$ may be a parametric choice model), we focus on the choice of $\mathcal{F}_n$. Namely, we propose using a sieve space closely related to the estimator discussed in \cite{koenker2014convex} and \cite{fox2016simple}. For each $n$, let us fix a grid of support for $X^{*}_k$ with $q_n < \infty$ points, $\mathcal{S}_n = \{\bar{x}_{n,1}^{*}, \ldots, \bar{x}_{n,q_n}^{*}\}$. We can then use the following sieve space,
\begin{align*}
\mathcal{F}_{n} = \left\{ \left. x^{*} \mapsto \sum_{s=1}^{q_n} \omega_{s} \mathbf{1}\{ x^* \le \bar{x}_{n,s}^{*}\} \ \right| \ \omega \in \Delta(q_n) \right\}
\end{align*}
where $\Delta(m)$ is the $(m-1)$-dimensional unit simplex. Notice that $\mathcal{F}_n$ is the space of distributions with support contained in $\mathcal{S}_n$. As long as the support points are chosen so that $\mathcal{S}_n$ becomes dense in $\mathbb{R}$ and the number of points grows at a suitable rate, this sieve space satisfies the conditions of Theorems \ref{thm:consistency} and \ref{thm:normality}.

Importantly for practical purposes, this sieve space turns out to be particularly convenient computationally. To see this, note that under the sieve space $\mathcal{F}_{n}$ considered above,
\begin{equation*}
dF(\theta^c)=\argmax_{\omega\in\Delta(q_n)}\sum_{i=1}^n\log\sum_{s=1}^{q_n}\omega_s~\ell^c(w_i,\bar{x}_{n,s}^{*};\theta^c).
\end{equation*}
Thus the profile step reduces to a convex programming problem. This problem can be solved very efficiently and reliably using recent convex optimization algorithms available in standard softwares. For example the algorithm proposed in \cite{kim2020fast} is specialized for this setting and readily implemented in the R package \textit{mixsqp}. This allows us to calculate the profile log likelihood so the full MLE problem can be solved by maximizing this function in $\theta^c$.\footnote{In Appendix \ref{sec:diff} we show how the gradient of the profile log likelihood function can be calculated implicitly, making it feasible to use first order optimization algorithms to maximize the profile log likelihood function over $\theta^c$ efficiently.} We implement our estimator using our companion Python package \texttt{spmlex}.

\subsection{Monte Carlo simulations}
\label{sec:numer-simul}

Next, we present results from Monte Carlo simulations which illustrate the computational tractability and finite-sample performance of the proposed estimator. We focus here for simplicity on a specification with a parametric assignment model. In Appendix \ref{app:dgp_risk_aversion} we consider a specification with a nonparametric assignment model, and show that the estimator achieves similar performance.

The data generating process (DGP) used in the simulations is based on the model in Section \ref{sec:kl} with both known and unknown heterogeneity. We include two time-invariant covariates, $X = (X_1, X_2)$, where $X_1$ has a standard normal distribution and $X_2$ as a Bernoulli distribution with equal weights. We assume that $X_1$ and $X_2$ are independent from each other, and from $X^{*}$.

Assignment probabilities are derived from a model in which agents maximize the following expected utility function,
\begin{align*}
 v_t(d, X^{*}_{k}, Y^{t-1}, X, D^{t-1}) = \rho E(Y_t(d)|X^{*}_{k}, Y^{t-1}, X, D^{t-1}) + \rho\kappa \textbf{1}(d = 2) X^{*}_{k} + \nu_{t}(d), 
\end{align*}
where $Y_t(d)=\alpha_{t,d}+X_1 \gamma^{(1)}_{t,d} + X_2\gamma^{(2)}_{t,d} + X_k^{*}\lambda^{k}_{t,d}+X_{u}^{*}\lambda^{u}_{t,d}+\epsilon_{t}(d)$, where $\epsilon_{t}(d)\sim{N}(0,\sigma_{d}^{2})$, and $\{\nu_{t}(d): t = 1, 2, 3, d = 1,2\}$ are exogenous and mutually independent with a standard Extreme Value Type 1 distribution. $\rho$ is a scale parameter that affects the relative weight of preference shocks compared to systematic preferences. $\kappa$ reflects heterogeneity in preferences and/or beliefs that allows $X^*_k$ to affect choices beyond its impact on the expectation of $Y_t(d)$. We assume $X_{u}^{*}\sim N(0, \sigma_{u}^{2})$ with $\sigma_{u}^{2} = 1.5$. Finally, $X^{*}_k$ is distributed following a finite mixture of three truncated normal distributions, with means $(-1.2, 0, 1.5)$, variances $(0.2, 0.1, 0.3)$, and mixing weights $(0.4, 0.3, 0.3)$.\footnote{Each component distribution is truncated at the third standard deviation of its distribution.} The parameter values used in the simulations are reported in Appendix \ref{App:DGP_det}. This expected utility function puts a weight on the expected choice-specific potential outcomes, and adds another term that depends on $X^{*}_{k}$. This additional term can reflect biased beliefs, heterogeneity in preferences, or a combination of both.\medskip

We perform a Monte Carlo experiment, estimating the parameters of the model with $200$ simulations and sample sizes of 250, 500, 1{,}000, 2{,}000 and 4{,}000. We use the sieve MLE estimator described in Section \ref{sec:estimation}, maintaining the parametric structure on the assignment probabilities that is implied by the DGP, but estimating $F_{X^{*}_k}$ nonparametrically using the sieve space described in Section \ref{sec:estim-impl}.\footnote{Since $X_1$ is independent of $X^{*}_k$, $F_{X^{*}_k \mid X_1} = F_{X^{*}_k}$.} The sieve is chosen to have $6 n^{1/3}$ uniformly spaced support points.\footnote{This rate of growth is consistent with the rate conditions of Theorem \ref{thm:normality}, in particular Assumptions \ref{assn:rate1} and \ref{assn:rate2}. To contain the unknown bounded support of $X_{k}^*$, the grid is chosen to have minimum and maximum values at $(-0.7 n^{1/6}, 0.7 n^{1/6})$.}

With this implementation method, computation remains highly tractable for all sample sizes considered in these simulations. The average computational times to evaluate the maximum likelihood estimator are reported in Table \ref{tab:mc_comp_times} below. Run times increase with the sample size from less than half a minute (for $n=250$), to around three and a half minutes for our largest sample size ($n=4{,}000$).\medskip

\begin{table}[ht]
  \centering
    \begin{tabularx}{\textwidth}{Xrrrrr}
    \toprule

        & \multicolumn{1}{c}{$n = 250$}    
        & \multicolumn{1}{c}{$n = 500$}
        & \multicolumn{1}{c}{$n = 1{,}000$}
        & \multicolumn{1}{c}{$n = 2{,}000$} 
        & \multicolumn{1}{c}{$n = 4{,}000$} \\

    \midrule

    Time (seconds) 
    &  24 &  31 &  55 &  135 &  212  \\
    
    \bottomrule
  \end{tabularx}
    \caption{Time to compute the estimator. \footnotesize{Note: Computational times were obtained using an Intel Core i9-12900K CPU, and are computed as the average over 200 simulations.}}
    \label{tab:mc_comp_times}
\end{table}

The squared bias and variance of the sieve estimator of the finite-dimensional parameters are presented in Table \ref{tab:mc_finite_params} below. (Note that all values in this table are multiplied by $1{,}000$.) For each of the parameters, the bias becomes negligible relative to the variance as the sample size grows. The variance also declines with sample size, as expected given the consistency of our estimators, at a rate consistent with $\sqrt{n}$-convergence of the mean squared error. Overall most of the parameters are quite precisely estimated for sample sizes $n \geq 2,000$.

\begin{table}[H]
\centering
    \begin{tabular}{lrrrrrrrrrr}
    \toprule

        & \multicolumn{2}{c}{n = 250}    
        & \multicolumn{2}{c}{n = 500}
        & \multicolumn{2}{c}{n = 1,000}
        & \multicolumn{2}{c}{n = 2,000} 
        & \multicolumn{2}{c}{n = 4,000} \\
        
        & Bias$^2$ & Var
        & Bias$^2$ & Var     
        & Bias$^2$ & Var
        & Bias$^2$ & Var
        & Bias$^2$ & Var \\
    
    \midrule
        
    $\alpha_{1, 2}$
        & 71.72 & 87.92 & 34.06 & 60.97 & 12.91 & 47.13 & 0.73 & 19.02 & 0.04 & 5.70  \\
    $\alpha_{2, 1}$
        & 0.15 & 27.98 & 0.26 & 12.38 & 0.12 & 7.39 & 0.00 & 2.88 & 0.01 & 1.38  \\
    $\alpha_{2, 2}$
        & 73.52 & 108.96 & 34.18 & 74.42 & 12.41 & 57.19 & 0.46 & 25.80 & 0.03 & 8.11  \\
    $\alpha_{3, 1}$
        & 0.01 & 36.56 & 0.45 & 13.82 & 0.20 & 5.31 & 0.00 & 2.24 & 0.01 & 0.96  \\
    $\alpha_{3, 2}$
        & 47.84 & 163.16 & 32.09 & 82.42 & 12.03 & 62.31 & 0.59 & 25.98 & 0.04 & 7.32  \\
    $\gamma_{1,1}^{(1)}$
        & 0.51 & 10.08 & 0.40 & 5.22 & 0.14 & 3.17 & 0.02 & 1.49 & 0.00 & 0.72  \\
    $\gamma_{1,2}^{(1)}$
        & 0.85 & 15.22 & 0.30 & 6.75 & 0.05 & 3.35 & 0.01 & 1.74 & 0.00 & 0.80  \\
    $\gamma_{2,1}^{(1)}$
        & 0.84 & 16.30 & 0.66 & 7.86 & 0.39 & 4.46 & 0.04 & 1.85 & 0.01 & 0.80  \\
    $\gamma_{2,2}^{(1)}$
        & 1.38 & 20.81 & 0.60 & 12.06 & 0.09 & 5.62 & 0.00 & 2.69 & 0.01 & 1.21  \\
    $\gamma_{3,1}^{(1)}$
        & 0.41 & 9.30 & 0.24 & 3.88 & 0.16 & 1.89 & 0.03 & 1.03 & 0.01 & 0.57  \\
    $\gamma_{3,2}^{(1)}$
        & 0.38 & 19.19 & 0.40 & 9.11 & 0.08 & 4.20 & 0.01 & 2.10 & 0.00 & 0.86  \\
    $\gamma_{1,1}^{(2)}$
        & 0.61 & 58.91 & 0.36 & 23.24 & 0.36 & 11.16 & 0.03 & 4.77 & 0.00 & 2.29  \\
    $\gamma_{1,2}^{(2)}$
        & 0.19 & 46.66 & 0.22 & 25.40 & 0.02 & 11.16 & 0.00 & 5.12 & 0.01 & 2.61  \\
    $\gamma_{2,1}^{(2)}$
        & 0.01 & 40.41 & 0.00 & 19.84 & 0.00 & 9.05 & 0.00 & 4.35 & 0.04 & 2.48  \\
    $\gamma_{2,2}^{(2)}$
        & 0.04 & 57.76 & 0.05 & 26.57 & 0.00 & 12.37 & 0.00 & 6.76 & 0.01 & 3.29  \\
    $\gamma_{3,1}^{(2)}$
        & 0.50 & 40.19 & 0.08 & 19.94 & 0.02 & 7.64 & 0.00 & 3.94 & 0.02 & 2.05  \\
    $\gamma_{3,2}^{(2)}$
        & 0.10 & 65.65 & 0.33 & 32.11 & 0.01 & 15.18 & 0.02 & 7.11 & 0.00 & 3.44  \\
    $\lambda_{1, 1}^{k}$
        & 2.75 & 27.52 & 1.70 & 12.89 & 0.62 & 7.27 & 0.01 & 3.68 & 0.00 & 1.47  \\
    $\lambda_{2, 1}^{k}$
        & 1.15 & 25.98 & 0.56 & 10.83 & 0.23 & 4.78 & 0.00 & 2.59 & 0.00 & 1.09  \\
    $\lambda_{2, 2}^{k}$
        & 0.87 & 10.98 & 0.25 & 5.82 & 0.07 & 2.65 & 0.01 & 1.38 & 0.00 & 0.74  \\
    $\lambda_{3, 1}^{k}$
        & 3.99 & 33.66 & 0.87 & 13.72 & 0.18 & 5.68 & 0.00 & 3.07 & 0.00 & 1.33  \\
    $\lambda_{3, 2}^{k}$
        & 5.70 & 36.86 & 0.67 & 12.56 & 0.22 & 5.30 & 0.01 & 2.41 & 0.01 & 1.08  \\
    $\lambda_{1, 2}^{u}$
        & 0.98 & 13.94 & 0.31 & 4.73 & 0.17 & 2.44 & 0.01 & 1.33 & 0.00 & 0.61  \\
    $\lambda_{2, 1}^{u}$
        & 0.04 & 8.32 & 0.03 & 5.14 & 0.04 & 1.95 & 0.01 & 1.00 & 0.00 & 0.48  \\
    $\lambda_{2, 2}^{u}$
        & 1.48 & 14.88 & 0.49 & 6.22 & 0.13 & 3.32 & 0.01 & 1.52 & 0.00 & 0.64  \\
    $\lambda_{3, 1}^{u}$
        & 0.45 & 9.91 & 0.09 & 5.00 & 0.06 & 2.19 & 0.03 & 0.97 & 0.02 & 0.47  \\
    $\lambda_{3, 2}^{u}$
        & 0.11 & 21.92 & 0.10 & 8.90 & 0.11 & 4.15 & 0.00 & 2.14 & 0.01 & 0.94  \\
    $\sigma^2(1)$
        & 0.45 & 2.48 & 0.09 & 1.24 & 0.03 & 0.67 & 0.01 & 0.30 & 0.00 & 0.14  \\
    $\sigma^2(2)$
        & 1.23 & 4.45 & 0.24 & 2.24 & 0.03 & 1.06 & 0.02 & 0.70 & 0.01 & 0.33  \\
    $\sigma^2_u$
        & 0.02 & 72.90 & 0.05 & 41.17 & 0.04 & 17.91 & 0.01 & 9.34 & 0.01 & 4.33  \\
        
    \bottomrule
  \end{tabular}
    \caption{Simulation results for the estimation of the finite dimensional parameters. \footnotesize{Note: `Bias$^2$' and `Var' refer to the average empirical squared bias and variance scaled by $1{,}000$, respectively, computed over 200 simulations.}}
    \label{tab:mc_finite_params}
\end{table}

Next, we present results for the nonparametric estimator of the distribution of known unobserved heterogeneity $X^{*}_k$, focusing on its quantiles $q_{\alpha}[{X^{*}_k}]$. For each value of $\alpha \in [0,1]$, we calculate the mean and the $5$th and $95$th percentile of the simulated distribution of the estimator of $q_{\alpha}[{X^{*}_k}]$. The results are presented in Figure \ref{fig:mc_factor_5_95} below. The red line shows the quantile function of the true distribution of $X^{*}_k$, while the blue lines that closely follow the red line are the mean of the simulated distribution of the quantile estimators for each sample size. Darker blue lines represent larger sample sizes. The blue lines above and below the quantile function are the $95$th and $5$th percentiles of the simulated distribution of the quantile estimators.

\begin{figure}[h]
    \centering
    \input{paper-inputs/quantiles_95_r031.tex}
    \caption{Quantiles of the estimator of $q_{\alpha}[X^{*}_k]$. \footnotesize{Note: The red line shows the true distribution of $X^{*}_k$. The blue lines show the mean, and the $5$th and $95$th percentiles of the simulated distribution of the estimator of $q_{\alpha}[X^{*}_k]$ for each sample size.}}
    \label{fig:mc_factor_5_95}
\end{figure}

The results indicate that the bias of the quantile estimators becomes negligible in moderate sample sizes. The estimator also broadly captures the shape of the true distribution of $X^{*}_k$. Besides, even though the simulated distribution is still relatively disperse for the sample sizes we consider in these simulations, the estimator also appears to converge toward the true distribution as the sample size grows.

Finally, we conclude this section by considering the plug-in estimator for one of the functionals discussed in Section \ref{sec2} and Appendix \ref{App:var_dec}. Namely, we focus on the decomposition of the present value of a stream of outcomes into known and unknown components at $t = 1$. Setting the discount rate equal to $0.95$, the variance of the unknown and known components corresponding to the two terms in Equation \eqref{eq:vardecom-t0} in Appendix \ref{App:var_dec} are, for a given choice sequence $d^3$,\footnote{The sum of these two terms is the variance of $\sum_{t=1}^3 (.95)^{1-t} Y_t(d_t)$, which is the present value of $(Y_1(d_1), Y_2(d_2), Y_3(d_3))$ at period $1$. This is a special case of the class of weighted sums of potential outcomes considered in Appendix \ref{App:var_dec}, where the weights are $\omega_t = (.95)^{1 - t}$, and the choice sequence is $d^3$. The two terms correspond to the two terms of Equation \eqref{eq:vardecom-t0} with $\omega_t$ defined as above.}
\begin{align}\label{decomp-mc}
\begin{split}
 V^u_{d^3}\coloneqq  &  \sigma_u^2 \sum_{1 \le t_1, t_2 \le 3} (.95)^{t_1 + t_2 -2} \lambda^u_{t_1, d_{t_1}}\lambda^u_{t_2, d_{t_2}}  + \sum_{1 \le t \le 3} (.95)^{2t-2} \sigma^2_{d_t}, \\
   V^k_{d^3}\coloneqq  &  \V(X^{*}_k) \sum_{1 \le t_1, t_2 \le 3} (.95)^{t_1 + t_2-2} \lambda^k_{t_1, d_{t_1}}\lambda^k_{t_2, d_{t_2}}.
\end{split}
\end{align}

We estimate these functionals, which involve both the finite-dimensional parameters and $F_{X^{*}_k}$, using the plug-in estimator described in Section \ref{sec:estimation}. The results are presented in Table \ref{tab:mc_functional}. For moderately small sample sizes starting with $n=500$, the squared bias is generally negligibly small relative to the variance. Besides, variance (and MSE) decrease with the sample sizes, at a rate that appears to be consistent with a $\sqrt{n}$-convergence rate.  

\begin{table}[h]
    \resizebox{\textwidth}{!}{
    \centering
      \begin{tabular}{lrrrrrrrrrr}
    \toprule

        Parameter & \multicolumn{2}{c}{n = 250}    
        & \multicolumn{2}{c}{n = 500}
        & \multicolumn{2}{c}{n = 1,000}
        & \multicolumn{2}{c}{n = 2,000} 
        & \multicolumn{2}{c}{n = 4,000} \\

         & Bias$^2$ & Var
        & Bias$^2$ & Var
        & Bias$^2$ & Var
        & Bias$^2$ & Var
        & Bias$^2$ & Var\\
    
    \midrule
        
    $V^k_{(1, 1, 1)}$
        & 0.01 & 0.99 & 0.00 & 0.45 & 0.00 & 0.21 & 0.00 & 0.14 & 0.00 & 0.07  \\
    $V^u_{(1, 1, 1)}$ 
        & 0.00 & 3.06 & 0.00 & 1.51 & 0.00 & 0.68 & 0.00 & 0.33 & 0.00 & 0.15  \\
    $V^k_{(1, 1, 2)}$ 
        & 0.00 & 1.46 & 0.01 & 0.70 & 0.00 & 0.38 & 0.00 & 0.23 & 0.00 & 0.09  \\
    $V^u_{(1, 1, 2)}$ 
        & 0.00 & 2.32 & 0.00 & 1.13 & 0.00 & 0.52 & 0.00 & 0.27 & 0.00 & 0.12  \\
    $V^k_{(1, 2, 1)}$ 
        & 0.32 & 1.77 & 0.13 & 0.93 & 0.04 & 0.53 & 0.00 & 0.28 & 0.00 & 0.11  \\
    $V^u_{(1, 2, 1)}$ 
        & 0.03 & 1.72 & 0.00 & 0.85 & 0.00 & 0.37 & 0.00 & 0.19 & 0.00 & 0.09  \\
    $V^k_{(1, 2, 2)}$ 
        & 0.21 & 3.13 & 0.16 & 1.53 & 0.05 & 0.88 & 0.00 & 0.41 & 0.00 & 0.15  \\
    $V^u_{(1, 2, 2)}$ 
        & 0.01 & 1.20 & 0.01 & 0.60 & 0.00 & 0.28 & 0.00 & 0.15 & 0.00 & 0.06  \\
    $V^k_{(2, 1, 1)}$ 
        & 0.24 & 1.49 & 0.07 & 0.82 & 0.02 & 0.36 & 0.00 & 0.22 & 0.00 & 0.10  \\
    $V^u_{(2, 1, 1)}$ 
        & 0.03 & 1.75 & 0.00 & 0.85 & 0.01 & 0.36 & 0.00 & 0.16 & 0.00 & 0.08  \\
    $V^k_{(2, 1, 2)}$ 
        & 0.15 & 2.43 & 0.08 & 1.13 & 0.03 & 0.56 & 0.00 & 0.32 & 0.00 & 0.14  \\
    $V^u_{(2, 1, 2)}$ 
        & 0.01 & 1.23 & 0.01 & 0.60 & 0.01 & 0.27 & 0.00 & 0.13 & 0.00 & 0.07  \\
    $V^k_{(2, 2, 1)}$ 
        & 1.00 & 3.04 & 0.30 & 1.56 & 0.07 & 0.73 & 0.00 & 0.38 & 0.00 & 0.17  \\
    $V^u_{(2, 2, 1)}$ 
        & 0.10 & 1.10 & 0.02 & 0.45 & 0.01 & 0.19 & 0.00 & 0.09 & 0.00 & 0.05  \\
    $V^k_{(2, 2, 2)}$ 
        & 0.45 & 5.84 & 0.21 & 2.77 & 0.04 & 1.56 & 0.00 & 0.76 & 0.00 & 0.33  \\
    $V^u_{(2, 2, 2)}$ 
        & 0.06 & 0.79 & 0.03 & 0.32 & 0.01 & 0.17 & 0.00 & 0.09 & 0.00 & 0.04  \\
        
    \bottomrule
  \end{tabular}}
      \caption{Simulation results for estimation of $V^p_{d^3}$ for $p=k,u$ as defined in Equation \eqref{decomp-mc}. \footnotesize{Note: `Bias$^2$' and `Var' refer to the average empirical squared bias and variance respectively, computed over 200 simulations.}}
      \label{tab:mc_functional}
  \end{table}

\FloatBarrier

\section{Empirical illustration}\label{sec6}

In this section, we illustrate the empirical framework developed above with an application to ability learning in the context of occupational choice, revisiting a question that has attracted significant interest in labor economics \citep[see, e.g.,][]{Miller84,AG12,James12,Papageorgiou14,Pastorino15,AAMR25}.

\subsection{Data and descriptive overview}

We use data from the National Longitudinal Survey of Youth 1997 (NLSY97). This is a nationally representative U.S. sample of individuals born between 1980 and 1984. We restrict the sample to white men who worked full-time between the ages of 27 and 32.\footnote{With this sample restriction, we use data from the 2007-2015 waves of the NLSY97.} With the demographic restrictions, we have a sample size of 2,031 individuals, and after restricting to people who worked full time continuously between ages 27 to 32, we obtain a sample of 965 individuals.\footnote{The sample sizes under these restrictions can be seen in Appendix Table \ref{tab:sample_size}. }$^{,}$\footnote{Full-time work status is calculated based on work history in October of each year, and requires at least 35 hours per week and four weeks worked during that month.} 

For our application, we use primarily data on labor market experience, labor force status, hourly wages, and census occupation codes. Our measurement of wages is an individual's average log hourly wages over a two-year period. Occupations are classified into high-skill or low-skill occupations based on the mean college completion rate of individuals working in that occupation. High-skill occupations are defined as those in which more than 50\% of individuals employed in that occupation have a college degree.\footnote{This classification follows the approach in \cite{AAMR25}, and uses the current population survey (CPS) to calculate the mean college graduation rate within each 3-digit occupation.}  

Table \ref{tab:descriptives} below shows the mean and standard deviation of various characteristics conditional on the number of periods worked in a high-skill occupation. A couple of comments are in order. There is a monotonic pattern across all variables in the number or periods worked in a high-skill occupation. Notably, there is a sharp increase in the share of college graduates among people who work in a high-skill occupation for at least one period, reflecting in part the effect of getting a college degree on the likelihood of finding a high-skill job. The table also shows a more continuous increasing relationship between the number of periods worked in a high-skill occupation and the education level of the individual's mother, family income, and the Armed Forces Qualification Test (AFQT) score, which may all be considered as correlates of the individual's underlying ability. As expected, log wages, including within each occupation, increase with the number of periods worked in the high-skill occupation. For example, the average log wage for people who work all three periods in the high-skill occupation is 0.36 log points higher than the average log wage for individuals who work all three periods in the low-skill occupation. Moreover, average wages in each occupation increase with additional high-skill experience: average low-skill wages are 0.07 log points higher (from 0 to 1-2 periods in high-skill occupation), and average high-skill wages are 0.16 log points higher (from 1-2 to 3 periods in high-skill occupation).
 
Taken together, these descriptive patterns are consistent with various potential mechanisms, including selection across occupations based on their ability. In particular, they are not directly informative about the role played by sorting on the portion of ability that may be revealed over time, rather than initially known by the workers. Our empirical framework allows us to identify, from the observed occupational choices and realized wages, the role played by learning about one's ability in this context.

\begin{table}[H]
  \centering
\begin{tabular}{lcccccc}
    \toprule
    & \multicolumn{6}{c}{Periods Worked in High-Skill Occupation} \\[.4em]
    & \multicolumn{2}{c}{0} & \multicolumn{2}{c}{1-2} & \multicolumn{2}{c}{3} \\[.4em]
    \cmidrule(lr){2-3} \cmidrule(lr){4-5} \cmidrule(lr){6-7}
    & Mean & S. D. & Mean & S. D. & Mean & S. D. \\[.4em]
    \midrule
    College graduate (\%)
        &  0.15 
        &   
        &  0.53 
        & 
        &  0.74 
        &  \\[.4em]
    Mother college graduate (\%)
        &  0.21 
        & 
        &  0.41 
        & 
        &  0.55 
        &  \\[.4em]
    Family income (,000s)
        &  71.5 
        &  53.6 
        &  88.0 
        &  60.8 
        &  107.0 
        &  78.6  \\[.4em]
    AFQT    
        &  0.12 
        &  0.90 
        &  0.56 
        &  0.68 
        &  1.01 
        &  0.54  \\[.4em]
    Log Wage
        &  2.48 
        &  0.46 
        &  2.62 
        &  0.49 
        &  2.84 
        &  0.54  \\[.4em]
    Log Wage (low-skill)
        &  2.48 
        &  0.46 
        &  2.55 
        &  0.59 
        & 
        &  \\[.4em]
    Log Wage (high-skill)
        & 
        & 
        &  2.68 
        &  0.52 
        &  2.84 
        &  0.54  \\[.4em]
        \midrule
    Nb. Individuals
        & \multicolumn{2}{c}{ 545 }
        & \multicolumn{2}{c}{ 130 }
        & \multicolumn{2}{c}{ 201 } \\
    \bottomrule
\end{tabular}
\caption{Descriptive statistics of NLSY subsample of white men who worked full time between ages 27 and 32. \footnotesize{Note: Low-skill log wages are defined as the average of log hourly wages in each period when an individual worked in the low-skill occupation. For individuals who worked in the low-skill occupation each period, this coincides with the observed log wage; for those who all periods in the high-skill occupations, this is not observed, and for those who worked in both occupations, it is their average log wage only for these periods when they worked in a low-skill occupation. High-skill log wages are defined analogously.}}
\label{tab:descriptives}
\end{table}

\subsection{Model set-up}

We divide the early career into three periods, based on the individual's age. The periods are each two years long, spanning age 27 to 32. In each period $t \in\{1,2,3\}$, individuals work in the labor market in either the high- or low- skill occupation and earn a wage.\footnote{In an appendix, we consider an extended specification of this model that includes college graduation as a covariate in the wage equation. Our main findings remain qualitatively similar for this extended specification. We refer the reader to Section \ref{sec:appendix_extended_model} for further details.} Their potential average log wage over each two-year period is denoted by $Y_t(1)$ or $Y_t(0)$, for the high- and low-skill occupation, respectively. We assume that potential log wages follow an interactive fixed effects model as in Equation \eqref{eq:outcome_model_intro}, and we maintain Assumptions \ref{assn:kl_norm} and \ref{assn:kl_support} on the distributions of $(X_k^*, X_u^*, \epsilon)$. That is, for $d \in \{0, 1\}$, potential log wages are given by
$$
Y_{t}(d) = \beta_{t,d}+X^{*}_{k} \lambda_{t,d}^k + X^{*}_{u} \lambda_{t,d}^u +\epsilon_{t}(d),
$$
where $X_u^* \sim N(0, \sigma_u^2)$ and $\epsilon_{t,d} \sim N(0, \sigma_{t,d}^2)$.\footnote{We impose the normalization $\beta_{1,0} = 0$ in accordance with Assumption \ref{assn:kl_norm}.} 

Consistent with our choice framework introduced in Section \ref{sec2} (see Eq.\eqref{eq:choice-ind} in particular), the time-varying conditional occupational choice probabilities are allowed to depend arbitrarily on $X_k^*$ and past outcomes and choices. We denote these as:
$$h_t ((1,D^{t-1}), Y^{t-1}, X_k^*) = P(D_t = 1 \mid Y^{t-1}, D^{t-1},  X_k^*).$$ 
This choice model is very flexible, and accommodates several different factors that have been shown in the literature to influence occupational choices. These include, among others, the individuals' beliefs (correct or biased) about their potential wages in each occupation, their preferences over non-pecuniary aspects of occupations, and search or informational frictions. The inclusion of the latent term $X_k^*$ in the choice model also accommodates the realistic scenario in which the researcher does not directly observe all the factors, such as unobserved worker-specific productivity, that jointly affect potential wages and occupational choices. Importantly, we also allow for a portion of individual productivity, $X^*_u$, to be unknown to the workers and thus excluded from individual choices.

\subsection{Estimation}

We estimate the model by implementing the sieve MLE estimator described in Section \ref{sec:estimation}, using a flexible logit model for the CCP function $h_t$ and the sieve space for the distribution of $X^*_k$, $F_{X_k^*}$, described in Section \ref{sec:estim-impl} with a grid of 56 equally spaced support points. We implement the estimator using our companion Python package \texttt{spmlex}.

Specifically, we estimate the CCPs using the functional form $h_t ((1,D^{t-1}), Y^{t-1}, X_k^*) = \Lambda (\phi_t(X_k^*, Y^{t-1}, D^{t-1}))$ where $\Lambda(u) = (1 - e^{-u})^{-1}$ and,\footnote{We are using the convention here that $\sum_{u=1}^v f_u = 0$ for $v < u$.} 
$$
\phi_t (X_k^*, Y^{t-1}, D^{t-1}) = \sum_{d^{t-1} \in \{0, 1\}^{t-1}} \indic(D^{t-1} = d^{t-1}) \left( \pi_{0, t, d^{t-1}} + \sum_{s = 1}^{t-1} \pi_{s, t, d^{t-1}} Y_s + \pi_{t, t, d^{t-1}} X_k^* \right).
$$
This specification nests a standard choice model in which individuals make a choice which depends on the expectation of the potential outcome for each choice and a preference shock with an extreme value type 1 distribution. However, it is flexible enough to allow the relative weights on $X_k^*$ or on past outcomes to be different from the coefficients derived from a standard Bayesian updating rule. In particular, it allows for biased beliefs, as well as for non-pecuniary preferences or search frictions that might be correlated with the expected outcomes.

Finally, even though these are not needed for identification, we impose some additional restrictions on the parameters of the outcome model to improve the precision of our estimator. Namely, for the sake of parsimony, we restrict the idiosyncratic error variances to be time-invariant and also assume that the factor loadings associated with the know and unknown heterogeneity component have a linear time trend. We then estimate the model by implementing the profile likelihood procedure described in Section \ref{sec:estim-impl}.\footnote{In order to check for local optima, we re-initialize the optimization algorithm at 20 different starting values. These 20 starting values are chosen as follows. First, we draw 5{,}000 parameter values from a grid centered at the estimated parameter values. Then, we bin parameter vectors by the decile of their Euclidean distance from the estimated parameter values, and choose those with the highest and lowest likelihood within each bin.}

\subsection{Model fit}

We begin by discussing the model fit before turning to the estimation results. We focus on the outcomes and report in Table \ref{tab:model_fit} below the mean and autocovariance of log wages across all three periods. Each panel displays these moments conditional on the number of periods individuals work in the high-skill occupation (zero, one or two, and three for Panels A, B and C, respectively), comparing the raw sample moments estimated from the data (``Data'') with the moments implied by the model at the estimated parameter values (``Est.'').\footnote{The latter are calculated by simulating the model $10{,}0000$ times at the estimated parameter values, computing the empirical means and covariances of the simulated data.}

A key takeaway from this table is that the estimated model is generally able to match these moments well. In Panel A, we see that all the moments implied by the model are within $0.01$ of the raw sample moments for workers employed in the low-skill occupation only. As shown in Panels B and C, there are slightly more differences between the raw and simulated moments conditional on having worked in a high-skill occupation. At any rate, these results indicate that, despite its parsimony, the estimated potential outcome model is able to satisfactorily capture the mean and dispersion of the realized log wages, along with their dependence over time.

\begin{table}[H]
  \centering
  \resizebox{.6\textwidth}{!}{\begin{tabular}{llrrrrrr}
    \toprule
        & & \multicolumn{2}{c}{$Y_1$}
        & \multicolumn{2}{c}{$Y_2$}
        & \multicolumn{2}{c}{$Y_3$} \\
        & & Est. & Data
        & Est. & Data
        & Est. & Data \\    
    \midrule    
        \multicolumn{8}{l}{A. No period in high-skill occupation} \\
        \addlinespace[2ex]
        \multicolumn{8}{l}{\textit{Mean}} \\
        \addlinespace[1ex]
        & & 2.45 
        & 2.45 
        & 2.51 
        & 2.52 
        & 2.57 
        & 2.57 \\
        \addlinespace[1ex]            
        \multicolumn{8}{l}{\textit{Covariance Matrix}} \\
        \addlinespace[1ex]        
        & $Y_1$
        & 0.18 
        & 0.17 
        & 0.15 
        & 0.14 
        & 0.14 
        & 0.13        
        \\
        & $Y_2$
        & \textemdash
        & \textemdash        
        & 0.18 
        & 0.19 
        & 0.17 
        & 0.17        
        \\
        & $Y_3$
        & \textemdash
        & \textemdash   
        & \textemdash
        & \textemdash               
        & 0.22 
        & 0.21        
        \\        
    \midrule    
        \multicolumn{8}{l}{B. Some periods in high-skill occupation} \\
        \addlinespace[2ex]
        \multicolumn{8}{l}{\textit{Mean}} \\
        \addlinespace[1ex]
        && 2.58 
        & 2.58 
        & 2.65 
        & 2.68 
        & 2.82 
        & 2.80 \\
        \addlinespace[1ex]            
        \multicolumn{8}{l}{\textit{Covariance Matrix}} \\
        \addlinespace[1ex]        
        & $Y_1$
        & 0.18 
        & 0.21 
        & 0.12 
        & 0.14 
        & 0.13 
        & 0.12        
        \\
        & $Y_2$
        & \textemdash
        & \textemdash        
        & 0.18 
        & 0.20 
        & 0.15 
        & 0.13        
        \\
        & $Y_3$
        & \textemdash
        & \textemdash   
        & \textemdash
        & \textemdash               
        & 0.22 
        & 0.19        
        \\    
    \midrule    
        \multicolumn{8}{l}{C. All periods in high-skill occupation} \\
        \addlinespace[2ex]
        \multicolumn{8}{l}{\textit{Mean}} \\
        \addlinespace[1ex]
        && 2.78 
        & 2.76 
        & 2.91 
        & 2.91 
        & 3.01 
        & 3.00 \\
        \addlinespace[1ex]            
        \multicolumn{8}{l}{\textit{Covariance Matrix}} \\
        \addlinespace[1ex]        
        & $Y_1$
        & 0.24 
        & 0.26 
        & 0.16 
        & 0.16 
        & 0.16 
        & 0.17        
        \\
        & $Y_2$
        & \textemdash
        & \textemdash        
        & 0.23 
        & 0.21 
        & 0.16 
        & 0.19        
        \\
        & $Y_3$
        & \textemdash
        & \textemdash   
        & \textemdash
        & \textemdash               
        & 0.25 
        & 0.26        
        \\         
    \bottomrule
  \end{tabular}}
    \caption{Model Fit.}
    \label{tab:model_fit}
\end{table}

\subsection{Estimation results}

We first discuss the determinants of sorting across occupations, with a focus on the role played by latent productivity. We then discuss the importance of heterogeneity versus uncertainty and its evolution over the course of the early career.  

\paragraph{Selection across occupations}

Occupational choices are determined conditional on the information set of individuals, which includes latent individual productivity $X_k^*$ along with past choices and outcomes. We focus in the following on how the distribution of $X_k^*$ varies across occupational choice sequences.
\medskip

\begin{figure}[H]
  \centering 
  \resizebox{.8\textwidth}{!}{\input{paper-inputs/6-23-2025/selection_r66-updated.pgf}}
  \caption{Selection into high-skill occupation. \footnotesize{Note: Each line represents the estimated CDF of $\lambda_{1, 1}^k X_k^*$, conditional on the number of periods an individual works in the high-skill occupation.}}
  \label{fig:xk_selection}
\end{figure}

Figure \ref{fig:xk_selection} reports the estimated distribution of $(\lambda_{1, 1}^k X_k^*)$ conditional on the number of periods an individual works in the high-skill occupation.\footnote{Since $X_k^*$ is a latent variable, it does not have natural units. We choose the scale to be units of outcome (log wages) in the high-skill occupation, in period 1.} A key takeaway is that the distribution of $X_k^*$ among individuals who always work in the high-skill occupation stochastically dominates the corresponding distributions among individuals who work in a low-skill occupation in at least one of the periods. This points to positive selection, whereby high-productivity individuals are substantially more likely to work in a high-skill occupation. In particular, $94\%$ of individuals who work in a low-skill occupation for at least one period have a realization of $X_k^*$ below the mean, compared to $11\%$ of those who always work in a high-skill occupation. 

Another pattern of note related to these results is that selection has an asymmetric impact on the dispersion of wages conditional on occupational choice. Namely, because the unconditional distribution of $X_k^*$ exhibits greater dispersion in the upper tail than in the lower tail, the selected distribution of $X_k^*$ conditional on working in a low-skill occupation is less dispersed than the distribution of $X_k^*$ conditional on working in a high-skill occupation. This highlights how the specific shape of the distribution of $X_k^*$ interacts with sorting across occupations to produce nuanced implications about the dispersion of realized wages. The flexibility of our framework in these two dimensions allows us to capture such patterns.

\paragraph{Decomposition of variance: Heterogeneity vs. uncertainty}
\label{sec:var}

We now use our framework to decompose the variance of future wages into components that are forecastable and unforecastable by the agents at the time of their decisions.

Specifically, we focus on the discounted value of log wages of the two later periods of our analysis ($t \in \{2,3\}$). We denote by $\overline{Y}(d_2)$ the discounted value of potential log wages associated with occupation $d_2$, namely $\overline{Y}(d_2) = \sum_{t=2}^3 (1 -\rho)^{t - 2} Y_t(d_2)$, for $d_2 \in \{0, 1\}$, where we set the discount factor $\rho = 0.05$. We consider two alternative decompositions. The first one is given by:
\begin{align}
\V(\overline{Y}(d_2)) 
&= \V(E(\overline{Y}(d_2) | X_k^*))
+ E(\V(\overline{Y}(d_2) | X_k^*)). \label{eq:app-vardecomp-1} 
\end{align}
Equation \eqref{eq:app-vardecomp-1} decomposes the variance of $\overline{Y} (d_2)$ into a term that corresponds to the component of $\overline{Y} (d_2)$ that is forecastable by the agents before making their occupational choice (which we refer to as period 0), $E(\overline{Y}(d_2) \mid X_k^*)$, and a term that corresponds to the portion of $\overline{Y} (d_2)$ that is unforecastable by the agents, $\overline{Y}(d_2) - E(\overline{Y}(d_2) \mid X_k^*)$. Using the terminology introduced earlier in the paper, the first and second terms of Equation \eqref{eq:app-vardecomp-1} capture the heterogeneity and uncertainty components of the variance decomposition, respectively. 

The second decomposition we consider is given by:
\begin{multline}
\V(\overline{Y}(d_2) | D_1 = d_1)
= \V(E(\overline{Y}(d_2) | X_k^*, Y_1, D_1 = d_1)|D_1 = d_1) \\ 
\qquad + E(\V(\overline{Y}(d_2) | X_k^*, Y_1, D_1 = d_1 ) | D_1 = d_1). \label{eq:app-vardecomp-2}
\end{multline}

Equation \eqref{eq:app-vardecomp-2} is a period 1 analogue of the period 0 decomposition (Eq. \eqref{eq:app-vardecomp-1}). For each period 1 choice (i.e., $D_1=0$ and $D_1=1$), it decomposes the variance of $\overline{Y} (d_2)$ into a term that is forecastable at the end of period 1, $E(\overline{Y}(d_2) | X_k^*, Y_1, D_1)$, and a part that is not.

Table \ref{tab:emp-var-decomp} presents estimates of these variance decompositions for each stream of potential wages (i.e., $\overline{Y}(1)$ and $\overline{Y}(0)$). For each decomposition, we present estimates of the total variance and the share of the total variance that is forecastable to the agent. For example, for the period 0 decomposition (Eq. \eqref{eq:app-vardecomp-1}), the total variance is $\V(\overline{Y}(d_2))$ and the share forecastable is $\frac{\V(E(\overline{Y}(d_2) | X_k^*))}{\V(\overline{Y}(d_2))}$. Results are presented with 95\% bootstrap confidence intervals.\footnote{In Monte Carlo experiments, bootstrap confidence intervals for the components of the variance decompositions in a DGP with this sample size exhibited close to nominal coverage (see Section~\ref{app:boot} in the appendix). A formal investigation of bootstrap validity is left for future research.}

\begin{table}[h]
  \resizebox*{\textwidth}{!}{\begin{tabular}{lcccc}
    \toprule
    & \multicolumn{2}{c}{$\overline{Y}(1)$} & \multicolumn{2}{c}{$\overline{Y}(0)$} \\
    Decomposition & Total Variance & Share Forecastable & Total Variance & Share Forecastable \\ 
    \midrule    
    Equation \eqref{eq:app-vardecomp-1}
        & 0.66 & 0.12 & 1.07 & 0.43 \\   
        & (0.52, 0.77)  
        & (0.07, 0.23)  
        & (0.73, 3.93)  
        & (0.20, 0.85)  \\[0.5em]       
    Equation \eqref{eq:app-vardecomp-2}, $d_1=0$
        & 0.57 & 0.65 & 0.64 & 0.80 \\   
        & (0.44, 0.69)  
        & (0.59, 0.74)  
        & (0.56, 0.75)  
        & (0.76, 0.85)  \\[0.5em]       
    Equation \eqref{eq:app-vardecomp-2}, $d_1=1$
        & 0.70 & 0.53 & 1.27 & 0.76 \\   
        & (0.57, 0.81)  
        & (0.41, 0.69)  
        & (0.80, 5.40)  
        & (0.64, 0.96)  \\[0.5em]       
    \bottomrule
\end{tabular}}
  \caption{Forecastability of Discounted Future Earnings. \footnotesize{Note: Each row reports the total variance of discounted future potential log wages in high and low-skill occupations and the share that is predictable at time $t$ conditional on a sequence of prior choices. The first row is the variance decomposition at period 0 before the first choices are made (i.e., equation \eqref{eq:app-vardecomp-1}). The second and third rows are the variance decomposition conditional on the first occupational choice  (i.e., equation \eqref{eq:app-vardecomp-2} for $d_1=0$ and $d_1=1$, respectively). The total variance is the variance of $\overline{Y}(d_2)$, conditional having made the choice $D^t$, which can therefore be a selected sample. The share forecastable is the ratio of the forecastable variance (including both the variance coming from $X_k^*$ and the posterior mean of $X_u^*$ after observing $D^t$) to the total variance. Bootstrap 95\% confidence intervals are given in parentheses.}}
  \label{tab:emp-var-decomp}
\end{table}

Three key results emerge from Table \ref{tab:emp-var-decomp}. First, a significantly larger share of variance in future earnings is initially unforecastable in the high-skill occupation compared to the low-skill occupation. In the first row of Table \ref{tab:emp-var-decomp}, we see that only $12\%$ of the variance in potential future wages in the high-skill occupation is forecastable compared to $43\%$ in the low-skill occupation. In that sense, uncertainty appears to play a particularly important role in accounting for the dispersion of discounted future wages in high-skill occupations.  

Second, much of this uncertainty is revealed as individuals accumulate work experience. This supports the idea that workers learn about their own productivity from their wages. This finding is consistent with earlier evidence that points to the importance of ability learning in the workforce \citep{Miller84,AG12,Pastorino22}. Also evident from Table  \ref{tab:emp-var-decomp} is that individuals appear to learn quite quickly about their future potential wages. Namely, after one (two-year) period of work in a high-skill occupation, the forecastable share of variance of future potential earnings in high-skill occupations increases sharply from 12\% to 53\% (Rows 1 and 3 in Table \ref{tab:emp-var-decomp}). Similarly, the forecastable share of variance of future potential wages in the low-skill occupation increases from 43\% to 80\%, after one period of work in low-skill occupations. Interestingly, there is a similarly large increase in the forecastable share of variance of future potential wages after a period of work in the other occupation. For example, after a period of work in high-skill occupations, the forecastable share of variance in future potential wages in low-skill occupations increases from $43\%$ to $76\%$.\footnote{In Appendix \ref{sec:appendix_extended_model} we calculate these variance decompositions in an extended model that includes college graduation as a covariate, and obtain qualitatively similar results.}

Finally, while the two key results highlighted above point to the importance of uncertainty and learning in this context, initially known latent productivity ($X^*_k$) plays an important role as well. Beyond the fact that, initially, more than 40\% of the variance of future wages in low-skill occupations is driven by the known portion of unobserved heterogeneity, this component is also central to understanding the overall dispersion of wages. In particular, in low-skill occupations, the total variance of future potential wages decreases from $1.07$ to $0.64$ after the first period of work. Key to this decrease in dispersion is selection based on the initially known heterogeneity component $X^*_k$. As illustrated in Figure \ref{fig:xk_selection}, the distribution of $X_k^*$ conditional on working in a low-skill occupation has much less variance than the distribution conditional on working in a high-skill occupation. Taken together, these findings highlight the importance of flexibly accounting for initially known unobserved heterogeneity.  
 
\section{Conclusion}\label{sec:conclude}

We provide new identification results for a general class of learning models that encompasses many of the set-ups that have been considered in the literature. We focus on a context where the researcher has access to a short panel of choices and realized outcomes only. As such, our approach is widely applicable, including in frequent environments where one does not have access to elicited beliefs data or auxiliary selection-free measurements. We show that the model is point-identified under two alternative sets of conditions. Our first set of conditions apply to a set-up with both known and unknown unobserved heterogeneity. We show that the model is identified under the assumption that the idiosyncratic shocks from the outcome equations and the unknown heterogeneity components are normally distributed, a very frequent restriction in empirical Bayesian learning models. We also show that normality can be relaxed in the case of a pure learning model without known heterogeneity, while preserving point-identification for this class of models.

We then derive a sieve maximum likelihood estimator for the model parameters and a particular class of functionals. The latter includes as special cases the predictable and unpredictable outcome variances, which can in turn be used to evaluate the relative importance of uncertainty versus heterogeneity in life-cycle earnings variability \citep{CHN05}. Under appropriate regularity conditions, the resulting estimators are consistent and asymptotically normal. Importantly, for practical purposes, we devise a profile likelihood-based procedure that allows us to implement our estimator at a modest computational cost. 

We illustrate our approach with an application to the role of uncertainty and learning in occupational choice, using data from the National Longitudinal Survey of Youth 1997. Our results indicate that uncertainty plays a particularly important role in accounting for the dispersion of future wages in high-skill occupations. Much of the uncertainty is revealed as individuals accumulate more work experience, pointing to the fact that ability learning plays an important role in this context.

\newpage
\bibliography{biblio}

\newpage
\appendix

\section{Proofs for identification section}

In this section, we let $\phi$ denote the standard normal p.d.f.

\subsection{Proof of Lemma \ref{lemma:conjugate}}

\begin{proof}
We proceed inductively. First, by Assumption \ref{assn:kl_normality} and the definition of $(\mu_1,\Sigma_1)$, $X^{*}_{u}\mid(X_{1},X^{*}_{k})=(x_1,x_{k}^{*})\sim N(\mu_1,\Sigma_1)$. Second, for $t\ge1$ suppose $X^{*}_{u}\mid(Y^{t-1},D^{t-1},X^{t},X^{*}_{k})=(y^{t-1},d^{t-1},x^{t},x_{k}^{*})\sim N(\mu_{t},\Sigma_{t})$. Then
\begin{align*}
  f&_{X^{*}_{u}|Y^{t},D^{t},X^{t+1},X^{*}_{k}}(x^{*}_{u};y^{t},d^{t},x^{t+1},x^{*}_{k})\\
  \propto&_{(1)}~ f_{X^{*}_{u}|Y^{t-1},D^{t-1},X^{t},X^{*}_{k}}(x^{*}_{u};y^{t-1},d^{t-1},x^{t},x^{*}_{k})\\&\times{f}_{Y_{t},D_{t},X_{t+1}|Y^{t-1},D^{t-1},X^{t},X^{*}}(y_{t},d_{t},x_{t+1};y^{t-1},d^{t-1},x^{t},x^{*})
  \\\propto&_{(2)}~f_{X^{*}_{u}|Y^{t-1},D^{t-1},X^{t},X^{*}_{k}}(x^{*}_{u};y^{t-1},d^{t-1},x^{t},x^{*}_{k})f_{Y_{t}(d_{t})|X_{t},X^{*}}(y_{t};x_{t},x^{*})
  \\\propto&_{(3)}~\exp\left(-\frac{1}{2}(x^{*}_{u}-\mu_{t})^{\intercal}\Sigma_{t}^{-1}(x^{*}_{u}-\mu_{t})\right)\phi\left(\frac{y_{t}-x_{t}^{\intercal}\beta_{t,d_{t}}-x^{*}_{k}\lambda_{t,d_{t}}^{k}-(x^{*}_{u})^{\intercal}\lambda_{t,d_{t}}^{u}}{\sigma_{t,d_{t}}}\right)
\\\propto&\exp\left(-\frac{1}{2}(x^{*}_{u}-\mu_{t})^{\intercal}\Sigma_{t}^{-1}(x^{*}_{u}-\mu_{t})\right)
\\&\times\exp\biggl(-\frac{1}{2}\left(x^{*}_{u}-\lambda_{t,d_{t}}^{u}\left((\lambda_{t,d_{t}}^{u})^{\intercal}\lambda_{t,d_{t}}^{u}\right)^{-1}(y_{t}-x_{t}^{\intercal}\beta_{t,d_{t}}-x^{*}_{k}\lambda_{t,d_{t}}^{k})\right)^{\intercal}
\\&\times\frac{\lambda_{t,d_{t}}^{u}(\lambda_{t,d_{t}}^{u})^{\intercal}}{\sigma_{t,d_{t}}^{2}}\left(x^{*}_{u}-\lambda_{t,d_{t}}^{u}\left((\lambda_{t,d_{t}}^{u})^{\intercal}\lambda_{t,d_{t}}^{u}\right)^{-1}(y_{t}-x_{t}^{\intercal}\beta_{t,d_{t}}-x^{*}_{k}\lambda_{t,d_{t}}^{k})\right)\biggr)
\\=&_{(4)}~\exp\left(-\frac{1}{2}(x^{*}_{u}-\mu_{t+1})^{\intercal}\Sigma_{t+1}^{-1}(x^{*}_{u}-\mu_{t+1})\right).
\end{align*}

Display (1) follows from Bayes' theorem. Display (2) holds since Assumption \ref{assn:kl_independence} has the following three implications: first $X_{t+1}\indep X^{*}\mid\left(Y^{t},D^{t},{X}^{t}\right)$; second $\epsilon_{t}(d_{t})\indep\left(Y^{t-1},D^{t},{X}^{t},X^{*}\right)\implies\epsilon_{t}(d_{t})\indep\left(Y^{t-1},D^{t},{X}^{t-1}\right)\mid\left({X}_{t},X^{*}\right)\implies Y_{t}(d_{t})\indep\left(Y^{t-1},D^{t},{X}^{t-1}\right)\mid\left({X}_{t},X^{*}\right)$; third $D_{t}\indep X^{*}_{u}\mid\left(Y^{t-1},D^{t-1},{X}^{t},X^{*}_{k}\right)$. Display (3) holds from the induction assumption and Assumptions \ref{assn:kl_independence} and \ref{assn:kl_normality}. Display (4) follows from the definitions in Lemma \ref{lemma:conjugate}.
\end{proof}

\subsection{Proof of Theorem \ref{thm:kl}}\label{sec:proofs_kl}

The proof of Theorem \ref{thm:kl} uses the following lemmas. 

\begin{lemma}\label{corollary:normality}
Let Assumptions \ref{assn:kl_independence} and \ref{assn:kl_normality} hold. Then $Y_{t}$ conditional on $(Y^{t-1},D^{t},X^{t},X^{*}_{k})=(y^{t-1},d^{t},x^{t},x^{*}_{k})$ is distributed
$$
N\left(x_{t}^{\intercal}\beta_{t,d_{t}}+x^{*}_{k}\lambda_{t,d_{t}}^{k}+\mu_{t}^{\intercal}\lambda_{t,d_{t}}^{u},~~(\lambda_{t,d_{t}}^{u})^{\intercal}\Sigma_{t}\lambda_{t,d_{t}}^{u}+\sigma_{t,d_{t}}^{2}\right).
$$
\end{lemma}

\begin{proof}
For $t>1$,
\begin{align*}
  &f_{Y_{t}|Y^{t-1},D^{t},X^{t},X^{*}_{k}}(y_{t};y^{t-1},d^{t},x^{t},x^{*}_{k})\\
  &=\int{f}_{Y_{t}(d_{t})|Y^{t-1},D^{t},X^{t},X^{*}}(y_{t};y^{t-1},d^{t},x^{t},x^{*})f_{X^{*}_{u}|Y^{t-1},D^{t},X^{t},X^{*}_{k}}(x^{*}_{u};y^{t-1},d^{t},x^{t},x^{*}_{k})dx^{*}_{u}\\
  &=_{(1)}\int{f}_{Y_{t}(d_{t})|X_{t},X^{*}}(y_{t};x_{t},x^{*})f_{X^{*}_{u}|Y^{t-1},D^{t-1},X^{t},X^{*}_{k}}(x^{*}_{u};y^{t-1},d^{t-1},x^{t},x^{*}_{k})dx^{*}_{u}\\
  &\propto_{(2)}\int\phi\left(\frac{y_{t}-x_{t}^{\intercal}\beta_{t,d_{t}}-x^{*}_{k}\lambda_{t,d_{t}}^{k}-(x^{*}_{u})^\intercal \lambda_{t,d_{t}}^{u}}{\sigma_{t,d_{t}}}\right)\exp\left((x^{*}_{u}-\mu_{t})^{\intercal}\Sigma_{t}^{-1}(x^{*}_{u}-\mu_{t})\right)dx^{*}_{u}\\
  &=\phi\left(\frac{y_{t}-x_{t}^{\intercal}\beta_{t,d_{t}}-x^{*}_{k}\lambda_{t,d_{t}}^{k}-\mu_{t}^{\intercal}\lambda_{t,d_{t}}^{u}}{\sqrt{(\lambda_{t,d_{t}}^{u})^{\intercal}\Sigma_{t}\lambda_{t,d_{t}}^{u}+\sigma_{t,d_{t}}^{2}}}\right)
\end{align*}
Display (1) holds because Assumption \ref{assn:kl_independence} implies $Y_{t}(d_{t})\indep ({Y}^{t-1},D^{t},X^{t-1}) \mid ({X}_{t},X^{*}) $ and $D_{t}\indep X^{*}_{u}\mid\left(Y^{t-1},D^{t-1},X^{t},X^{*}_{k}\right)$. Display (2) holds because Assumption \ref{assn:kl_independence} and \ref{assn:kl_normality} imply Lemma \ref{lemma:conjugate} and $\epsilon_{t}(d)\mid(X_{t},X^{*})\sim{N}(0,\sigma_{t,d}^{2})$. A similar argument applies for $t=1$.
\end{proof}

For the following results, it is useful to note that, for $t\ge1$,
\begin{align*}
\Sigma_{t+1}= & \left(\Sigma^{-1}_{u}(x_{1})+\sum_{s=1}^{t}\sigma^{-2}_{s,d_{s}}\lambda_{s,d_{s}}^{u}(\lambda_{s,d_{s}}^{u})^{\intercal}\right)^{-1}, \\
\mu_{t+1}=      & \Sigma_{t+1}\left(\sum_{s=1}^{t}\lambda_{s,d_{s}}^{u}\frac{y_{s}-x_{s}^{\intercal}\beta_{s,d_{s}}-x^{*}_{k}\lambda_{s,d_{s}}^{k}}{\sigma^{2}_{s,d_{s}}}\right).
\end{align*}

\begin{lemma}\label{lemma:t1-1}
Let Assumptions \ref{assn:kl_independence}, \ref{assn:kl_normality}, \ref{assn:kl_support} (A,B,C) and \ref{assn:kl_regularity} (C) hold. Then, for each $(y^{t-1},d^{t},x^{t})\in\Supp\left((Y^{t-1},D^{t},X^{t})\right)$ there exists an affine function $\pi$ such that, for all $y_t\in\Supp(Y_{t})$, $F_{Y^t,D^t,X^t,X_k^*}(y^t,d^t,x^t,\pi(x_k^*))$ is identified.
\end{lemma}

\begin{proof}
Fix $(y^{t-1},d^{t},x^{t})\in\Supp\left((Y^{t-1},D^{t},X^{t})\right)$. Since $f_{Y_{t}|Y^{t-1},D^{t},X^{t}}(y_{t};y^{t-1},d^{t},x^{t})=$
\begin{align*}
\int{f_{Y_{t}|Y^{t-1},D^{t},X^{t},X^{*}_{k}}(y_{t};y^{t-1},d^{t},x^{t},x_{k}^{*})}dF_{X^{*}_{k}|Y^{t-1},D^{t},X^{t}}(x_{k}^{*};y^{t-1},d^{t},x^{t}),
\end{align*}
Lemma \ref{corollary:normality} implies $f_{Y_{t}|Y^{t-1},D^{t},X^{t}}(y_{t};y^{t-1},d^{t},x^{t})$ is a mixture of normal random variables. To identify the component and mixture distributions, we apply \citet[Theorem 3]{bruni1985identifiability}. First, for any $t$ and $(y^{t-1},d^t,x^{t})\in\Supp\left((Y^{t-1},D^t,X^{t})\right)$, define $\Lambda\coloneqq$
$$
\left\{x_k^{*}\mapsto\left(x_t^\intercal\beta_{t,d_{t}}+x_{k}^{*}(\lambda_{t,d_t}^{k}+(\mu_t^{k})^\intercal\lambda_{t,d_t}^{u})+(\mu_t^{u})^\intercal\lambda_{t,d_t}^{u},~(\lambda_{t,d_{t}}^{u})^{\intercal}\Sigma_{t}\lambda_{t,d_{t}}^{u}+\sigma^{2}_{t,d_{t}}\right)\colon\theta^{t}\in\Theta^{t}\right\},
$$
where $\theta^{t}\coloneqq\left\{\{\beta_{s,d_s},\lambda^{k}_{s,d_s},\lambda^{u}_{s,d_s},\sigma_{s,d_s}^2:s=1,\dots,t\},\Sigma_{u}(x_{1})\right\}$, $\Theta^t$ is the corresponding subset of $\Theta$, and $\mu_{t}=\mu_{t}^{k}x_{k}^{*}+\mu_{t}^{u}$ for all $x_k^*$. I.e., $\mu_{1}^k=\mu_{1}^u=0$ and for $t>1$,
\begin{align*}
\mu_{t}^{k}&\coloneqq-\Sigma_{t} \sum_{s=1}^{t-1}\lambda_{s,d_{s}}^{u}\frac{\lambda_{s,d_{s}}^{k}}{\sigma_{s,d_{s}}^2}~,&
\mu_{t}^{u}&\coloneqq \Sigma_{t}\sum_{s=1}^{t-1}\lambda_{s,d_{s}}^{u}\frac{y_{is}-x_{is}^{\intercal}\beta_{s,d_{s}}}{\sigma_{s,d_{s}}^2}
~.
\end{align*} 
Under Assumptions \ref{assn:kl_support} (A,B,C) and \ref{assn:kl_regularity} (C),  $\Lambda\subset\Lambda_{4}$ where $\Lambda_{4}$ is defined in \citet[p. 1344]{bruni1985identifiability}. Thus \citet[Theorem 3]{bruni1985identifiability} applies and 
\begin{align*}
\left\{x_t^\intercal\beta_{t,d_{t}}+\pi(x_{k}^{*})(\lambda_{t,d_t}^{k}+(\mu_t^{k})^\intercal\lambda_{t,d_t}^{u})+(\mu_t^{u})^\intercal\lambda_{t,d_t}^{u},~(\lambda_{t,d_{t}}^{u})^{\intercal}\Sigma_{t}\lambda_{t,d_{t}}^{u}+\sigma^{2}_{t,d_{t}}\right\}
\end{align*}
and $F_{X^{*}_{k}|Y^{t-1},D^{t},X^{t}}(\pi(x_{k}^{*});y^{t-1},d^{t},x^{t})$ are identified with $\pi(x_{k}^{*})=\pi_{0}+\pi_{1}x_{k}^{*}$.
\end{proof}

\begin{lemma}\label{lemma:t1-2}
    Let Assumptions \ref{assn:kl_independence}, \ref{assn:kl_normality}, \ref{assn:kl_norm} (A), \ref{assn:kl_support} and \ref{assn:kl_regularity} (C) hold. Then $\Supp(X_k^*)$ is identified from $F_{Y_{1},D_{1},X_{1}}(y_1,d_1,x_1)$.
\end{lemma}

\begin{proof}
In this proof, it will be useful to denote $\beta_{1,d}=(\alpha_{1,d},\gamma_{1,d}^\intercal)^\intercal$, where $\alpha_{1,d}$ is the coefficient on the constant term in $X_1$. 

For any $x_1\in\Supp(X_1)$ and $d\in\Supp(D_1)$, Lemma \ref{lemma:t1-1} implies
\begin{equation*}
\left\{x_1^\intercal\beta_{1,d}+(\pi_0 + \pi_1x_{k}^{*})\lambda_{1,d}^{k},~(\lambda_{1,d}^{u})^{\intercal}\Sigma_{1}(x_1)\lambda_{1,d}^{u}+\sigma^{2}_{1,d},~F_{X^{*}_{k}|D_{1},X_{1}}(\pi_0 + \pi_1x_{k}^{*};d,x_{1})\right\}
\end{equation*}
is identified. Set $d\in\Supp(D_1)$ as in Assumption \ref{assn:kl_norm} (A). We now show $(\pi_0,\pi_1)=(0,1)$.\footnote{Recall from Lemma \ref{lemma:t1-1} that the affine function $\pi$ may depend on the history $(y^{t-1},d^t,x^t)$. In this lemma we show that the affine function is the identity for one particular choice history.} By Assumption \ref{assn:kl_support} (D), $\exists$ $x_k^*\neq\tilde{x}_k^*$ such that $dF_{X^{*}_{k}|D_{1},X_{1}}(\pi_0 + \pi_1x_{k}^{*};d,x_{1})>0$ and $dF_{X^{*}_{k}|D_{1},X_{1}}(\pi_0 + \pi_1\tilde{x}_{k}^{*};d,x_{1})>0$. Then by Assumption \ref{assn:kl_norm} (A), $1=\lambda_{1,d}^{k}=\frac{(x_1^\intercal\beta_{1,d}+(\pi_0 + \pi_1x_{k}^{*})\lambda_{1,d}^{k})-(x_1^\intercal\beta_{1,d}+(\pi_0 + \pi_1\tilde{x}_k^*)\lambda_{1,d}^{k})}{x_k^*-\tilde{x}_k^*}=\pi_{1}$. Thus $x_1^\intercal\beta_{1,d}+\pi_0$ is identified by $(x_1^\intercal\beta_{1,d}+(\pi_0 + x_{k}^{*})) - x_k^*$. If $\exists~ x_1,\tilde{x}_1\in\Supp(X_1)$ such that their respective $\pi_0$ differ, then $\Supp(X_k^*\mid X_1=x_1,D_1=d)\neq\Supp(X_k^*\mid X_1=\tilde{x}_1,D_1=d)$, which contradicts Assumption \ref{assn:kl_support} (D). Therefore $(\alpha_{1,d}+\pi_0,\gamma_{1,d}^\intercal)^\intercal = E[X_1 X_1^\intercal|D_1=d]^{-1}E[X_1 \left(X_1^\intercal\beta_{1,d}+\pi_0\right)\mid D_1=d]$, which exists by Assumption \ref{assn:kl_support} (E). Finally, by Assumption \ref{assn:kl_norm} (A), $0=\alpha_{1,d}=(x_1^\intercal\beta_{1,d}+\pi_0)-x_1^\intercal(\alpha_{1,d},\gamma_{1,d}^\intercal)^{\intercal}=\pi_0$. To conclude, by Assumption \ref{assn:kl_support} (D), $\Supp(X_k^*)=\Supp(X_k^*\mid D_1=d_1,X_1=x_1)$.
\end{proof}

\begin{lemma}\label{lemma:t1-3}
Under the assumptions in Theorem \ref{thm:kl}, $F_{Y^T,D^T,X^T,X_k^*}(y^T,d^T,x^T,x_k^*)$ is identified on its support.
\end{lemma}

\begin{proof}
For any $t$ and $(y^{t-1},d^t,x^t)\in\Supp((Y^{t-1},D^t,X^t))$, it follows from Lemma \ref{lemma:t1-1} that $dF_{X^{*}_{k}|Y^{t-1},D^{t},X^{t}}(\pi(x_{k}^{*});y^{t-1},d^{t},x^{t})$ is identified. Then since $\Supp(X_k^*)$ is known by Lemma \ref{lemma:t1-2}, Assumption \ref{assn:kl_support} (D) implies $\Supp(X^{*}_{k})= $
\begin{align*}
   dF^{-1}_{X^{*}_{k}|Y^{t-1},D^{t},X^{t}}(\cdot;y^{t-1},d^{t},x^{t})[\mathbb{R}^{*}_{+}]=(dF_{X^{*}_{k}|Y^{t-1},D^{t},X^{t}}(\cdot;y^{t-1},d^{t},x^{t})\circ\pi)^{-1}[\mathbb{R}^{*}_{+}],
\end{align*}
where $\mathbb{R}^{*}_{+}=\{x\in\mathbb{R}\colon{x}>0\}$. Then, since $\pi$ is bijective, $\pi[\Supp(X_k^*)]=\Supp(X_k^*)$. The only affine functions that satisfy this identity are $\pi(x_k^*)=x_k^*$ and $\pi(x_k^*)=\sup\Supp(X_k^*)+\inf\Supp(X_k^*)-x_k^*$. To conclude the proof, we need to rule out the second function.

To proceed, let $\mu_{t}^{k}$ and $\mu_{t}^{u}$ be defined as in the proof to Lemma \ref{lemma:t1-1}, and, for any $1\le s<t$, let $\tilde{\mu}_{t,s}(d^{t-1})\coloneqq\Sigma_{t}\frac{\lambda_{s,d_{s}}^{u}}{\sigma^{2}_{s,d_{s}}}$. Now note that by Lemma \ref{lemma:t1-1} and Assumption \ref{assn:kl_support}, for any $t$ and $d^t\in\Supp(D^t)$, $j c_t(d^t)=\lambda_{t,d_t}^{k}+(\mu_t^{k})^\intercal\lambda_{t,d_t}^{u}$ with $j\in\{-1,1\}$ unknown and $c_t(d^t)\coloneqq\frac{(x_t^\intercal\beta_{t,{d}_{t}}+\pi(x_{k}^{*})\lambda_{t,d_t}^{k}+\mu_t^\intercal\lambda_{t,d_t}^{u})-(x_t^\intercal\beta_{t,{d}_{t}}+\pi(\tilde{x}_{k}^{*})\lambda_{t,d_t}^{k}+\mu_t^\intercal\lambda_{t,d_t}^{u})}{x_k^*-\tilde{x}_k^*}$ known. In addition, for any $1\le s< t$, $\frac{\partial}{\partial y_s}(x_t^\intercal\beta_{t,{d}_{t}}+\pi(x_{k}^{*})\lambda_{t,d_t}^{k}+\mu_t^\intercal\lambda_{t,d_t}^{u})=(\lambda_{t,d_t}^u)^\intercal\tilde{\mu}_{t,s}(d^{t-1}).$

The argument is inductive. First consider $t=1$. Applying the above argument to the sequences $\{\tilde{d}_1,(d_1,d_2),(\tilde{d_1},d_2)\}$ for $d_1\in\Supp(D_1)$ as in Assumption \ref{assn:kl_norm} (A), $\tilde{d}_{1}\in\Supp(D_1)\setminus\{{d}_{1}\}$, and $d_2\in\Supp(D_2)$, yields identification of $j_1 c_1(\tilde{d}_1)$, ${j_{d_{2}}}c_2((d_1,d_2))$ $(\lambda_{2,d_{2}}^{u})^{\intercal}\tilde{\mu}_{2,1}(d_{1})$, ${\tilde{j}_{d_{2}}}c_2((\tilde{d}_1,d_2))$, and $(\lambda_{2,d_{2}}^{u})^{\intercal}\tilde{\mu}_{2,1}(\tilde{d}_{1})$  with $(j_{1},\tilde{j}_{d_{2}},j_{d_{2}})\in\{-1,1\}^{3}$ unknown. Since $\lambda_{1,d_1}^k=1$, $j_1 c_1(\tilde{d}_1)=\lambda_{1,\tilde{d}_{1}}^{k}$, ${j_{d_{2}}}c_2((d_1,d_2))=\lambda_{2,d_{2}}^{k}-(\lambda_{2,d_{2}}^{u})^{\intercal}\tilde{\mu}_{2,1}(d_{1})$, and ${\tilde{j}_{d_{2}}}c_2((\tilde{d}_1,d_2))=(\lambda_{2,d_{2}}^{k}-(\lambda_{2,d_{2}}^{u})^{\intercal}\tilde{\mu}_{2,1}(\tilde{d}_{1})\lambda_{1,\tilde{d}_{1}}^{k})$, it must be that
\begin{equation}\label{eq:lemma5-id}
(\lambda_{2,d_{2}}^{u})^{\intercal}\tilde{\mu}_{2,1}(d_{1})+{j_{d_{2}}}c_2((d_1,d_2))=(\lambda_{2,d_{2}}^{u})^{\intercal}\tilde{\mu}_{2,1}(\tilde{d}_{1}) j_1 c_1(\tilde{d}_1)+{\tilde{j}_{d_{2}}}c_2((\tilde{d}_1,d_2)).
\end{equation}
We use this identity to show $(j_{1},\tilde{j}_{d_{2}},j_{d_{2}})=(1,1,1)$. Suppose $j_{d_{2}}=1$. It is straightforward to show that Equation \eqref{eq:lemma5-id} implies:
\begin{align*}
      (j_{1},\tilde{j}_{d_{2}})=(-1,-1)&\implies \lambda_{2,d_{2}}^{k}=0,\\
      (j_{1},\tilde{j}_{d_{2}})=(1,-1)&\implies \lambda_{2,d_{2}}^{k}-(\lambda_{2,d_{2}}^{u})^{\intercal}\tilde{\mu}_{2,1}(\tilde{d}_{1})\lambda_{1,\tilde{d}_{1}}^{k}=0,\\
      (j_{1},\tilde{j}_{d_{2}})=(-1,1)&\implies (\lambda_{2,d_{2}}^{u})^{\intercal}\tilde{\mu}_{2,1}(\tilde{d}_{1})\lambda_{1,\tilde{d}_{1}}^{k}=0,
\end{align*}
which contradict Assumptions \ref{assn:kl_regularity} (B), (C) and (D), respectively. Now suppose $j_{d_{2}}=-1$, then
\begin{align*}
    (j_{1},\tilde{j}_{d_{2}})=(1,1)&\implies \lambda_{2,d_{2}}^{k}-(\lambda_{2,d_{2}}^{u})^{\intercal}\tilde{\mu}_{2,1}(d_{1})\lambda_{1,d_{1}}^{k}=0,\\
    (j_{1},\tilde{j}_{d_{2}})=(-1,-1)&\implies (\lambda_{2,d_{2}}^{u})^{\intercal}\tilde{\mu}_{2,1}(d_{1})\lambda_{1,d_{1}}^{k}=0,\\
    (j_{1},\tilde{j}_{d_{2}})=(1,-1)&\implies (\lambda_{2,d_{2}}^{u})^{\intercal}\tilde{\mu}_{2,1}(\tilde{d}_{1})\lambda_{1,\tilde{d}_{1}}^{k}-(\lambda_{2,d_{2}}^{u})^{\intercal}\tilde{\mu}_{2,1}(d_{1})\lambda_{1,d_{1}}^{k}=0,\\
    (j_{1},\tilde{j}_{d_{2}})=(-1,1)&\implies {\lambda_{2,d_{2}}^{k}}-(\lambda_{2,d_{2}}^{u})^{\intercal}\tilde{\mu}_{2,1}(\tilde{d}_{1})\lambda_{1,\tilde{d}_{1}}^{k}-(\lambda_{2,d_{2}}^{u})^{\intercal}\tilde{\mu}_{2,1}(d_{1})\lambda_{1,d_{1}}^{k}=0.
\end{align*}
The first three implications contradict Assumptions \ref{assn:kl_regularity} (C), (D) and (A), respectively. To conclude, for each $d\in\{d_{2,i},\tilde{d}_{2,i}\in\Supp{(D_{2})}:i=1,2,\dots,p\}$ of Assumption \ref{assn:kl_regularity} (E), by considering the sequences $\{(d_{1},{d}),(\tilde{d}_{1},{d})\}$, ${j_{d}}c_2((d_1,d))$ and $\tilde{j}_d c_2((\tilde{d}_1,d))$ are identified with $(j_{d},\tilde{j}_{d})\in\{(-1,1),(1,1)\}$. Since $\lambda_{1,\tilde{d}_{1}}^{k}\neq0$ by Assumption \ref{assn:kl_regularity} (B), for the sign of $\lambda_{1,\tilde{d}_{1}}^{k}$ to be constant across sequences, we can rule out all signs except $\left(j_{1},(j_{d_{2,i}},\tilde{j}_{d_{2,i}},j_{\tilde{d}_{2,i}},\tilde{j}_{\tilde{d}_{2,i}}:i=1,\dots,p)\right)\in\left\{\left(1,(1,1,1,1)^{p}\right),\left(-1,(-1,1,-1,1)^{p}\right)\right\}$. If $\left(j_{1},(j_{d_{2,i}},\tilde{j}_{d_{2,i}},j_{\tilde{d}_{2,i}},\tilde{j}_{\tilde{d}_{2,i}}:i=1,\dots,p)\right)=\left(-1,(-1,1,-1,1)^{p}\right)$, then 
\begin{align*}
&0=\mathrm{vec}\left(\lambda_{2,d_{2,1}}^{k},\dots,\lambda_{2,d_{2,p}}^{k}\right)-\left(\lambda_{2,d_{2,1}}^{u}\cdots{\lambda}_{2,d_{2,p}}^{u}\right)^{\intercal}\left(\tilde{\mu}_{2,1}(\tilde{d}_{1})\lambda_{1,\tilde{d}_{1}}^{k}+\tilde{\mu}_{2,1}(d_{1})\lambda_{1,d_{1}}^{k}\right)\\&=
\mathrm{vec}\left(\lambda_{2,\tilde{d}_{2,1}}^{k},\dots,\lambda_{2,\tilde{d}_{2,p}}^{k}\right)-\left(\lambda_{2,\tilde{d}_{2,1}}^{u}\cdots{\lambda}_{2,\tilde{d}_{2,p}}^{u}\right)^{\intercal}\left(\tilde{\mu}_{2,1}(\tilde{d}_{1})\lambda_{1,\tilde{d}_{1}}^{k}+\tilde{\mu}_{2,1}(d_{1})\lambda_{1,d_{1}}^{k}\right),
\end{align*}
which contradicts Assumption \ref{assn:kl_regularity} (E).

For the induction step, suppose $\pi$ is identity for each history ($y^{s-1},d^{s},x^{s}$), $s=1,\dots,t-1$, and let $d^{t},\tilde{d}^t\in\Supp(D^{t})$ satisfy $d_t=\tilde{d}_t$ and $d_{t-1}\neq\tilde{d}_{t-1}$. By the preceding arguments, $j_1 c_t(d^t)$, $j_2 c_t(\tilde{d}^{t})$ with $(j_{1},j_{2})\in\{-1,1\}^{2}$, and, for each $s<t$, $(\lambda_{t,d_{t}}^{u})^{\intercal}\tilde{\mu}_{t,s}(d^{t-1})$ and $(\lambda_{t,d_{t}}^{u})^{\intercal}\tilde{\mu}_{t,s}(\tilde{d}^{t-1})$ are identified. Since $\lambda_{s,d}^k$ is identified for any $s<t$ and $d\in\Supp(D_s)$, $ j_1 c_t(d^t)=\lambda_{t,d_{t}}^{k}-(\lambda_{t,d_{t}}^{u})^{\intercal}\sum_{s=1}^{t-1}\tilde{\mu}_{t,s}(d^{t-1})\lambda_{s,d_{s}}^{k}$ and $j_2 c_t(\tilde{d}^{t})=\lambda_{t,d_{t}}^{k}-(\lambda_{t,d_{t}}^{u})^{\intercal}\sum_{s=1}^{t-1}\tilde{\mu}_{t,s}(d^{t-1})\lambda_{s,\tilde{d}_{s}}^{k}$, it must be that
\begin{equation}
    j_1 c_t(d^{t}) +(\lambda_{t,d_{t}}^{u})^\intercal\sum_{s=1}^{t-1}\tilde{\mu}_{t,s}(d^{t-1})\lambda_{s,d_{s}}^{k} = j_2 c_t(\tilde{d}^{t}) +(\lambda_{t,d_{t}}^{u})^\intercal\sum_{s=1}^{t-1}\tilde{\mu}_{t,s}(\tilde{d}^{t-1})\lambda_{s,\tilde{d}_{s}}^{k}.
\end{equation}
We use this identity to show $(j_1,j_2)=(1,1)$. Consider
\begin{align*}
  (j_{1},j_{2})=(1,-1)&\implies\left(\lambda_{t,d_{t}}^{k}-(\lambda_{t,d_{t}}^{u})^{\intercal}\sum_{s=1}^{t-1}\tilde{\mu}_{t,s}(\tilde{d}^{t-1})\lambda_{s,\tilde{d}_{s}}^{k}\right)=0,\\
  (j_{1},j_{2})=(-1,1)&\implies\left(\lambda_{t,d_{t}}^{k}-(\lambda_{t,d_{t}}^{u})^{\intercal}\sum_{s=1}^{t-1}\tilde{\mu}_{t,s}({d}^{t-1})\lambda^{k}_{s,{d}_{s}}\right)=0,\\
  (j_{1},j_{2})=(-1,-1)&\implies (\lambda_{t,d_{t}}^{u})^{\intercal}\sum_{s=1}^{t-1}\tilde{\mu}_{t,s}({d}^{t-1})\lambda^{k}_{s,{d}_{s}}-(\lambda_{t,d_{t}}^{u})^{\intercal}\sum_{s=1}^{t-1}\tilde{\mu}_{t,s}(\tilde{d}^{t-1})\lambda_{s,\tilde{d}_{s}}^{k}=0,
\end{align*}
which contradict Assumptions \ref{assn:kl_regularity} (C), (C) and (A), respectively.
\end{proof}

\begin{proof}[Proof of Theorem \ref{thm:kl}]
  By Lemma \ref{lemma:t1-3}, $f_{Y^{T},D^{T},X^{T},X^{*}_{k}}$, and thus $h_{t}$, is identified. First,
\begin{align*}
& f_{Y^{T},D^{T},X^{T},X^{*}_{k}}\left(y^{T},d^{T},x^{T},x_{k}^{*}\right)\\
=&\int{f}_{Y^{T}(d^{T}),D^{T},X^{T},X^{*}}\left(y^T,d^{T},x^{T},x^{*}\right)dx_{u}^{*}\\
=&\int{f}_{Y_{T}(d_{T})|X_T,X^{*}}\left(y_{T};x_{T},x^{*}\right)f_{D_{T}|Y^{T-1},D^{T-1},X^{T},X^{*}_{k}}(d_{T};y^{T-1},d^{T-1},x^{T},x_{k}^{*})\\
&\times f_{X_{T}|Y^{T-1},D^{T-1},X^{T-1}}(x_{T};y^{T-1},d^{T-1},x^{T-1})\dots{f}_{Y_{1}(d_{1})|{X}_{1},X^{*}}\left(y_{1};x_{1},x^{*}\right)\\
&\times f_{D_{1}|X_{1},X^{*}_{k}}(d_{1};x_1,x_{k}^{*})f_{X^{*}_{u}|X_1,X^{*}_{k}}(x_{u}^{*};x_1,x_{k}^{*})f_{X_{1},X^{*}_{k}}(x_{1},x_{k}^{*})dx_{u}^{*}.
\end{align*}
This implies that on the support of $f_{Y^{T},D^{T},X^{T},X^{*}_{k}}$,
\begin{align*}
&  \frac{f_{Y^{T},D^{T},X^{T},X^{*}_{k}}\left(y^{T},d^{T},x^{T},x_{k}^{*}\right)}{f_{D_{1},X_1,X^{*}_{k}}(d_{1},x_1,x_{k}^{*})\prod_{t=2}^{T}f_{D_{t},X_{t}|Y^{t-1},D^{t-1},X^{t-1},X^{*}_{k}}(d_{t},x_{t};y^{t-1},d^{t-1},x^{t-1},x_{k}^{*})}\\
&=\int\prod_{t=1}^{T}{f}_{Y_{t}(d_{t})|X_{t},X^{*}}\left(y_{t};x_t,x^{*}\right)f_{X^{*}_{u}|X^{*}_{k},X_{1}}(x_{u}^{*};x_{k}^{*},x_{1})dx_{u}^{*}.
\end{align*}
The function is equal to the probability density function of a jointly normal random variable with mean
\begin{equation*}
\left(x_{t}^{\intercal}\beta_{t,d_{t}}+x_{k}^{*}\lambda_{t,d_{t}}^{k}\right)_{t=1}^{T},
\end{equation*}
and covariance matrix
\begin{equation*}
(\lambda^{u}_{d})^{\intercal}\Sigma_{u}(x_{1})\lambda^{u}_{d}+\mathrm{diag}\left(\sigma_{t,d_{t}}^{2}\colon{t}=1,\dots,T\right),
\end{equation*}
where $\lambda^{u}_{d}=\left(\lambda_{1,d_{1}}^{u}\cdots{\lambda}_{T,d_{T}}^{u}\right)$. By Assumptions \ref{assn:kl_support} (D) and (E), the components of the mean function are identified. The components of the covariance matrix are identified under Assumptions \ref{assn:kl_norm} (B) and \ref{assn:kl_regularity} (F).
\end{proof}

\subsection{Proof of Theorem \ref{Thm:scalar}}
\label{sec:proofs_l}

In this section denote $\mathcal{L}=\{m\colon\mathbb{R}^{k}\rightarrow\mathbb{R}:\sup_{a\in\mathbb{R}^{k}}|m(a)|<\infty, \int|m(a)|da<\infty\}$ and $\mathcal{L}_{A}=\{m\colon\mathbb{R}^{k}\rightarrow\mathbb{R}:\sup_{a\in\mathbb{R}^{k}}|m(a)|<\infty,~\int|m(a)|f_{A}(a)da<\infty\}$ for a random variable $A$ with p.d.f. $f_{A}$.

\begin{proof}
Let $x\in\Supp(X)$ and $d^T\in\Supp(D^T)$ whose first $p$ elements satisfy Assumption \ref{assn:l_norm}, and define $W_{1}=(Y_{1},\dots,Y_{p})$, $W_{2}=Y_{p+1}$ and $W_{3}=(Y_{p+2},\dots,Y_{T})$. Let $L_{123}:\mathcal{L}_{W_{3}}\rightarrow\mathcal{L}$ and $L_{13}:\mathcal{L}_{W_{3}}\rightarrow\mathcal{L}$ be defined as $[L_{123}m](w_{1})=$
\begin{equation*}
\int\frac{ f_{Y,D,X}(y,d,x)}{f_{D_{1},X_1}(d_{1},x_1)\prod_{t=2}^{T}f_{D_{t},X_t|Y^{t-1},D^{t-1},X^{t-1}}(d_{t},x_t;y^{t-1},d^{t-1},x^{t-1})}m(w_{3})dw_{3},
\end{equation*}
and $[L_{13}m](w_{1})=\int[L_{123}m](w_{1})dw_{2}$. In addition, define
\begin{align*}
  L_{1X^{*}}:\mathcal{L}\rightarrow\mathcal{L}\qquad&[L_{1X^{*}}m](w_{1})=\int\prod_{t=1}^{p}{f}_{Y_{t}({d_{t}})|X_t,X^{*}}(y_{t};x_t,x^{*})m(x^{*})dx^{*},\\
L_{X^{*}3}:\mathcal{L}_{W_{3}}\rightarrow\mathcal{L}\qquad&[L_{X^{*}3}m](x^{*})=\int\prod_{t=p+2}^{T}{f}_{Y_{t}({d_{t}})|X_t,X^{*}}(y_{t};x_t,x^{*})f_{X^{*}|X_1}(x^{*};x_1)m(w_{1})dw_{1},\\
D_{X^{*}}:\mathcal{L}_{X^{*}}\rightarrow\mathcal{L}_{X^{*}}\qquad&[D_{X^{*}}m](x^{*})=f_{Y_{p+1}(d_{p+1})|X_{p+1},X^{*}}(y_{p+1};x_{p+1},x^{*})m(x^{*}).
\end{align*}
The following derivation shows that $L_{123}=L_{1X^{*}}D_{X^{*}}L_{X^{*}3}$. First,
\begin{align*} f_{Y,D,X}(y,d,x)=&\int{f}_{Y,D,X,X^{*}}(y,d,x,x^{*})dx^{*}\\
=&\int{f}_{Y_{T}(d_{T})|X_T,X^{*}}(y_{T};x_T,x^{*})f_{D_{T},X_T|Y^{T-1},D^{T-1},X^{T-1}}(d_{T},x_{T};y^{T-1},d^{T-1},x^{T-1})\\
&\times{f}_{Y_{T-1}(d_{T-1})|X_{T-1},X^{*}}(y_{T-1};x_{t-1},x^{*})\dots{f}_{D_1,X_1}(d_1,x_1){f}_{X^{*}|X_1}(x^{*};x_1)dx^{*}.
\end{align*}
Then, by Assumption \ref{assn:l_support} (A),
\begin{align*}
    &\frac{ f_{Y,D,X}(y,d,x)}{f_{D_{1},X_1}(d_{1},x_1)\prod_{t=2}^{T}f_{D_{t},X_t|Y^{t-1},D^{t-1},X^{t-1}}(d_{t},x_t;y^{t-1},d^{t-1},x^{t-1})}\\&=\int\prod_{t=1}^T{f}_{Y_{t}(d_{t})|X_t,X^{*}}(y_{t};x_t,x^{*}){f}_{X^{*}|X_1}(x^{*};x_1)dx^{*},
\end{align*}
and therefore it follows that
\begin{align*}
[L_{123}m](w_{1})=&\int\left(\int\prod_{t=1}^{T}{f}_{Y_{t}(d_{t})|X_t,X^{*}}(y_{t};x_t,x^{*})f_{X^{*}|X_t}(x^{*};x_t)dx^{*}\right){m}(w_{3})dw_{3}\\=&\int\prod_{t=1}^{p+1}{f}_{Y_{t}(d_{t})|X_t,X^{*}}(y_{t};x_t,x^{*})\left(\int\prod_{t=p+2}^{T}{f}_{Y_{t}(d_{t})|X_t,X^{*}}(y_{t};x_t,x^{*})f_{X^{*}|X_t}(x^{*})m(w_{3})dw_{3}\right)dx^{*}\\=&\int\prod_{t=1}^{p}{f}_{Y_{t}(d_{t})|X_t,X^{*}}(y_{t};x_t,x^{*})\left(f_{Y_{p+1}(d_{p+1})|X_{p+1},X^{*}}(y_{p+1};x_{p+1},x^{*})[L_{X^{*}3}m](x^{*})\right)dx^{*}\\=&\int\int\prod_{t=1}^{p}{f}_{Y_{t}(d_{t})|X_t,X^{*}}(y_{t};x_t,x^{*})[D_{X^{*}}L_{X^{*}3}m](x^{*})dx^{*}\\=&[L_{1X^{*}}D_{X^{*}}L_{X^{*}3}m](w_{1}),
\end{align*}
and $L_{123}=L_{1X^{*}}D_{X^{*}}L_{X^{*}3}$. Similarly,
$L_{13}=L_{1X^{*}}L_{X^{*}3}$.

From here, Assumptions \ref{assn:l_independence}, \ref{assn:l_normality}, \ref{assn:l_norm}, \ref{assn:l_support} (B), and \ref{assn:l_regularity} imply the arguments of Theorem 1 \citet{freyberger2017non} apply, so that $\lambda_{t,d_t}$, $f_{Y_{t}(d_{t})|X_t,X^{*}}(\cdot;x_t,\cdot)$ and $f_{X^{*}|X_1}(\cdot;x_1)$ are identified for each $t$ for the given $(d_t,x)$.\footnote{The listed assumptions imply the assumptions of \citet[Theorem 1]{freyberger2017non} with the primary exception of Assumption \ref{assn:l_independence} that differs from Assumption N5 in \citet{freyberger2017non} by allowing period $t$ variables to impact the evolution of period $t'$ covariates for $t'>t$. However, since Assumption \ref{assn:l_independence} implies $f_{Y_t(d_t)|X_t,X^{*}}(y;x,x^{*})=f_{\epsilon_t(d_t)}(y-x^\intercal\beta_{t,d_t}- (x^{*})^\intercal \lambda_t)$, \citet[Lemma 1]{freyberger2017non} and \citet{dhaultfoeuille2011completeness} can be applied with minor modifications.} Given identification of $f_{Y_{t}(d_{t})|X_t,X^{*}}(\cdot;x_t,\cdot)$ for each $x_t\in\Supp(X_t)$ and $\lambda_{t,d_t}$, Assumption \ref{assn:l_support} (C) implies identification of $\beta_{t,d_{t}}$ and thus $f_{\epsilon_t(d_t)}$.

Next, given an arbitrary $t$ and $d_{t}$, define $\tilde{d}$ by replacing the $t$-th element of $d$ with $d_{t}$. Then consider a permutation $(1,2,\dots,T)\mapsto(t_{1},t_{2},\dots,t_{T})$ such that $t\mapsto{t}_{1}$ and define $\tilde{W}_{1}=(Y_{t_{1}},Y_{t_{2}},\dots,Y_{t_{p}})$, $\tilde{W}_{2}=(Y_{t_{p+1}},Y_{t_{p+1}},\dots,Y_{t_{T}})$,
\begin{align*}
  \tilde{L}_{2X^{*}}:\mathcal{L}\rightarrow\mathcal{L}\qquad&[\tilde{L}_{2X^{*}}m](\tilde{w}_{2})=\int\prod_{i=p+1}^{T}{f}_{Y_{t_{i}}{(d_{t_{i}})}|X_{t_i},X^{*}}(y_{t_{i}};x_{t_i},x^{*})f_{X^{*}|X_1}(x^{*};x_1)m(x^{*})dx^{*},\\
  \tilde{L}_{X^{*}1}:\mathcal{L}_{\tilde{W}_{1}}\rightarrow\mathcal{L}\qquad&[\tilde{L}_{X^{*}1}m](x^{*})=\int\prod_{i=1}^{p}{f}_{Y_{t_{i}}({d_{t_{i}})}|X_{t_i},X^{*}}(y_{t_{i}};x_{t_i},x^{*})m(\tilde{w}_{1})d\tilde{w}_{1},
\end{align*}
and $\tilde{L}_{21}:\mathcal{L}_{\tilde{W}_{1}}\rightarrow\mathcal{L}$ as
\begin{equation*}
    [\tilde{L}_{21}m](\tilde{w}_{2})=\int\frac{ f_{Y,D,X}(y,d,x)}{f_{D_{1},X_1}(d_{1},x_1)\prod_{t=2}^{T}f_{D_{t},X_t|Y^{t-1},D^{t-1},X^{t-1}}(d_{t},x_t;y^{t-1},d^{t-1},x^{t-1})}m(\tilde{w}_1)d\tilde{w}_1.
\end{equation*}
As before, $\tilde{L}_{21}=\tilde{L}_{2X^{*}}\tilde{L}_{X^{*}1}$. Since $\tilde{L}_{2X^{*}}$ and $\tilde{L}_{21}$ are identified and injective, $\tilde{L}_{X^{*}1}$ is identified by $\tilde{L}_{2X^{*}}^{-1}\tilde{L}_{21}=\tilde{L}_{X^{*}1}$ and thus $\beta_{t,d_{t}},\lambda_{t,d_{t}},f_{\epsilon(d_{t})}$.
\end{proof}

\newpage
\section{Online Appendix}

\subsection{Proof of Corollary \ref{thm:normalized}}
\label{prf:normalized}

In this proof, we denote $\beta_{t,d}=(\alpha_{t,d},\gamma_{t,d}^\intercal)^\intercal$, where $\alpha_{t,d}$ is the coefficient on the constant term in $X_t$. Fix $d^p$ as in the statement and define $\lambda_{u}=\left(\lambda_{1,d_{1}}^{u}\cdots\lambda_{p,d_{p}}^{u}\right)$, $\tilde{X}^{*}_{u}=\lambda_{u}^{\intercal}\left(X^{*}_{u}-\mu_{u}\right)$, $\tilde{\epsilon}_{t}(d)=\epsilon_{t}(d)-c_{t,d}$, $\tilde{X}^{*}_{k}=b+\lambda_{1,d_{1}}^{k}X^{*}_{k}$ where $b=\alpha_{1,d_{1}}+\mu_{u}^\intercal\lambda_{1,d_{1}}^{u}+c_{1,d_{1}}$. Finally, define $\tilde{\lambda}_{t,d_{t}}^{k}=(\lambda_{1,d_{1}}^{k})^{-1}\lambda_{t,d_{t}}^{k}$, $\tilde{\lambda}_{t,d_{t}}^{u}=\lambda_{u}^{-1}\lambda_{t,d_{t}}^{u}$, and $\tilde{\alpha}_{t,d_t}=\alpha_{t,d_t}-\tilde{\lambda}_{t,d_t}^{k}b+\mu_{u}^\intercal\lambda_{t,d_t}^{u}+c_{t,d_t}$. We then have that 
\begin{equation*}
Y_{t}(d_t)=X_t^{\intercal}\left(\tilde{\alpha}_{t,d_t},\gamma_{t,d_t}^\intercal\right)^\intercal+(\tilde{X}^{*}_{u})^{\intercal}\tilde{\lambda}_{t,d_t}^{u}+\tilde{X}^{*}_{k}\tilde{\lambda}_{t,d_t}^{k}+\tilde{\epsilon}_{t}(d_t),
\end{equation*}
$E[\tilde{\epsilon}_{t}(d_t)]=0$ and $E[\tilde{X}^{*}_{u}\mid{X_{1}=x},X^{*}_{k}=x_{k}^{*}]=0$ so that the reparameterized model satisfies Assumption \ref{assn:kl_normality} (with $\tilde\Sigma_u(x_1)=\lambda_{u}^{\intercal}\Sigma_u(x_1)\lambda_{u}$). Also, $\tilde{\lambda}_{1,d_{1}}^{k}=1$, $\tilde{\alpha}_{1,d_{1}}=0$ and $\left(\tilde{\lambda}_{1,d_{1}}^{u}\cdots\tilde\lambda_{p,d_{p}}^{u}\right)=I_{p\times p}$ so the reparameterized model satisfies Assumption \ref{assn:kl_norm}. By Theorem \ref{thm:kl}, $\tilde{\theta}=\left\{\{\tilde{\alpha}_{t,d_t},\gamma_{t,d_t},\tilde{\lambda}^{k}_{t,d_t},\tilde{\lambda}^{u}_{t,d_t},\sigma^2_{t,d_t},g_t,\tilde{h}_t\}_{t=1}^{T},\tilde\Sigma_{u},F_{\tilde{X}^{*}_{k}X_1}\right\}$ is identified, where $\tilde{h}_t$ and $F_{\tilde{X}^{*}_{k}X_1}$ are the CCPs and distribution of $(\tilde{X^{*}_k},X_1)$, respectively. This, in turn, implies the identification of the distribution of $C_{t,d_t}^j$ for $j=k,u$. Finally,
\begin{align*}
&x^{\intercal}\left(\tilde{\alpha}_{t,d_t},\gamma_{t,d_t}^\intercal\right)^\intercal+Q_{\alpha}[\tilde{C}^k_{t,d_t}+\tilde{C}^u_{t,d_t}+\tilde{\epsilon}_{t}(d_t)]\\
=&x^{\intercal}\beta_{t,d_t}-\tilde{\lambda}_{t,d_t}^{k}b+\mu_u^\intercal\lambda^{u}_{t,d_t}+c_{t,d_t}+Q_{\alpha}[\tilde{C}^k_{t,d_t}+\tilde{C}^u_{t,d_t}+\tilde{\epsilon}_{t}(d_t)]\\
=&x^{\intercal}{\beta}_{t,d_t}-\tilde{\lambda}_{t,d_t}^{k}b+\mu_u^\intercal\lambda^{u}_{t,d_t}+c_{t,d_t}+Q_{\alpha}[C^k_{t,d_t}+\tilde{\lambda}_{t,d_t}^{k}b+C^u_{t,d_t}-\mu_u^\intercal\lambda^{u}_{t,d_t}+\epsilon_{t}(d_t)-c_{t,d_t}]\\
=&x^{\intercal}\beta_{t,d_t}+Q_{\alpha}[C^k_{t,d_t}+C^u_{t,d_t}+\epsilon_{t}(d_t)].
\end{align*}

\subsection{Variance decompositions}
\label{App:var_dec}

As discussed in Section \ref{sec2}, an important class of parameters in learning models are terms that decompose the variance of potential outcomes into components that are predictable and unpredictable given the agents' information. These parameters can be expressed as functionals of the finite- and infinite-dimensional components of the model parameters. Section \ref{sec:estimation} provides general inference results, which can be applied to a plug-in sieve MLE estimator of these parameters. In this section, we define these parameters and discuss their relevance to quantifying the importance of uncertainty and learning.   

To define this class of parameters, consider a weighted sum of potential outcomes, $Y(\omega^T, d^T) = \sum_{t} \omega_t Y_t(d_t)$ for a sequence of choices $d^T$ and weights, $\omega^T$. \cite{CH16} consider a special case of this parameter in the context of an educational choice model. In particular, they consider the present value of lifetime earnings, which is defined as $Y(\omega^T, d^T)$, with $\omega_t = \indic{(t \ge t_0)(1 - \rho)^{t_0 - t}}$, for some discount rate $0 \le \rho < 1$. 

Next, define the agent's information set as $\mathcal{I}_t = \{Y^{t-1},D^{t-1},X^t,X^{*}_k\}$ for $t>1$ and $\mathcal{I}_1=\{X_1,X_k^*\}$. Restricting attention to weighted sums where $\omega_s = 0$ for $s < t$, the variance of $Y(\omega^T, d^T)$ conditional on $\mathcal{I}_t$ can be understood as the variance that is due to the agent's uncertainty over $Y(\omega^T, d^T)$ given their information up to period $t$. We refer to this as the {\it posterior variance}, because this is derived from the posterior distribution of $X^{*}_u$ after performing a Bayesian update with the information in $\mathcal{I}_t$.

In its full generality, the model allows for endogeneity in $X_t$ as the transition probabilities depend on past choices and outcomes. Therefore, the posterior variance of $Y(\omega^T, d^T)$ includes terms that reflects uncertainty about the future realizations of $X_t$ conditional on $X^{*}_k$. In order to focus on uncertainty over $X^{*}$, we abstract from this by assuming that the covariates are not time varying, which we denote as $X$.\footnote{When the covariates are time varying and transitions depend on $(D^{t-t}, Y^{t-1})$, the posterior variance will include the covariances between future realizations of $X_t$ and between $X_t$ and $X^{*}_u$ conditional on the information set. These terms reflect another channel through which unobserved heterogeneity is related to the agents' uncertainty. In this case, the plug-in estimator of the posterior variance will involve other infinite dimension parameters of the model (e.g., $f_{D_t \mid X^t, Y^{t-1}, D^{t-1}, X^{*}_k}$).} 

In particular, with this restriction on the covariates, Lemma \ref{lemma:conjugate} implies that the posterior variance, which we denote as $V_t^u(X, D^{t-1}; \omega^T, d^T)\coloneqq \V \left( Y(\omega^T, d^T) \mid \mathcal{I}_t \right)$, has the form
\begin{align*}
    V_t^u(X, D^{t-1}; \omega^T, d^T)
    &:= \sum_{t_1, t_2 \ge t} \omega_{t_1} \omega_{t_2} (\lambda^u_{t_1, d_{t_1}})^{\intercal} \Sigma_{t} \lambda^u_{t_2, d_{t_2}}
    + \sum_{t_1 \ge t} \omega^2_{t_1} \sigma_{t_1, d_{t_1}}^2
\end{align*}
for $t>1$ where $\Sigma_t$ is the posterior variance of $X^{*}_u$ as written in Lemma \ref{lemma:conjugate}.\footnote{Note that $\Sigma_t$ depends on certain components of $\mathcal{I}_t$.} When $t = 1$, $D^{t-1}$ is empty so we write $V_1^u(X; \omega^T, d^T)\coloneqq \V \left( Y(\omega^T, d^T) \mid \mathcal{I}_1 \right)$.

At $t = 1$, the following variance decomposition provides a natural way to quantify the relative importance of uncertainty in potential outcomes, 
\begin{align}
\V(Y(\omega^T,d^T) \mid X) 
=  V_1^u(X; \omega^T, d^T) 
+ \sum_{t_1, t_2 \ge 1}\omega_{t_1}\omega_{t_2}\lambda^k_{t_1, d_{t_1}} \lambda^k_{t_2, d_{t_2}} \V(X^{*}_k \mid X)
\label{eq:vardecom-t0}
\end{align}
This corresponds to the decomposition in \cite{CH16} and in that context, has the simple interpretation that the first term is the portion of variance in the lifetime earnings that is due to uncertainty and the second part is due to privately known heterogeneity.

For $t > 1$, the analysis is more complicated. For any $t > 1$, $V_t^u(X, D^{t-1}; \omega^T, d^T) < V_1^u(X; \omega^T, d^T)$, because the realized outcomes are informative about $X^{*}_u$. Agents also select $d^{t-1}$ based on their private information ($X^{*}_k$), which induces a selected distribution of $X^{*}_k$ (i.e., conditional on $(X, Y^{t-1}, D^{t-1})=(x,y^{t-1},d^{t-1}))$. Given these contributions of learning and selection to variance of $Y(\omega^T, d^T)$, there are several possible ways to quantify the relative importance of uncertainty. The following are three alternative decompositions that express total variance (conditional on some subset of observables) as the sum of a term that reflects uncertainty and another reflecting variance induced by private information ($X^{*}_k$),  
{\small \begin{align}
  \V&(Y(\omega^T, d^T) \mid D^{t-1} = d^{t-1}, X = x) \notag \\ &\qquad= V_t^u(d^{t-1}, x; \omega^T, d^T) \notag \\
  &\qquad + \V(E(Y(\omega^T, d^T) \mid \mathcal{I}_t) \mid D^{t-1} = d^{t-1}, X = x), \label{eq:vardecomp-cond} \\
  \V&(Y(\omega^T, d^T) \mid X = x) \notag \\ &\qquad = E(V_t^u(D^{t-1}, x; \omega^T, d^T)) + \V(E(Y(\omega^T, d^T) \mid \mathcal{I}_t) \mid X = x), 
  \label{eq:vardecomp-uncond} \\
  \V&(Y(\omega^T, d^T) \mid X = x) \notag \\ &\qquad = V_t^u(d^{t-1}, x; \omega^T, d^T) + \tilde{\V}(\tilde{E}(Y(\omega^T, d^T) \mid \mathcal{I}_t) \mid X = x). 
  \label{eq:vardecomp-counter}
\end{align}}
Decomposition \eqref{eq:vardecomp-cond} compares the variance of uncertainty to the total variance conditional on choosing the sequence $d^t$. These are natural parameters to consider, but the ratio, $V_t^u(d^t,x; \omega^T, d^T) / \V(Y(\omega^T, d^T) \mid D^t = d^t, X = x)$ reflects both the effect of learning in the numerator and selection in the denominator. 

Decomposition \eqref{eq:vardecomp-uncond} compares the total variance $Y(\omega^T, d^T)$ to the expected posterior variance of $Y(\omega^T, d^T)$ after $t$ periods. The expectation of $V^u(D^t, x; \omega^T, d^T)$ can be understood as the uncertainty that a randomly chosen person would have in period $t$ after observing their outcomes and endogenously choosing actions based on that information and their private information.

Finally, decomposition \eqref{eq:vardecomp-counter} is based on a counterfactual distribution. Here, $\tilde{E}$ and $\tilde{\V}$ represent the expectation and variance in a counterfactual distribution where $D^t$ is assigned randomly. This decomposition compares the variance in $Y(\omega^T, d^T)$ which is due to uncertainty vs. known heterogeneity among people randomly assigned to the choice sequence $d^t$.

\newpage

\subsection{Appendix to estimation section}\label{sec:est-app}

\subsubsection{Consistency of sieve MLE}\label{sec:est-consistency}
In this section we introduce conditions for the sieve maximum likelihood estimator defined in Equation \eqref{eq:sieve-mle} to be consistent for the true model parameter $\theta^*\in\Theta$. We begin by imposing smoothness restrictions on the unknown functions. To do so, given $\gamma>0$, $\omega\ge0$ and $\mathcal{X}$ a subset of a Euclidean space, let $\Lambda^{\gamma}(\mathcal{X})$ denote a H\"older space equipped with the H\"older norm $\|h\|_{\Lambda^{\gamma}}$ (that is, for $k$ the largest integer smaller than $\gamma$, $\Lambda^{\gamma}(\mathcal{X})$ is a space of functions $h\colon\mathcal{X}\rightarrow\mathbb{R}$ having at least $k$ continuous derivatives, the $k$th of which is H\"older continuous with exponent $\gamma-k$). Then define a weighted H\"older ball with radius $c\in(0,\infty)$ as $\Lambda_{c}^{\gamma,\omega}(\mathcal{X})=\left\{h\in\Lambda^{\gamma}(\mathcal{X})\colon\|h(\cdot)[1+\|\cdot\|_{E}^{2}]^{-\omega}\|_{\Lambda^{\gamma}}\le{c}\right\}$, where $\|\cdot\|_{E}$ is the Euclidean norm.

Without loss of generality, suppose that the CCP function ${h}_{t}(d^{t},x^{t},y^{t-1},x_{k}^{*})$ depends on $(d^{t},x^{t},y^{t-1})$ via some measurable vector-valued function $(d^{t},x^{t},y^{t-1})\mapsto{j}_t$ which is known up to $\left((\beta_{s},\lambda_{s},\sigma_{s})_{s=1}^{T},\Sigma_{u}(x_1)\right)$. This is without loss of generality since the function may be identity. Other examples include rational learning where ${j}_t\in\mathbb{R}^{p(p+3)/2+2}$ includes sufficient statistics for $X^{*}_u$ (i.e, the mean and variance), and a sort of myopia where  $j_t\in\mathbb{R}^{3+2}$ depends on the history only via the previous period $(d_{t-1},x_{t-1},y_{t-1})$. Write $J_t=(J_{1,t}^\intercal,J_{2,t}^\intercal)^\intercal$ and $X_t=(X_{1,t}^\intercal,X_{2,t}^\intercal)^\intercal$ where $J_{1,t},X_{1,t}$ are continuous random variables and $J_{2,t},X_{2,t}$ are random variables with finite support and, with some abuse of notation, redefine the CCP function as $h_t(j_{1,t},j_{2,t},x_{k}^{*})$. Define
\begin{align*}
  \mathcal{H}_{t}&=\Lambda_{c}^{\gamma_{1},\omega_{1}}\left(\Supp(X^{*}_{k})\times\Supp(J_{1,t})\right),\\
  \mathcal{F}&=\{f\colon\Supp(X^{*}_{k},X_{1,1})\rightarrow\mathbb{R}\big| F(\cdot,x_1)\text{ is c\`{a}dl\`{a}g }, F(x_{k}^{*},\cdot)\in \Lambda_{c}^{\gamma_{2},\omega_{2}}(\Supp(X_{1,1})) \}\\
  \mathcal{G}_{t}&=\Lambda_{c}^{\gamma_{3},\omega_{3}}\left(\Supp(X_{1,t+1})\times\Supp(Y_{t})\times\Supp(X_{1,t})\right).
\end{align*}

The use of a weighted Holder space enables us to allow the support of the continuous random variables to be unbounded. Although not required for consistency, Assumption \ref{assn:rate1} places restrictions on $(\gamma_{1},\gamma_{2},\gamma_{3})$, the parameters that govern the smoothness of the function classes. Next, to simplify notation we make the following assumption which strengthens Assumption \ref{assn:kl_independence}:
\begin{manualassumption}{E1}\label{assn:markov}
For any $t$, $F_{X_{t+1}|Y^{t},D^{t},X^{t}}=F_{X_{t+1}|Y_{t},D_{t},X_{t}}$, and $F_{X_U^*|X_1}=F_{X_U^*}$.
\end{manualassumption}

Define $k_{1,t}=|\Supp(J_{2,t})|$, $k_{2}=|\Supp(X_{2,1})|$, and $k_{3,t}=|\Supp((X_{2,t+1},D_{t},X_{2,t}))|$. Notice that $\Theta=\Theta_{1}\times\mathcal{H}_{1}^{k_{1,1}}\times\dots\times\mathcal{H}_{T}^{k_{1,T}}\times\mathcal{F}^{k_{2}}\times\mathcal{G}_{1}^{k_{3,1}}\times\dots\times\mathcal{G}_{T-1}^{k_{3,T-1}}$ and we denote an element of $\Theta$ as $\theta=(\theta_{1},{h}_{1},\dots,{h}_{T},f_{X^{*}},{g}_{1},\dots,{g}_{T-1})$. Define the norms on $\mathcal{H}_{t}^{k_{1,t}}$, $\mathcal{F}^{k_{2}}$ and $\mathcal{G}_{t}^{k_{3,t}}$ as follows:
\begin{align*}
\|{h}_{t}\|_{\infty,\omega_{1}}&=\sup_{j_2\in\Supp(J_{2,t})}\|{h}_{t}(\cdot,j_2,\cdot)[1+\|\cdot\|_{E}^{2}]^{-\omega_{1}}\|_{\infty},\\
\|F_{X^{*}}\|_{\infty,\omega_{2}}&=\sup_{x_{2}\in\Supp(X_{2,1})}\|F_{X^{*}}\left(\cdot,(\cdot,x_{2})\right)[1+\|\cdot\|_{E}^{2}]^{-\omega_{2}}\|_{\infty},\\
\|g_{t}\|_{\infty,\omega_{3}}&=\sup_{\substack{(x'_{2},d,x_2)\in\Supp(X_{2,t+1},D_{t},X_{2,t})}}\|g_{t}\left((\cdot,x'_{2});\cdot,d,(\cdot,x_2)\right)[1+\|\cdot\|_{E}^{2}]^{-\omega_{3}}\|_{\infty},
\end{align*}
where $\|\cdot\|_{\infty}$ is the uniform norm. Finally, define a metric $d$ on $\Theta$ as
\begin{equation*}
d(\theta,\tilde{\theta})=\|\theta_{1}-\tilde{\theta}_{1}\|_{E}+\sum_{t=1}^{T}\|{h}_{t}-\tilde{{h}}_{t}\|_{\infty,\tilde{\omega}_{1}}+\|F_{X^{*}}-\tilde{F}_{X^{*}} \|_{\infty,\tilde{\omega}_{2}}+\sum_{t=1}^{T-1}\|g_{t}-\tilde{g}_{t}\|_{\infty,\tilde{\omega}_{3}},
\end{equation*}
for scalars $\tilde{\omega}_{1},\tilde{\omega}_{2},\tilde{\omega}_{3}$. Now, let $\mathcal{H}_{n,t}$, $\mathcal{F}_{n}$ and $\mathcal{G}_{n,t}$ be sieve spaces for $\mathcal{H}_{t}$, $\mathcal{F}$ and $\mathcal{G}_{t}$ respectively. Then $\Theta_{n}=\Theta_{1}\times\mathcal{H}^{k_{1,1}}_{n,1}\times\dots\mathcal{H}^{k_{1,T}}_{n,T}\times\mathcal{F}^{k_{2}}_{n}\times\mathcal{G}^{k_{3,1}}_{n,1}\times\dots\times\mathcal{G}^{k_{3,T-1}}_{n,T-1}$ and 
\begin{equation*}
\frac{1}{n}\sum_{i=1}^{n}\ell(w_{i};\hat\theta)\ge\sup_{\theta\in\Theta_{n}}\frac{1}{n}\sum_{i=1}^{n}\ell(w_{i};\theta)-o_{p}(1/n).
\end{equation*}

\begin{manualassumption}{E2}\label{assn:e1}
$\theta^{*}\in\Theta$ and $(\Theta,d)$ is compact.
\end{manualassumption}

\begin{manualassumption}{E3}\label{assn:e3}
For each $n\ge1$, $\Theta_{n}\subseteq\Theta_{n+1}\subseteq\Theta$ and $\Theta_{n}$ is compact under $d$. As $n\rightarrow\infty$, $\min_{\theta\in\Theta_{n}}d(\theta,\theta_{0})\rightarrow0$.
\end{manualassumption}

\begin{manualassumption}{E4}\label{assn:e4}
$\text{E}[\ell(W ,\theta)]$ is continuous at $\theta=\theta^{*}$ 
\end{manualassumption}

\begin{manualassumption}{E5}\label{assn:e5}\mbox{}
\begin{enumerate}
\item[(i)] For each $n$, $\text{E}[\sup_{\theta\in\Theta_{n}}\lvert\ell(W ,\theta)\rvert]$ is finite.
\item[(ii)] There is a non-zero $s<\infty$ and integrable random variable $g(W)$
such that $\forall~\theta,\tilde{\theta}\in\Theta_{n}$, ${d}(\theta,\tilde{\theta})<\delta\implies\lvert\ell(W,\theta)-\ell(W,\tilde{\theta})\rvert\le\delta^{s}g(W)$.
\item[(iii)] For all $\delta>0$, $\log{N}(\delta^{1/s},\Theta_{n},d)=o(n)$.
\end{enumerate}
\end{manualassumption}

The identification assumptions imply $\theta^{*}=\arg\max_{\theta\in\Theta}\text{E}[\ell(W,\theta)]$ and for all $\theta\in\Theta\setminus\{\theta^{*}\}$, $\text{E}[\ell(W,\theta^{*})]\ge\text{E}[\ell(W,\theta)]$.  By assuming compactness of $\Theta$, we ensure that $\theta^{*}$ is a well-separated maximum of $\text{E}[\ell(W,\theta)]$. Assumption \ref{assn:e3} requires the sieve space $\Theta_{n}$ to be a good approximation to $\Theta$. Assumption \ref{assn:e4} requires the population criterion to be continuous.  Finally, Assumption \ref{assn:e5} is similar to Condition 3.5M in \citet{chen2007large}.

Theorem \ref{thm:consistency} follows from Remark 3.3 in \citet{chen2007large}, so its proof is omitted.

\subsubsection{Plug-in sieve estimator}
\label{App:sieve_plugin}

We first assume a linear sieve space and limit its complexity. 
\begin{manualassumption}{E6}{\label{assn:rate1}}
(i) $\mathcal{H}_{n,t}$, $\mathcal{F}_{n}$ and $\mathcal{G}_{n,t}$ are linear sieves of length $M_{Hn,t}$, $M_{Fn}$ and $M_{Gn,t}$ respectively, where $M_{Hn,t}=O(n^{\frac{1}{2\gamma_{1}/(1+\dim(J_{1,t}))+1}})$, $M_{Fn}=O(n^{\frac{1}{2\gamma_{2}/(1+\dim(X_{1,1}))+1}})$, and $M_{Gn,t}=O(n^{\frac{1}{2\gamma_{3}/(\dim(X_{1,t+1})+1+\dim(X_{1,t}))+1}})$. (ii) $\min\left\{\frac{\gamma_{1}}{1+\dim(J_{1,t})},\frac{\gamma_{2}}{1+\dim(X_{1,1})},\frac{\gamma_{3}}{\dim(X_{1,t+1})+1+\dim(X_{1,t})}\right\}>1/2$.
\end{manualassumption}
Assumption \ref{assn:rate1} controls the rate at which the number of sieve terms grow. To achieve this, part (i) of Assumption \ref{assn:rate1} requires that the nonparametric functions have adequate smoothness. In applied work, one may focus on discrete $X_t$ and posit a parametric model for ${h}_{t}$, in which case the above restrictions are milder.

The next assumption strengthens \ref{assn:e3} and ensures that the number of sieve terms grows sufficiently quickly.
\begin{manualassumption}{E7}{\label{assn:rate2}}
$\min_{\theta\in\Theta_{n}}d(\theta,\theta^*)=o(n^{-1/4})$.
\end{manualassumption}

Assume $\ell$ is pathwise differentiable and define an inner product on $\Theta$ as
\begin{equation}\label{eq:inner-product}
\left\langle\theta_{1}-\theta^*,\theta_{2}-\theta^*\right\rangle=-\frac{\partial^{2}}{\partial \tau_{1} \partial \tau_{2}} E\left[ \ell\left(W,\theta^*+\tau_{1}\left(\theta_{1}-\theta^*\right)\right.\left.+\tau_{2}\left(\theta_{2}-\theta^*\right)\right)\right]\left.\right|_{\tau_{1}=0, \tau_{2}=0},
\end{equation}
for $\theta_1,\theta_2\in\Theta$. The corresponding norm for $\theta\in\Theta$ is
\begin{equation}\label{eq:norm}
  \left\|\theta-\theta^*\right\|^{2}\coloneqq-\left.\frac{\partial^{2}}{\partial \tau^{2}} E\left[ \ell\left(W,\theta^*+\tau\left(\theta-\theta^*\right)\right)\right]\right|_{\tau=0}.
\end{equation}

\begin{manualassumption}{E8}\label{assn:rate4}
There is $C_{1}>0$ such that for all small $\varepsilon>0$
$$
\sup _{\left\{\theta \in \Theta_{n}:\left\|\theta-\theta^{*}\right\| \leqslant
\varepsilon\right\}} \operatorname{Var}\left(\ell\left(W,\theta\right)-\ell\left(W,\theta^{*}\right)\right) \leqslant C_{1}
    \varepsilon^{2}
$$
\end{manualassumption}
\begin{manualassumption}{E9}\label{assn:rate5}
For any $\delta>0$, there exists a constant $s \in(0,2)$ such that
$$
\sup _{\left\{\theta \in \Theta_{n}:\left\|\theta-\theta^{*}\right\| \leqslant
\delta\right\}}\left|\ell\left(W,\theta\right)-\ell\left(W,\theta^{*}\right)\right| \leqslant \delta^{s} U\left(W\right)
$$
with $E\left(\left[U\left(W\right)\right]^{\gamma}\right) \leqslant C_{2}$ for some $\gamma \geqslant 2$.
\end{manualassumption}
The following theorem is now a consequence of Theorem 3.2 in \citet{chen2007large} or Theorem 1 in \cite{shen1994convergence}.

\begin{theorem}\label{thm:rate}
Let $(Y_{i,t},D_{i,t},X_{i,t}\colon t=1,\ldots,{T})_{i=1}^n$ be i.i.d. data where $T\ge{2p}+1$ and Assumptions \ref{assn:kl_independence}-\ref{assn:kl_regularity} and Assumptions
\ref{assn:markov}-\ref{assn:rate5} hold. Then $\|\hat{\theta}-\theta^{*}\|=o_{p}(n^{-1/4})$.
\end{theorem}

Given the preceding result, we focus on a a shrinking neighborhood of $\theta^{*}$. Let
\begin{equation*}
\mathcal{N}_{0}\coloneqq\left\{\theta\in\Theta\colon\|\theta-\theta^{*}\|=o(n^{-1/4}),~d(\theta,\theta^{*})=o(1)\right\},
\end{equation*}
and $\mathcal{N}_{n}\coloneqq\mathcal{N}_{0}\cap\Theta_{n}$. Define $\theta_{n}^*=\argmin_{\theta\in\mathcal{N}_{n}}{\|\theta-\theta^*\|}$. Let $\mathcal{V}$ denote the closed (under $\|\cdot\|$) linear span of $\mathcal{N}_{0}$ centered at $\theta^*$,  and define $\mathcal{V}_{n}$ as the analogous closure of $\mathcal{N}_n$.

Then we define a linear approximation to $\ell(W,\theta)-\ell(W,\theta^{*})$ as the directional derivative of $\ell$ at $(W,\theta^{*})$ in the direction $(\theta-\theta^{*})$:
\begin{equation*}
 \frac{\partial \ell\left(W,\theta^*\right)}{\partial \theta}[\theta-\theta^{*}]\coloneqq\left.\frac{\partial \ell\left(W,\theta^*+\tau (\theta-\theta^{*})\right)}{\partial \tau}\right|_{\tau=0}.
  \end{equation*}
Likewise, let $\frac{\partial f\left(\theta^*\right)}{\partial\theta}[v]=\left.\frac{\partial f\left(\theta^*+\tau v\right)}{\partial\tau}\right|_{\tau=0}$ for any $v \in \mathcal{V}$.

\begin{manualassumption}{E10}\label{assn:regularity} Let $\mathcal{T}$ be an epsilon ball about $0\in\mathbb{R}$.  (i) For all $\theta \in \mathcal{N}_{0}$ and $W$, the derivative $\partial\ell\left(W,\theta^*+\tau(\theta-\theta^*)\right) / \partial \tau$ exists for all $\tau \in\mathcal{T}$; (ii) for all $\theta \in \mathcal{N}_{0}$, $\text{E}\left[\ell\left(W,\theta^*+\tau\left(\theta-\theta^*\right)\right)\right]$ is finite for each $\tau\in\mathcal{T}$; (iii) for all $\theta \in \mathcal{N}_{0}$, $\text{E}\left[\sup _{\tau \in \mathcal{T}}\left|\frac{\partial}{\partial \tau} \ell\left(W,\theta^*+\tau\left[\theta-\theta^*\right]\right)\right|\right]$ $<\infty$.
\end{manualassumption}

Assumption \ref{assn:regularity} provides sufficient conditions for the set $\mathcal{V}$ to be a Hilbert space under $\langle\cdot ,\cdot \rangle$.\footnote{See \citet[p. 642]{chen2014sieve2}.} Define $v_{n}^{*}$ to be the Riesz representer of $\frac{\partial{f}(\theta^{*})}{\partial\theta}[\cdot]$ on $\mathcal{V}_{n}$, which exists under Assumption \ref{assn:smooth-functional}.

\begin{manualassumption}{E11}\label{assn:smooth-functional}
(i) $v\mapsto\frac{\partial{f}(\theta^*)}{\partial\theta}[v]$ is a linear functional. (ii) If $\lim_{n\rightarrow\infty}\|v_{n}^{*}\|$ is finite then $\|v_{n}^{*}-v^{*}\|\times\|\theta_n^*-\theta^*\|=o(n^{-1/2})$ where $v^{*}$ is the limit of $v_{n}^{*}$. Otherwise ${\left|\frac{\partial{f}(\theta^*)}{\partial\theta}[\theta_{n}^*-\theta^*]\right|}/{\|v_{n}^{*}\|}=o(n^{-1/2})$. (iii) $\sup_{\theta\in\mathcal{N}_{0}}\frac{\left|f(\theta)-f(\theta^*)-\frac{\partial{f}(\theta^*)}{\partial\theta}[\theta-\theta^*]\right|}{\|v_{n}^{*}\|}=o(n^{-1/2})$.
\end{manualassumption}

Assumption \ref{assn:smooth-functional} imposes some restrictions on the functional of interest $\theta\mapsto f(\theta)$. Part (i) imposes that the directional derivative is a linear functional, a mild condition that is satisfied by our examples in Section \ref{sec:estimation}. Part (ii) is a restriction on the growth rate of the dimension of the sieve space. Part (iii) restricts the linear approximation error of $f(\cdot)$ in a neighborhood of $\theta^*$, for which sufficient conditions could be stated in terms of the smoothness of $f(\cdot)$ and the growth rate of the dimension of the sieve space. See \citet{chen2014sieve2} for further discussion.

Let $u_{n}^{*} \coloneqq \frac{v_{n}^{*}}{\left\|v_{n}^{*}\right\|}$, $\varepsilon_{n}=o\left(n^{-1 / 2}\right)$ and $\mu_{n}\{g(\bm{W})\} \coloneqq n^{-1}\sum_{i=1}^{n}\left[g\left(W_{i}\right)-\mathrm{E}[g\left(W_{i}\right)]\right]$ denote the centered empirical process indexed by the function $g$.
\begin{manualassumption}{E12}\label{assn:smooth-criterion}
$\mu_{n}\{\frac{\partial \ell\left(\bm{W},\theta^*\right)}{\partial \theta}[v]\}$ is linear in $v\in\mathcal{V}$.
\begin{equation*}
 \sup _{\theta \in \mathcal{N}_{n}} \mu_{n}\left\{\ell\left(\bm{W},\theta \pm \varepsilon_{n} u_{n}^{*}\right)-\ell(\bm{W},\theta)-\frac{\partial \ell\left(\bm{W},\theta^*\right)}{\partial \theta}\left[\pm \varepsilon_{n} u_{n}^{*}\right]\right\} 
= O_{p}\left(\varepsilon_{n}^{2}\right).
\end{equation*}
 For some positive sequence $\eta_{n}\rightarrow0$,
\begin{equation*}
\sup _{\theta \in \mathcal{N}_{n}}  \left| E\left[\ell(W,\theta)-\ell\left(W,\theta\pm \varepsilon_{n} u_{n}^{*}\right)\right] 
-\frac{\left\|\theta \pm \varepsilon_{n} u_{n}^{*}-\theta^*\right\|^{2}-\left\|\theta-\theta^*\right\|^{2}}{2}\left(1+O\left(\eta_{n}\right)\right) \right|=O\left(\varepsilon_{n}^{2}\right).
  \end{equation*}
\end{manualassumption}

\begin{manualassumption}{E13}\label{assn:CLT}
$\sqrt{n} \mu_{n}\left\{\frac{\partial \ell\left(\bm{W},\theta^*\right)}{\partial \theta}\left[u_{n}^{*}\right]\right\} \rightarrow_{d} N(0,1)$
\end{manualassumption}

Theorem \ref{thm:normality} is a direct application of Lemma 2.1 in \citet{chen2014sieve} so its proof is omitted.

\subsection{Appendix to implementation and Monte Carlo simulations section}

\subsubsection{Implicit differentiation}\label{sec:diff}
For implementing the estimator, it can be useful to input the gradient of the objective function. In this section, we show how our profiling approach and choice of sieve space simplify this task. Recall that in Section \ref{sec:estim-impl}, the profile log likelihood function with our proposed sieve space for $F_{X_k^*}$ is
\begin{equation*}
\ell^p(\theta^c)\coloneqq\sum_{i=1}^n\log\sum_{s=1}^{q_n}\omega_s(\theta^c)~\ell^c(w_i,\bar{x}_{n,s}^*;\theta^c),
\end{equation*}
where $\omega(\theta^c)=\argmax_{\omega\in\Delta(q_n)}\sum_{i=1}^n\log\sum_{s=1}^{q_n}\omega_s~\ell^c(w_i,\bar{x}_{n,s}^{*};\theta^c)$ is the solution to the inner problem for a fixed $\theta^{c}$. Given an analytical expression for $\ell^c(w_i,x_k^*;\theta^c)$\footnote{Given the analytical expression for $\ell^c$, we use the software Google JAX to compute the derivative via autodifferention.}, the challenge of computing the gradient of $\ell^p(\theta^c)$ reduces to finding the Jacobian of $\omega(\theta^c)$ (i.e., $\frac{\partial}{\partial (\theta^c)^\intercal}\omega(\theta^c)$), which is defined implicitly by the Karush-Kuhn-Tucker (KKT) conditions of the inner optimization problem. In the following, we derive an analytical expression for $\frac{\partial}{\partial (\theta^c)^\intercal}\omega(\theta^c)$ in terms of $\ell^c(w_i,x_k^*;\theta^c)$, $\frac{\partial}{\partial \theta^c} \ell^c(w_i,x_k^*;\theta^c)$, and $\omega(\theta^c)$.

Proposition 3.3 in \cite{kim2020fast} shows that $\omega(\theta^c)$ can be equivalently expressed as $\argmax_{\omega\ge0}\{\sum_{i=1}^n\log\sum_{s=1}^{q_n}\omega_s~\ell^c(w_i,\bar{x}_{n,s}^{*};\theta^c)+\sum_{s=1}^{q_n}\omega_s\}$, where $\omega\ge0$ means $\omega_{s}\ge0$ for all $s=1,\dots,q_n$. Letting $\lambda\in\mathbb{R}^{q_n}$ be the dual parameter corresponding to the constraint $\omega \ge 0$, and  $\ell_i^{c}(\theta^c)\coloneqq(\ell^c(w_i,\bar{x}_{n,s}^*;\theta^c)\colon s=1,\ldots,q_n)$, the equality constraints in the KKT conditions of this problem are,
$$
0_{2q_n\times 1} = \begin{pmatrix}\sum_{i=1}^n\frac{1}{\omega^\intercal \ell_i^c(\theta^c)}\ell_i^c(\theta^c) + 1_{q_n} + \lambda \\ ~\lambda\circ\omega\end{pmatrix},
$$
where $\circ$ is the Hadamard product. By definition, these constraints are identically zero for all $\theta^c$, so under an implicit function theorem, $\frac{d}{d(\theta^c)^\intercal}\omega(\theta^c)=-G_1(\theta^c)^{-1} G_2(\theta^c)$,\footnote{$G_1$ and $G_2$ are the partial derivatives of right hand side of the previous equation with respect to $(\omega, \lambda)$ and $\theta^c$ respectively, evaluated at $\omega(\theta^c)$ and $\lambda(\theta^c)$.} where
$$
G_1(\theta^c) = \begin{pmatrix}
    \sum_{i=1}^n\frac{1}{(\omega(\theta^c)^\intercal\ell_i^{c}(\theta^c))^2}\ell_i^c(\theta^c)(\ell_i^c(\theta^c)^\intercal & I_{q_n\times q_n} \\
    \text{diag}(\lambda(\theta^c)) & \text{diag}(\omega(\theta^c))
\end{pmatrix},
$$
and
\begin{align*}
G_2(\theta^c)=&
 \begin{pmatrix}\sum_{i=1}^n \left(\frac{\frac{\partial}{\partial (\theta^c)^\intercal}\ell_i^c(\theta^c)}{\omega(\theta^c)^\intercal \ell_i^c(\theta^c)}- \frac{\ell_i^c(\theta^c)\omega(\theta^c)^\intercal \frac{\partial}{\partial (\theta^c)^\intercal}\ell_i^c(\theta^c)}{(\omega(\theta^c)^\intercal \ell_i^c(\theta^c))^2}\right) \\~ 0_{q_n\times\dim(\theta^c)} \end{pmatrix}
\end{align*}
Finally, note that the KKT conditions imply that $\lambda(\theta^c) = -1_{q_n}-\sum_{i = 1}^n \frac{\ell_i^c(\theta^c)}{\omega(\theta^c)^\intercal\ell_i^{c}(\theta^c)}$.

\subsubsection{Details on DGP}
\label{App:DGP_det}

This section gives further details on the DGP used for Monte Carlo simulations discussed in Section \ref{sec:numer-simul}. The values of the finite parameters used in the DGP are given in the table below. 
\begin{table}[H]
  \centering
  \begin{threeparttable}
  \label{tab:sim_main_dgp_params}
  \begin{tabular}{lllll}
    \toprule
    $\color{lightgray} \alpha_{1,1} = 0  $                    & $\gamma^{(1)}_{1,1} = -0.5$ &  $\gamma^{(2)}_{1,1} = -0.58$ & $\color{lightgray} \lambda^{u}_{1,1} = 1$                        & $\lambda^{k}_{1,1} = 0.3$ \\
    $\alpha_{2,1} = 0.1$ & $\gamma^{(1)}_{2,1} = -0.8$ & $\gamma^{(2)}_{2,1} = -0.83$ & $ \lambda^{u}_{2,1} = 1.05$ & $\lambda^{k}_{2,1} = 0.35$ \\
    $\alpha_{3,1} = 0.2$ & $\gamma^{(1)}_{3,1} = 0.12$ & $\gamma^{(2)}_{3,1} = -0.83$ & $  \lambda^{u}_{3,1} = 1.01$ & $ \lambda^{k}_{3,1} = 0.33$ \\
    $\sigma^2_{1} = 0.5$   &    \\    
    \midrule
    $\alpha_{1,2} = -0.1$ & $\gamma^{(1)}_{1,2} = 0.13$ & $\gamma^{(2)}_{1,2} = 0.71$ & $\lambda^{u}_{1,2} = 0.4$ & $\color{lightgray} \lambda^{k}_{1,2} = 1$                      \\
    $\alpha_{2,2} = -0.22$ & $\gamma^{(1)}_{2,2} = 0.89$ & $\gamma^{(2)}_{2,2} = -0.36$ & $\lambda^{u}_{2,2} = 0.36$ & $ \lambda^{k}_{2,2} = 1.05$ \\        
    $\alpha_{3,2} = -0.33$ & $\gamma^{(1)}_{3,2} = 0.32$ & $\gamma^{(2)}_{3,2} = -0.36$ & $\lambda^{u}_{3,2} = 0.44$ & $\lambda^{k}_{3,2} = 1.02$ \\    
    $\sigma^2_2 = 0.7$   &&&&  \\
    \midrule
    $\sigma^2_u = 1.5$ & $\rho = 2.0$ & $\kappa = 0.5$ && \\
    \bottomrule
  \end{tabular}
  \caption{Finite parameter values}
  \end{threeparttable}
\end{table}

\subsubsection{DGP with risk aversion}
\label{app:dgp_risk_aversion}

In this section, we present results from an alternative DGP in which agents maximize their expected utility in each period which incorporates risk aversion, through constant relative risk aversion (CRRA) preferences, and subjective (possibly biased) beliefs. The expected utility that individual $i$ derives from choice $d$ in period $t$ is given by:
\begin{align*}
  v_{i,t}(d) := \mathcal{E}_{i,t} \left( \frac{Y_{i,t}(d)^{1 - \chi}}{1 - \chi}\right) + \eta_{i,t}(d)
\end{align*}
where $\mathcal{E}_{i,t}$ denotes the expectation under individual $i$'s subjective beliefs over $X^{*}_{u,i}$, given the information up to period $t$. $\eta_{i,t}(d)$ are independent preference shocks, which are supposed to follow an Extreme Value Type 1 distribution.

We assume that individuals' subjective beliefs over $X^{*}_{u,i}$ in time $t$ are distributed $N(\mu_{i,t} +  \delta X^{*}_{k, i}, \Sigma_{i,t})$ where $\mu_{i,t}, \Sigma_{i,t}$ are the correct posterior mean and variance of $X^{*}_{u,i}$ given the information up to period $t - 1$. This subjective belief process allows agents to have biased beliefs that can be correlated with the known part of their unobserved heterogeneity, $X^{*}_{k,i}$.

Under this specification, the expected utility has the following analytical form,
\begin{equation}
  v_{i,t}(d) = \frac{\exp\bigg(\mu_{i,t}(d)(1 - \chi) + \frac{1}{2}\sigma_{i,t}(d)(1 - \chi)^2\bigg)}{1 - \chi} + \eta_{i,t}(d) \label{eq:crra-utility-form}
\end{equation} 
where $\mu_{i,t}(d)$ ($\sigma_{i,t}(d)$) denote the subjective mean (variance) of $\log(Y_{i,t}(d))$.

A naive approach to estimating ${v}_{i,t}(d)$ nonparametrically would be to use a tensor product of polynomials $(X^{*}_k, X, Y^{t-1}, D^{t-1})$ as the sieve space. That is, for a univariate random variable $X$, let $\mathcal{P}_q(X) = \text{sp}(\{1, X, \ldots, X^q\})$. Assume $D_t$ is binary, and let $\delta_{t} = 1(D_t = 1)$, then the sieve space is, 
$$
\mathcal{P}_q(X^{*}_k) \otimes \mathcal{P}_q(X_1) \otimes \cdots \otimes \mathcal{P}_q(Y_1) \otimes \mathcal{P}_q(\delta_1) \otimes \cdots \otimes \mathcal{P}_q(Y_{t-1}) \otimes \mathcal{P}_q(\delta_{t-1}).
$$
For an $q$-order polynomial, the number of terms would be $(q+1)^{3} + (q+1)^{5} + (q+1)^{7}$, which grows very quickly in practical terms.

The alternative approach we consider here is to use the following approximation
$$
{v}_{i,t}(d) = \varphi \left(\sum_{h \in \mathcal{D}^{t-1}} 1(D^{t-1} = h) (\pi_{t, h, d, 0} + \pi_{t, h, d, 1}^{\intercal} X + \pi_{t, h, d, 2}X^{*}_k + \pi_{t, h, d, 3}^{\intercal} Y_i^{t-1} ) \right)
$$
for some unknown function $\varphi$. Since the argument of $\varphi$ is scalar-valued, this means that the nonparametric estimation problem is greatly simplified to estimating a scalar-valued function. For this we use the sieve space of polynomials, with the order growing at the rate of $n^{1/3}$ with 3 terms with $n = 500$ and $6$ terms for $n = 4{,}000$. Our choice of approximation is motivated by the fact that under Lemma \ref{lemma:conjugate} and Equation \ref{eq:crra-utility-form}, there is a set of $\pi$ parameters such that this equality holds, with $\varphi(\cdot)=\frac{1}{1-\chi}\exp(\cdot)$.

The finite parameters are the same as in our baseline simulations considered in Section~\ref{sec:numer-simul}, with the added risk aversion parameter $\chi$, which we set to $1.5$. $X^{*}$ and $X$ are generated from the same distributions as in the DGP considered in Section~\ref{sec:numer-simul}.

With the additional $\pi$ parameters to estimate, the $\theta^c$ has a total of 103 parameters. Given this large number of parameters to estimate, we expect $n = 250$ to be too small a sample size to perform well, and begin the Monte Carlo simulations with a sample size of $n = 500$. The large number of parameters to estimate in $\theta^c$ results in longer but still manageable computational times, which are reported in Table \ref{tab:mc_comp_times_crra}.

\begin{table}[h]
  \centering
    \begin{tabularx}{\textwidth}{Xrrrr}
    \toprule

        & \multicolumn{1}{c}{$n = 500$}
        & \multicolumn{1}{c}{$n = 1{,}000$}
        & \multicolumn{1}{c}{$n = 2{,}000$} 
        & \multicolumn{1}{c}{$n = 4{,}000$} \\

    \midrule

    Time (minutes) 
    &  3 &  7.5 &  19.5 &  56  \\
    
    \bottomrule
  \end{tabularx}

    \caption{Time to compute the estimator: DGP with risk aversion. Computational times were obtained using an Intel Core i9-12900K CPU, and are computed as the average over 200 simulations.}
    \label{tab:mc_comp_times_crra}
\end{table}

The results of the Monte Carlo simulations are presented in Table \ref{tab:mc_finite_params_risk_aversion} and Figure \ref{fig:mc_factor_5_95_crra_risk_aversion}. Despite the increased complexity of the model, our estimation procedure exhibits finite sample performances similar to the DGP considered in Section~\ref{sec:numer-simul}.

\begin{table}[h]
  \centering
    \begin{tabular}{lrrrrrrrr}
    \toprule

        & \multicolumn{2}{c}{n = 500}
        & \multicolumn{2}{c}{n = 1,000}
        & \multicolumn{2}{c}{n = 2,000} 
        & \multicolumn{2}{c}{n = 4,000} \\
        
        & Bias$^2$ & Var     
        & Bias$^2$ & Var
        & Bias$^2$ & Var
        & Bias$^2$ & Var \\
    
    \midrule
        
    $\alpha_{1, 2}$
        & 66.15 & 38.25 & 18.40 & 20.20 & 3.97 & 12.19 & 0.05 & 7.69  \\
    $\alpha_{2, 1}$
        & 0.17 & 28.07 & 0.05 & 12.99 & 0.08 & 5.50 & 0.05 & 2.10  \\
    $\alpha_{2, 2}$
        & 69.24 & 42.16 & 18.40 & 23.49 & 3.25 & 14.33 & 0.00 & 9.20  \\
    $\alpha_{3, 1}$
        & 1.29 & 24.63 & 0.07 & 9.98 & 0.00 & 4.73 & 0.00 & 1.83  \\
    $\alpha_{3, 2}$
        & 68.62 & 42.86 & 23.69 & 21.80 & 3.41 & 13.62 & 0.01 & 8.28  \\
    $\gamma_{1,1}^{(1)}$
        & 0.08 & 6.61 & 0.05 & 3.30 & 0.01 & 1.72 & 0.02 & 0.95  \\
    $\gamma_{1,2}^{(1)}$
        & 0.12 & 8.29 & 0.09 & 3.55 & 0.02 & 1.64 & 0.01 & 0.78  \\
    $\gamma_{2,1}^{(1)}$
        & 0.03 & 7.69 & 0.08 & 3.81 & 0.04 & 2.11 & 0.02 & 1.08  \\
    $\gamma_{2,2}^{(1)}$
        & 0.21 & 9.49 & 0.25 & 4.13 & 0.06 & 2.18 & 0.03 & 0.79  \\
    $\gamma_{3,1}^{(1)}$
        & 0.14 & 5.52 & 0.03 & 2.52 & 0.01 & 1.38 & 0.02 & 0.72  \\
    $\gamma_{3,2}^{(1)}$
        & 0.08 & 9.43 & 0.11 & 4.03 & 0.03 & 1.84 & 0.02 & 0.83  \\
    $\gamma_{1,1}^{(2)}$
        & 1.65 & 35.50 & 0.00 & 12.36 & 0.22 & 5.58 & 0.01 & 2.75  \\
    $\gamma_{1,2}^{(2)}$
        & 0.09 & 28.70 & 0.09 & 11.52 & 0.16 & 6.99 & 0.06 & 3.19  \\
    $\gamma_{2,1}^{(2)}$
        & 1.47 & 31.77 & 0.00 & 12.37 & 0.06 & 5.50 & 0.03 & 2.79  \\
    $\gamma_{2,2}^{(2)}$
        & 0.08 & 28.45 & 0.11 & 13.67 & 0.23 & 7.50 & 0.11 & 3.25  \\
    $\gamma_{3,1}^{(2)}$
        & 0.73 & 25.40 & 0.02 & 11.07 & 0.13 & 4.71 & 0.01 & 2.65  \\
    $\gamma_{3,2}^{(2)}$
        & 0.17 & 29.53 & 0.00 & 14.60 & 0.16 & 7.89 & 0.09 & 3.35  \\
    $\lambda_{1, 1}^{k}$
        & 0.34 & 20.38 & 1.18 & 6.84 & 0.02 & 4.11 & 0.01 & 1.71  \\
    $\lambda_{2, 1}^{k}$
        & 0.18 & 21.01 & 2.41 & 9.54 & 0.42 & 5.21 & 0.09 & 1.91  \\
    $\lambda_{2, 2}^{k}$
        & 0.18 & 9.49 & 0.00 & 3.31 & 0.01 & 1.60 & 0.01 & 0.80  \\
    $\lambda_{3, 1}^{k}$
        & 0.45 & 17.32 & 1.53 & 8.13 & 0.15 & 4.25 & 0.01 & 1.53  \\
    $\lambda_{3, 2}^{k}$
        & 0.03 & 10.43 & 0.21 & 3.97 & 0.01 & 2.22 & 0.01 & 1.10  \\
    $\lambda_{1, 2}^{u}$
        & 0.11 & 6.31 & 0.03 & 2.65 & 0.00 & 1.23 & 0.00 & 0.52  \\
    $\lambda_{2, 1}^{u}$
        & 0.05 & 3.54 & 0.04 & 1.41 & 0.01 & 0.78 & 0.01 & 0.43  \\
    $\lambda_{2, 2}^{u}$
        & 0.09 & 8.36 & 0.01 & 3.61 & 0.00 & 1.65 & 0.01 & 0.69  \\
    $\lambda_{3, 1}^{u}$
        & 0.06 & 3.89 & 0.02 & 1.44 & 0.01 & 0.60 & 0.00 & 0.33  \\
    $\lambda_{3, 2}^{u}$
        & 0.35 & 9.16 & 0.15 & 4.34 & 0.00 & 1.90 & 0.01 & 0.87  \\
    $\sigma^2(1)$
        & 0.15 & 0.68 & 0.01 & 0.36 & 0.01 & 0.17 & 0.00 & 0.07  \\
    $\sigma^2(2)$
        & 0.06 & 0.24 & 0.00 & 0.15 & 0.00 & 0.07 & 0.00 & 0.03  \\
    $\sigma^2_u$
        & 1.38 & 19.53 & 0.02 & 6.64 & 0.01 & 3.74 & 0.00 & 1.83  \\
        
    \bottomrule
  \end{tabular}
    \caption{Simulation results for estimation of finite dimensional parameters. \footnotesize{Note: `Bias$^2$' and `Var' refer to the average empirical squared bias and variance scaled by $1{,}000$, respectively, computed over 200 simulations.}}
    \label{tab:mc_finite_params_risk_aversion}
\end{table}

\begin{figure}[h]
    \centering
    \input{paper-inputs/quantiles_95_r045.tex}
    \caption{Quantiles of Estimator of $q_{\alpha}[X^{*}_k]$ under DGP with risk aversion. \footnotesize{Note: The red line shows the true distribution of $X^{*}_k$. The blue lines show the mean, and the $5$th and $95$th percentiles of the simulated distribution of the estimator of $q_{\alpha}[X^{*}_k]$ for each sample size.}}
    \label{fig:mc_factor_5_95_crra_risk_aversion}
\end{figure}

\FloatBarrier

\subsection{Appendix to the empirical illustration}

\subsubsection{Sample size after restrictions}

\begin{table}[H]
\begin{tabular*}{\linewidth}{@{\extracolsep{\fill}}llrrrr}

\toprule
& \multicolumn{2}{c}{Full Sample} & \multicolumn{2}{c}{Full-time Workers} \\
\cmidrule(lr){2-3}\cmidrule(lr){4-5}   
& Observations & Share & Observations & Share \\ 

\midrule\addlinespace[2.5pt]
Male & & & & \\
\quad Blacks & 891 & 0.10 & 273 & 0.10 \\
\quad Hispanics & 806 & 0.09 & 352 & 0.13 \\
\quad Whites & 2,031 & 0.23 & 965 & 0.37 \\
\addlinespace[2.5pt]
Female & & & & & \\
\quad Blacks & 949 & 0.11 & 229 & 0.09 \\
\quad Hispanics & 731 & 0.08 & 224 & 0.09 \\
\quad Whites & 1,786 & 0.20 & 571 & 0.22 \\
\bottomrule
\end{tabular*}
\caption{Sample sizes in subsamples defined by gender, race/ethnicity and work status.}
\label{tab:sample_size}
\end{table}

\subsubsection{Specification with college graduation}
\label{sec:appendix_extended_model}

In this section, we explore the robustness of the main findings of the application to an extended specification that includes educational attainment  as a covariate in the potential wage equation.\footnote{Specifically, we include a binary variable for graduation from a four-year university.} Education level enters the outcome equation additively, and we allow the distribution of $X^*_k$, and the occupational assignment function $h_t$ to depend arbitrarily on the educational level. While the estimation of this model has the advantage of shedding some light on how college education affects assignment probabilities to occupations and selection on the latent factor $X^*_k$, the main results of our variance decomposition are robust to this alternative specification. We conclude from this exercise that in our baseline model, the scalar latent variable $X_k^*$ captures the combined effect of college education and pre-college ability in a way that appears to be flexible enough to account for the uncertainty individuals face over their future earnings. 

The extended model includes college graduation as a covariate, which is allowed to depend arbitrarily on the known heterogeneity component $X^{*}_k$. The potential outcome equation can then be written as:

$$Y_{t}(d) = \beta_{t,d}'X +X^{*}_{k} \lambda_{t,d}^k + X^{*}_{u} \lambda_{t,d}^u +\epsilon_{t}(d),$$

where $X = (1, X^c)$ is a two-dimensional vector of a constant and a binary variable for college graduation ($X^c$). Since we start modeling choices at age $27$, we assume that college graduation is realized before then, but is allowed to flexibly depend on $X_k^*$. As a result, we estimate two conditional distributions for the known heterogeneity component, $F_{X^*_k \mid X^c = 0}$ and $F_{X^*_k \mid X^c = 1}$.

Occupational choice probabilities can now also arbitrarily depend on college graduation $X^c$, and are given by:
$$h_t ((1, D^{t-1}), X^c, Y^{t-1}, X_k^*) := P(D_t = 1 \mid X^c, Y^{t-1}, D^{t-1}, X_k^*).$$

In practice, we implement this specification using the same sieve space as in our baseline specification for each of the conditional distributions of $X^*_k$. Specifically, we estimate the CCPs using a similar functional form as before, with $h_t ((1,D^{t-1}), X^c, Y^{t-1}, X_k^*) = \Lambda (\phi_t(X_k^*, X^c, Y^{t-1}, D^{t-1}))$, and
\begin{multline*}
\phi_t (X_k^*, X^c, Y^{t-1}, D^{t-1}) = \\\sum_{d^{t-1} \in \{0, 1\}^{t-1}} \textbf{1}(D^{t-1} = d^{t-1}) \left( \pi_{0, t, d^{t-1}}^{\intercal}X + \sum_{s = 1}^{t-1} \pi_{s, t, d^{t-1}} Y_s + \pi_{t, t, d^{t-1}} X_k^* \right).
\end{multline*}


\paragraph{Model fit}

Table \ref{tab:extended_model_fit} below reports the model fit based on the same moments as in Table \ref{tab:model_fit}. Overall, the fit is nearly identical to the baseline specification. No estimated moment varies by more than $.01$ from the baseline fit, and most are exactly the same. While including college in the model reveals patterns of sorting by education level, it actually does not appear to meaningfully affect the model fit.  
\begin{table}
    \centering
    \begin{tabular}{llrrrrrr}
    \toprule
        & & \multicolumn{2}{c}{$Y_1$}
        & \multicolumn{2}{c}{$Y_2$}
        & \multicolumn{2}{c}{$Y_3$} \\
        & & Est. & Data
        & Est. & Data
        & Est. & Data \\    
    \midrule    
        \multicolumn{8}{l}{A. No period in high-skill occupation} \\
        \addlinespace[2ex]
        \multicolumn{8}{l}{\textit{Mean}} \\
        \addlinespace[1ex]
        & & 2.45 
        & 2.45 
        & 2.50 
        & 2.52 
        & 2.56 
        & 2.57 \\
        \addlinespace[1ex]            
        \multicolumn{8}{l}{\textit{Covariance Matrix}} \\
        \addlinespace[1ex]        
        & $Y_1$
        & 0.18 
        & 0.17 
        & 0.14 
        & 0.14 
        & 0.13 
        & 0.13        
        \\
        & $Y_2$
        & \textemdash
        & \textemdash        
        & 0.18 
        & 0.19 
        & 0.17 
        & 0.17        
        \\
        & $Y_3$
        & \textemdash
        & \textemdash   
        & \textemdash
        & \textemdash               
        & 0.22 
        & 0.21        
        \\        
    \midrule    
        \multicolumn{8}{l}{B. Some periods in high-skill occupation} \\
        \addlinespace[2ex]
        \multicolumn{8}{l}{\textit{Mean}} \\
        \addlinespace[1ex]
        && 2.57 
        & 2.58 
        & 2.67 
        & 2.68 
        & 2.83 
        & 2.80 \\
        \addlinespace[1ex]            
        \multicolumn{8}{l}{\textit{Covariance Matrix}} \\
        \addlinespace[1ex]        
        & $Y_1$
        & 0.18 
        & 0.21 
        & 0.12 
        & 0.14 
        & 0.13 
        & 0.12        
        \\
        & $Y_2$
        & \textemdash
        & \textemdash        
        & 0.18 
        & 0.20 
        & 0.15 
        & 0.13        
        \\
        & $Y_3$
        & \textemdash
        & \textemdash   
        & \textemdash
        & \textemdash               
        & 0.23 
        & 0.19        
        \\    
    \midrule    
        \multicolumn{8}{l}{C. All periods in high-skill occupation} \\
        \addlinespace[2ex]
        \multicolumn{8}{l}{\textit{Mean}} \\
        \addlinespace[1ex]
        && 2.78 
        & 2.76 
        & 2.92 
        & 2.91 
        & 3.01 
        & 3.00 \\
        \addlinespace[1ex]            
        \multicolumn{8}{l}{\textit{Covariance Matrix}} \\
        \addlinespace[1ex]        
        & $Y_1$
        & 0.23 
        & 0.26 
        & 0.16 
        & 0.16 
        & 0.16 
        & 0.17        
        \\
        & $Y_2$
        & \textemdash
        & \textemdash        
        & 0.23 
        & 0.21 
        & 0.16 
        & 0.19        
        \\
        & $Y_3$
        & \textemdash
        & \textemdash   
        & \textemdash
        & \textemdash               
        & 0.25 
        & 0.26        
        \\         
    \bottomrule
  \end{tabular}
    \caption{Extended Model: Model Fit}
    \label{tab:extended_model_fit}
\end{table}

\FloatBarrier

\paragraph{Selection patterns}

In the baseline specification, we noted that there was a strong pattern of selection into the high-skill occupation based on the known heterogeneity component $X^*_k$. As shown in Figure~\ref{fig:extended_xk_selection_exp} below, this pattern is closely replicated in this model with college education as an additional variable.\medskip

\begin{figure}[H]
    \centering
    \input{paper-inputs/6-23-2025/xk_selection_high_skill_work_run_27-updated.pgf}
    \caption{Extended Model: Selection into High-Skill Occupation}
    \label{fig:extended_xk_selection_exp}
\end{figure}

A portion of this selection is explained by selection into college graduation. On that note, Figure \ref{fig:extended_xk_selection_college} below reports the estimated conditional CDFs $F_{X^*_k \mid X^c = 0}$ and $F_{X^*_k \mid X^c = 1}$. This figure shows that there is a mass point in both the college and non-college graduate sub-populations at the low-skill level, but that approximately half of the mass among college-graduates is at higher skill levels.  

\begin{figure}[ht]
    \centering
    \input{paper-inputs/6-23-2025/xk_selection_college_run_27-updated.pgf}
    \caption{Extended Model: Selection into High Skill Occupation by Education level}
    \label{fig:extended_xk_selection_college}
\end{figure}

We can further examine the role of selection on college education by considering the probabilities that an individual works in a high-skill occupation, conditional on education and skill level. Table \ref{tab:extended_ccps} shows the probability of working in a high-skill occupation conditional on the percentile of $X^*_k$ and the college graduation status.

\begin{table}
    \centering
\newcommand{\tablecell}[1]{%
  \ifthenelse{\equal{#1}{}}{\textemdash}{%
    \ifthenelse{\equal{#1}{nan}}{\textemdash}{#1}%
  }%
}

\begin{tabular}{lccc ccc}
    \toprule
    & \multicolumn{2}{c}{Share in High-Skill Occupation} & \multicolumn{2}{c}{Share of Population} \\
    \cmidrule(lr){2-3} \cmidrule(lr){4-5}
    $X_k^*$ Group & Low & High & Low & High \\
    \midrule
        Non-College Graduate 
            & \tablecell{0.12}
            & \tablecell{0.65}
     
            & \tablecell{0.58}
            & \tablecell{0.09}
      \\
        College Graduate 
            & \tablecell{0.49}
            & \tablecell{0.97}
     
            & \tablecell{0.17}
            & \tablecell{0.16}
      \\
     
    \bottomrule
\end{tabular}
    \caption{Extended Model: Conditional Choice Probabilities of Working in a High-Skill Occupation. \footnotesize{Note: $X_k^*$ is divided into the a high and low skill group for ease of interpretation. The ``High'' group corresponds to the 75th to the 100th percentile of $X_k^*$.}}
    \label{tab:extended_ccps}
\end{table}

From this table we see that the probability of working in a high-skill occupation increases by $37$ percentage points for low-skill workers who have a college degree, and increases by $32$ percentage points for high-skill individuals. This magnitude is similar to the effect of going from low-skill to high-skill. For non-college graduates, the probability of working in a high-skill occupation jumps $53$ percentage points for high-skill individuals compared to low-skill individuals.\footnote{Note, however, that the high-skill individuals without a college degree make up only $9\%$ of the population.}


\paragraph{Variance decomposition}

Finally, we return to the variance decomposition exercise and reproduce in Table~\ref{tab:extended_variance_decomp} below the analysis in Table \ref{tab:emp-var-decomp} for the baseline model. The qualitative patterns of the variance decomposition are quite similar. In particular, the forecastable share of variance is much smaller in high-skill occupations, and the rate of learning is fast in both occupations. The patterns observed in the baseline model are somewhat accentuated, with the estimate of the initial share of variance forecastable in the high-skill occupation decreasing from $0.12$ to $0.10$, and increasing in the low-skill occupation from $0.43$ to $0.49$.  Given the similarity of these results, we conclude that while college education does play an important role in determining the occupation and wages of workers, the baseline model that absorbs college into $X^*_k$ actually appears to do a good job of capturing the selection and uncertainty faced by individuals in their early career.

\begin{table}[H]
\centering
  \resizebox*{\textwidth}{!}{\begin{tabular}{lcccc}
    \toprule
    & \multicolumn{2}{c}{$\overline{Y}(1)$} & \multicolumn{2}{c}{$\overline{Y}(0)$} \\
    Decomposition & Total Variance & Share Forecastable & Total Variance & Share Forecastable \\ 
    \midrule    
     Equation \eqref{eq:app-vardecomp-1}
        & 0.63 & 0.10 & 1.15 & 0.49 \\   
     Equation \eqref{eq:app-vardecomp-2}, $d_1=0$
        & 0.57 & 0.66 & 0.62 & 0.80 \\   
    Equation \eqref{eq:app-vardecomp-2}, $d_1=1$
        & 0.68 & 0.51 & 1.50 & 0.81 \\   
    \bottomrule
\end{tabular}
}
\caption{Extended Model: Variance Decomposition}
\label{tab:extended_variance_decomp}
\end{table}

\subsubsection{Bootstrap confidence intervals: Empirical coverage}
\label{app:boot}

In this section, we provide evidence, based on Monte Carlo simulations, that the bootstrap confidence intervals used to quantify statistical uncertainty surrounding the variance decomposition parameters yield near-nominal coverage in simulations with the same sample size as in our application.

In order to explore this issue, we calculate the same variance decomposition parameters reported in Table \ref{tab:emp-var-decomp} using the DGP specified for our Monte Carlo simulations in Section \ref{sec:numer-simul}, and calculate 95\% bootstrap confidence intervals. Table \ref{tab:mc_bootstrap_cov} below reports the coverage rate for the 95\% bootstrap confidence intervals. Empirical coverage is quite close to the nominal rate for most parameters, with exact 95\% coverage for several of the parameters (7 out of 12). Empirical coverage remains generally close to the nominal rate in the other cases, although we see some under-coverage (89\%) for one particular case (initial period, share forecastable in low-skill occupations).

\begin{table}[H]
  \resizebox*{\textwidth}{!}{\begin{tabular}{lcccc}
    \toprule
    & \multicolumn{2}{c}{$\overline{Y}(1)$} & \multicolumn{2}{c}{$\overline{Y}(0)$} \\
    Decomposition & Total Variance & Share Forecastable & Total Variance & Share Forecastable \\ 
    \midrule    
     Equation \eqref{eq:app-vardecomp-1}
        & 0.99 & 0.95 & 0.95 & 0.89  \\    
     Equation \eqref{eq:app-vardecomp-2}, $d_1=0$
        & 0.98 & 0.95 & 0.95 & 0.93  \\    
   Equation \eqref{eq:app-vardecomp-2}, $d_1=1$
        & 0.95 & 0.95 & 0.95 & 0.93  \\    
    \bottomrule
\end{tabular}
}
\caption{Variance Decomposition: Bootstrap Confidence Interval Coverage (Monte Carlo Simulations). \footnotesize{Note: Each entry shows the empirical coverage for a nominal coverage of $95\%$. Results were obtained estimating the model for 100 Monte Carlo simulations, calculating 100 bootstrap samples for each simulation. The sample size is 965.}}
\label{tab:mc_bootstrap_cov}
\end{table}

\end{document}